\RequirePackage{tikz}
\documentclass[pdflatex,final]{sn-jnl}% Math and Physical Sciences Reference Style

\usepackage{amsmath}
\usepackage{amsfonts}
\usepackage{thm-restate}
\usepackage{bbm}
\usepackage{xcolor}  % `table` --- to be able to call \cellcolor
\usetikzlibrary{arrows,automata,positioning,matrix,calc,petri,backgrounds,shapes,decorations.pathreplacing}
\usepackage{xparse}    % to be able to overload commands with optional arguments (NewDocumentCommand)
\usepackage[cal=esstix]{mathalfa}  %enables calligraphic small letters
  \newcommand{\mc}{\mathcal}

\usepackage{mathtools} % to define \floor and \ceil

\usepackage{subcaption}
\usepackage{soul}
\usepackage{xspace}
\usepackage{stmaryrd}  % for short arrows
\usepackage{todonotes}%[disable]{todonotes}

\usepackage{hyperref}   % goes last
\hypersetup{
    colorlinks,
    linkcolor=black,
    citecolor=black,
    urlcolor=black,
}
\jyear{2022}%

\DeclarePairedDelimiter\floor{\lfloor}{\rfloor}

\renewcommand\d{{\mathcal d}}
\renewcommand\x\times
\renewcommand|{\vert}
\renewcommand{\v}{\nu}
\newcommand{\VA}{V_\forall}
\newcommand{\VE}{V_\exists}
\newcommand*{\indata}{*}
\newcommand*{\crossSection}[2]{\Lbag #1 \Rbag_{#2}}

\newif\ifshowfixes
\showfixesfalse
\ifshowfixes
  \newcommand{\manufixed}[3][]{\todo[inline,color=green!50!blue!30,#1]{\small M. #2\\{\sc \checkmark} #3}}
\else
  \newcommand{\manufixed}[3][]{}
\fi

\newcommand\D{\mathcal D}
\newcommand{\inp}{\mathbbm{i}}
\newcommand{\outp}{\mathbbm{o}}
\newcommand{\dom}{\mathrm{dom}}

\newcommand{\update}{\mathit{update}}
\newcommand{\qinit}{q_\iota}
\newcommand{\Val}{\D^R}

\newcommand{\tst}{\textnormal{\textsf{tst}}}
\newcommand{\asgn}{\textnormal{\textsf{asgn}}}
\newcommand{\Tst}{\textnormal{\textsf{Tst}}}
\newcommand{\Asgn}{\textnormal{\textsf{Asgn}}}
\newcommand{\inc}{\mathrm{inc}}
\newcommand{\dec}{\mathrm{dec}}
\newcommand{\ifz}{\mathrm{ifz}}
\newcommand{\halt}{\mathrm{halt}}
\newcommand{\invlightning}{\rotatebox[origin=c]{180}{$\lightning$}}
\newcommand{\badState}{\lightning}
\newcommand{\goodState}{\invlightning}

\newcommand{\bbN}{\mathbb{N}}

\newcommand{\bbQ}{\mathbb{Q}}
\newcommand{\bbZ}{\mathbb{Z}}

\newcommand{\impl}{\rightarrow}
\newcommand{\Impl}{\Rightarrow}

\newcommand{\Implied}{\Leftarrow}
\newcommand\LAND\bigwedge
\newcommand\LOR\bigvee

\newcommand{\da}{\!\downarrow\!}

\newcommand{\tup}[1]{(#1)}
\newcommand\li{\begin{itemize}}
\newcommand\il{\end{itemize}}
\renewcommand{\-}{\item}
\newcommand\lo{\begin{enumerate}}
\newcommand\ol{\end{enumerate}}

\newcommand{\parit}[1]{\smallskip\noindent {\it #1.}}

\newcommand\tr{\triangleright}

\newcommand{\otherPlayer}[1]{\overline{#1}}

\newcommand\Bound{\text{\sc b}}

\newcommand{\idle}{\mathrm{idle}}

\newcommand\pto{\rightharpoonup}

\newcommand{\length}[1]{\left\lVert #1 \right\rVert}
\newcommand{\size}[1]{\left| #1 \right|}

\newcommand{\QFeasible}{\textsf{QFeasible}}
\newcommand{\Feasible}{\textsf{Feasible}}
\renewcommand{\:}{\!:}

\newcommand\condAT{\ensuremath{\sf A\frak{2}}}
\newcommand\condBT{\ensuremath{\sf B\frak{2}}}
\newcommand\condAO{\ensuremath{\sf A\frak{1}}}
\newcommand\condBO{\ensuremath{\sf B\frak{1}}}

\theoremstyle{thmstyleone}%
\newtheorem{theorem}{Theorem}%  meant for continuous numbers
\newtheorem{proposition}[theorem]{Proposition}%
\newtheorem{lemma}[theorem]{Lemma}%
\newtheorem{claim}[theorem]{Claim}%

\theoremstyle{thmstyletwo}%
\newtheorem{example}{Example}%
\newtheorem{remark}{Remark}%
\newtheorem*{remark*}{Remark}%

\theoremstyle{thmstylethree}%
\newtheorem{definition}{Definition}%

\begin{document}

\title[Church Synthesis on Register Automata with a Linear Order]{Church Synthesis on Register Automata over Linearly Ordered Data Domains\footnote{This article is an extended version of~\cite{DBLP:conf/stacs/ExibardFK21}, which features full proofs and incorporates elements of~\cite[Chapter~7]{ExibardThesis}.}}

\author[1]{\fnm{L\'{e}o} \sur{Exibard}}%\email{leoe@ru.is}
\author[2]{\fnm{Emmanuel} \sur{Filiot}}%\email{efiliot@ulb.ac.be}
\author[2]{\fnm{Ayrat} \sur{Khalimov}}%\email{ayrat.khalimov@gmail.com}

\affil[1]{\orgname{Reykjavik University}, \orgaddress{\country{Iceland}}}
\affil[2]{\orgname{Universit\'{e} libre de Bruxelles}, \orgaddress{\country{Belgium}}}

\abstract{In a Church synthesis game, two players, Adam and Eve,
  alternately pick some element in a \emph{finite} alphabet, for an
  infinite number of rounds. The game is won by Eve if the
  $\omega$-word formed by this infinite interaction belongs to a given
  language $S$, called the specification. It is well-known that for
  $\omega$-regular specifications, it is decidable whether Eve has a strategy to
  enforce the specification no matter what Adam does.
  We study the extension of Church synthesis games to the linearly ordered data
  domains $(\bbQ,\leq)$ and $(\bbN,\leq)$. In this setting, the infinite interaction
  between Adam and Eve results in an $\omega$-\emph{data} word, i.e.,
  an infinite sequence of elements in the domain.

  We study this problem when specifications are given as register
  automata. Those automata consist in finite automata equipped with a finite set of registers
  in which they can store data values, that they can then compare with
  incoming data values with respect to the linear order. Church games over $(\bbN,\leq)$ are however undecidable, even for
  deterministic register automata. Thus, we introduce one-sided Church games,
  where Eve instead operates over a finite alphabet, while Adam
  still manipulates data. We show that they are determined, and that deciding
  the existence of a winning strategy is in \textsc{ExpTime}, both for
  $\bbQ$ and $\bbN$. This follows from a study of constraint sequences, which
  abstract the behaviour of register automata, and allow us to reduce
  Church games to $\omega$-regular games. We present an application of
  one-sided Church games to a transducer synthesis
  problem. In this application, a transducer models a reactive system
  (Eve) which outputs data stored in its registers, depending on its
  interaction with an environment (Adam) which inputs data to the
  system.
}

\manufixed{R2: Already in the abstract it is unclear why we restrict Eve instead of Adam because it is not yet clear that Eve represents the system. This is only mentioned in the introduction.}{}

\manufixed{R2: - Your recorded STACS'21 talk starts with the higher level picture and introduces new technical concepts only when needed. This makes the talk much easier to follow than this paper where you start with technical results about chains when the reader did not yet fully understand how these results are connected to the core problem or why the results are needed. Before giving the necessary technical details, more motivation and connection to the big picture of these details is needed in the written presentation to guide the reader.
- To achieve the goal of guiding the reader, a good introduction is very important. Your current introduction is closer to preliminaries as you present the basics of church synthesis, register automata, and constraint sequences. Instead, you should move the technical foundations to a proper, currently missing preliminaries section and focus on the big picture and motivation in your introduction. By line 50 on page 5, the reader needs to know why this topic is interesting and important in the first place (give context and motivation for your topic), what it is you are trying to achieve (explain the gap and what the open problem is that you solve), and get a high level idea of your innovative solution (describe the important milestones along the way) so the readers are able to follow your subsequent presentation.
- Smaller, more digestible chapters usually also help the reader to keep track of the bigger picture, especially when recapping the context, motivation, and intuition of each section - which you have missed e.g. at the beginning of Sections 2.2 / 2.2.1 and at the end of Section 3.2.
- The paper is not self-contained as it e.g. relies on and uses automata classes and operations (e.g. page 14 in $A_{neg inf}$), and clique theory (page 11, line 54) without properly introducing them or even citing a reference for the reader to find the missing foundations. Consider adding the needed foundations to a preliminaries section or an appendix.
}{}

\manufixed{R2: - While you motivate your work through its application to transducer synthesis the latter deserves some more motivation of its own. From reading the paper it does not become super clear why the reader should care about the results you are presenting and some more discussion on the connection to real problems would help in appreciating your contribution.
- I would appreciate a more in depth discussion of related work and the current state of the art.
}{}


\pacs[2012 ACM Subject Classification]{\begin{tabular}{lll}Theory of computation & \textrightarrow{} & Logic and verification \\ Theory of computation & \textrightarrow{} & Automata over infinite objects \\ Theory of computation & \textrightarrow{} & Transducers \end{tabular}}


\maketitle

\manufixed{General comment: we could simplify decidability by invoking a
  paper we submitted with Edwin, showing that max-counter games can be
  solved in exptime. So, from Theorem 6, we immediately get the
  result. This would grealy simplify the paper, but we would have to
  invoke a conference paper.}{maybe, but 1) be careful about the exact conditions and 2) data-assignment function is still needed. Therefore, doesn't look like ``greatly simplifies.}

\section{Introduction}\label{sec1}

\manufixed{R2 asks for a proper intro and suggest to move formal details in
  a preliminary section. I agree, it should be more to the point,
  first give some motivation, explain our objective, and give our
  contributions. We could keep the example and the text about
  constraint sequences, we it should be moved at the end of the
  introduction in a ``proof techniques'' section, which would also
  mentioned interesting side results.}{}

\subparagraph{Church synthesis}
Reactive synthesis is the problem of automatically constructing a
reactive system from a specification of correct executions, i.e. a
non-terminating system which interacts with an environment, and whose
executions all comply with the specification, no matter how the
environment behaves. The earliest formulation of synthesis dates back
to Church, who proposed to formalize it as a game problem:
two players, Adam in the role of the environment and Eve in the role of the system,
alternately pick the elements from two finite alphabets $I$ and $O$ respectively.
Adam starts with $i_0 \in I$, Eve responds with $o_0 \in O$, ad infinitum.
Their interaction results in the $\omega$-word $w = i_0 o_0 i_1 o_1... \in (I\cdot O)^\omega$.
The winner is decided by a winning condition,
represented as a language $S \subseteq (I \cdot O)^\omega$ called \emph{specification}:
if $w\in S$, the play is won by Eve, otherwise by Adam.
Eve wins the game if she has a strategy $\lambda_\exists: I^+ \to O$ to pick
elements in $O$,
depending on what has been played so far, so that no matter the input
sequence $i_0i_1\dots$ chosen by Adam, the resulting $\omega$-word
$i_0\lambda(i_0)i_1\lambda(i_0i_1)\dots$ belongs to $S$.
Similarly, Adam wins the game if he has a strategy $\lambda_\forall : O^* \to I$ to win against any strategy Eve uses.
In the original Church problem, specifications are $\omega$-regular
languages, i.e. languages definable in monadic second-order logic with
one successor or
equivalently, deterministic parity automata.
The seminal papers~\cite{BL69,Rab72} have shown that Church games (for
$\omega$-regular specification) are \emph{determined}: either Eve wins
or otherwise Adam wins. Moreover, given a Church game, the winner of
the game is computable. Finally, justifying the use of Church games as a formulation of reactive synthesis,
finite-memory strategies are sufficient to win (both for Eve and
Adam). This implies that if Eve wins a Church game, one can
effectively construct a finite-state machine (e.g. a Mealy machine)
implementing a winning strategy.

Church synthesis and games on graphs have been extensively studied
for specifications given in linear-time temporal logic
(LTL)~\cite{PR89a} --~recently supported by a tool competition~\cite{Syn14}~--, as well
as in many other settings, for example, quantitative, distributed,
non-competitive~(see
\cite{DBLP:reference/mc/BloemCJ18,DBLP:journals/siglog/Bruyere21} and
the references therein). Yet, those works focus on control, sometimes
with complex interactions between the synthesized systems, rather than
on data. This is reflected already in the original formulation by
Church: Adam and Eve interact via \emph{finite} alphabets $I$ and
$O$, intended to model control actions rather than proper pieces of
data. But real-life systems often operate values from a large to \emph{infinite} data domain.
Examples include data-independent programs~\cite{Wol86,HDB97,LN00},
software with integer parameters~\cite{BHM03},
communication protocols with message parameters~\cite{DST13},
and more~\cite{BHJS07,Via09,CFBBCM02}. The goal of this paper is to
study extensions of reactive synthesis, and its formulation as Church
games, to infinite data domains: $(\bbQ,\leq)$ and $(\bbN,\leq)$ in
particular.

\subparagraph{Church synthesis over infinite data domains} Church
games naturally extend to an infinite data domain $\D$: Adam and Eve
alternately pick data in $\D$, and their infinite interaction results in
an $\omega$-\emph{data} word $d_0d'_0d_1d'_1\dots \in \D^\omega$. The
game is won by Eve if it belongs to a given specification $S\subseteq
\D^\omega$. Accordingly, strategies for Eve have type $\D^+\rightarrow
\D$, while strategies for Adam have type $\D^*\rightarrow \D$. In this
paper, we study specifications given by a standard extension of
finite-state automata to infinite data domains called
\emph{register automata}~\cite{KF94}: they use a finite set of
registers to store data values, and a finite set of predicates over the data domain
to test those values. In each step, the automaton reads a data value from $\D$ and
compares it with the values held in its registers using the predicates
(and possibly constants).
Depending on this comparison, it decides to store the value in some of the registers,
and then moves to a successor state. This way, it builds a sequence of
\emph{configurations} (pairs of state and register values)
representing its run on reading a data word from $\D^\omega$: it is
accepted if the visited states satisfy a certain parity condition. In
this paper, we study specifications given by deterministic register
automata over $\bbQ$ or $\bbN$, which can use the predicate $\leq$ and
the constant $0$ to test data values.

\subparagraph{Contributions} Our first result is an impossibility
result: deciding the winner of a Church game for specifications given
by deterministic register automata over $(\bbN,\leq)$ is an
undecidable problem (Theorem~\ref{thm:Church_2sided_undec}). We
introduce the \emph{one-sided restriction} on Church games: Adam still
has the full power of picking data values, but Eve's behaviour is restricted
to picking elements from a \emph{finite} alphabet only. Despite being
asymmetric, one-sided Church games are quite expressive. For example,
they model synthesis scenarios for runtime data monitors
that monitor the input data stream and raise a Boolean flag when a
critical trend happens (like oscillations above a certain
amplitude), and for systems that need to take control actions
depending on sensor measurements (a heating controller for instance).
Formally, in one-sided Church games, there is a finite set of elements
$\Sigma$ in which Eve picks her successive choices. Accordingly,
specifications are languages $S\subseteq (\D\Sigma)^\omega$, in this
paper defined by deterministic one-sided register automata (defined naturally by
alternating between register automata transitions and finite-state
automata transitions). Eve's strategies have type
$\lambda_\exists:\D^+\rightarrow \Sigma$ while Adam's strategies have type
$\lambda_\forall:\Sigma^*\rightarrow \D$. We prove the following about
one-sided Church games whose specifications are given by one-sided
    deterministic register automata over $(\bbQ,\leq)$ and
    $(\bbN,\leq)$:
\begin{enumerate}
  \item they are determined: every game is either won by Eve or Adam
  \item they are decidable:  the winner can be computed in time exponential in the number of registers of the specification,
  \item if Eve wins, then she has a winning strategy which can be
    implemented by a transducer with registers (which can be
    effectively constructed). 
\end{enumerate}
Transducers with registers extend Mealy machines with a finite set of
registers: they have finitely many states, and given any state and a
test over the input data value, deterministically, they assign the current value to
some registers (or none), output an element of $\Sigma$, and update their
state. Therefore, the last result echoes the similar result in the $\omega$-regular
setting (finite-memory strategies can be effectively constructed for
the winner), and supports the fact that one-sided Church games on
register automata are an adequate framework for effective synthesis of
machines processing streams of data.

% % Here, Church games are undecidable already for specifications given as deterministic register automata,
% % because by leveraging the antagonism between the two players, they can simulate two-counter machines
% (Theorem~\ref{thm:Church_2sided_undec}).
% For example, to simulate an increment of a counter,
% whose value is currently kept in a register $c$,
% % the automaton asks Adam to provide a data $\d$ above the value $\v(c)$ of the counter,
% saves it into a register $c_{new}$,
% and asks Eve to provide a value strictly between $\v(c)$ and $\v(c_{new})$.
% % If Eve can do this, then Adam cheated and Eve wins, otherwise the game continues.
% Adam wins if eventually the halting state is reached.

\manufixed{R1: Some motivation why the asymmetric case is nice should be
  put here, before the example.}{DONE}

\begin{example}
    Figure~\ref{fig:UpWeGo} illustrates a specification given by a
    deterministic one-sided register automaton, alternating between
    square and circle states, depending on whether their outgoing
    transitions read data values or elements in a finite alphabet $\Sigma = \{a,b\}$. It can be
    seen as a game arena where Adam controls the square states while
    Eve controls the circle states. To simplify the presentation, two
    parts of the automaton are not depicted and have been summarised
    as ``Eve wins'' and ``Eve loses'': any run going in the former part
    is non-accepting and any run going in the latter part is accepting
    (this can be modelled by a parity condition). So, Eve's objective
    is to force executions into ``Eve wins'', whatever input data values are
    issued by Adam. There are two registers, $r_M$ and $r_l$.
The test $\top$ (true) means that the transition can be taken
irrespective of the value played,
the test $r_l<\indata<r_M$ means that the value should be between the values of registers $r_l$ and $r_M$,
and the test `else' means the opposite.
The writing $\da r$ means that the value is stored into the register $r$.
At first, Adam provides some data value $\d_M$, serving as a maximal value stored in $r_M$.
Register $r_l$, initially $0$, holds the last data value $\d_l$ played by Adam.
Consider state $C$:
if Adam provides a value outside of the interval $]\d_l,\d_M[$, he loses;
if it is strictly between $\d_l$ and $\d_M$, it is stored into register $r_l$ and the game proceeds to state $D$.
There, Eve can either respond with label $b$ and move to state $E$, or with $a$ to state $C$.
%She cannot do the latter forever if she wants to win.
In state $E$,
Adam wins if he can provide a data value strictly between $\d_l$ and $\d_M$,
otherwise he loses.
Eve wins this game in $\bbN$:
for example, she could always respond with label $a$,
looping in states $C$--$D$.
After a finite number of steps,
Adam is forced to provide a data value $\geq\d_M$, losing the game.
An alternative Eve winning strategy, that does depend on Adam data,
is to loop in $C$--$D$ until $\d_M-\d_l=1$
(thus, she has to memorise the first Adam value $\d_M$),
then move to state $E$, where Adam loses.
In the dense domain $(\bbQ,\leq)$, however, the game is won by Adam,
because he can always provide a value within $]\d_l,\d_M[$ for any $\d_l<\d_M$,
so the game either loops in $C$--$D$ forever or reaches ``Eve loses''.
%\smallskip

\begin{figure}[ht]
  \resizebox{\textwidth}{!}{%
    ~~~~\begin{tikzpicture}[auto, node distance=2.5cm]
      \tikzstyle{every state}=[text=black,font=\scriptsize]
      \tikzstyle{input}=[rectangle,fill=red!30,minimum size=0.7cm,inner sep=0cm]
      \tikzstyle{output}=[fill=green!30,minimum size=0.8cm,inner sep=0cm]
      \tikzstyle{input char}=[text=purple] \tikzstyle{output char}=[text=teal]

      \tikzset{every edge/.append style={font=\small}}

      \node[state, input] (1) {$A$};
      \node[state, right= of 1, output] (2) {$B$};
      \node[state, right= of 2, input] (3) {$C$};
      \node[state, below=1.8cm of 3, output] (4) {$D$};
      \node[state, right= of 4, input] (5) {$E$};
      \node at (3-|5) (6) {Eve wins};
      \node[right= of 5] (7) {Eve loses};
      \node (8) at ($(3)!0.52!(4)$) {\footnotesize\begin{tabular}{c}Infinite: \\ Eve \\ loses \end{tabular}};

      \path[->, >=stealth'] (1) edge node[above,input char] {$\top / \da r_M$} (2);
      \path[->, >=stealth'] (2) edge node[above,output char] {$a,b$} (3);
      \path[->, >=stealth'] (3) edge[bend right=35] node[left,input char] {$r_l < \indata < r_M / \da r_l$} (4);
      \path[->, >=stealth'] (3) edge node[above,input char] {$else$} (6);
      \path[->, >=stealth'] (4) edge[bend right=35] node[right,output char] {$a$} (3);
      \path[->, >=stealth'] (4) edge node[above,output char] {$b$} (5);
      \path[->, >=stealth'] (5) edge node[right,input char] {$else$} (6);
      \path[->, >=stealth'] (5) edge node[above,input char] {$r_l < \indata < r_M$} (7);
      % \draw [decorate, decoration = {brace,raise=10em}] (4) --  (3);
    \end{tikzpicture}~~~~
    }
    \caption{Eve wins this game in $\bbN$ but loses in $\bbQ$.}
    \label{fig:UpWeGo}
\end{figure}

\end{example}

% % In this paper, we apply them to the resolution of the register transducer synthesis problem
% % they allow for synthesis of register transducers
% % % which can output data present in one of the registers of the specification automaton
% % (also studied in~\cite{DBLP:journals/lmcs/ExibardFR21}), and it is our main motivation for introducing them.
% % % Register-transducer synthesis serves as our main motivation for studying Church games.

\subparagraph{Proof overview} We give intuitions about the main
ingredients to show decidability. The key idea used to solve problems about register automata is
to forget the precise values of input data and registers, and track instead the constraints (sometimes called types) describing the relations between them.
In our example,
all registers start in $0$ so the initial constraint is $r^1_l = r^1_M$,
where $r^{i}$ abstracts the value of register $r$ at step $i$.
Then, if Adam provides a data above the value of $r_l$,
the constraint becomes $r^2_l<r^2_M$ in state $B$.
Otherwise, if Adam had provided a data equal to the value in $r_l$, the constraint would be $r^2_l = r^2_M$.
In this way the constraints evolve during the play, forming an infinite sequence.
Looping in states $C$--$D$ induces the constraint sequence
$\big(r_l^{i} \!\!<\! r_l^{i+1} \!\!<\! r_M^{i} \!=\! r_M^{i+1}\big)_{i>2}$.
It forms an infinite chain $r_l^3 < r_l^4 < ...$ bounded by constant $r^3_M=r^4_M=...$ from above.
In $\bbN$, as it is a well-founded order, it is not possible to assign values to the registers at every step to satisfy all constraints,
so the sequence is not satisfiable.
Before elaborating on how this information can be used to solve Church games,
we describe our results on satisfiability of constraint sequences.
This topic was inspired by the work~\cite{ST11} which studies, among others, the nonemptiness problem of constraint automata,
whose states and transitions are described by constraints.
In particular, they show~\cite[Appendix C]{ST11} that
satisfiability of constraint sequences can be checked by \emph{nondeterministic} $\omega$B-automata~\cite{BC06}.
Nondeterminism however poses a challenge in synthesis,
and it is not known whether games with a winning objective given as a nondeterministic $\omega$B-automaton are decidable.
In contrast, we describe a \emph{deterministic} max-automaton~\cite{B11}
characterising the satisfiable constraint sequences in $\bbN$. As a consequence
of~\cite{DBLP:conf/icalp/Bojanczyk14a}, games over such automata are decidable.
Then we study two kinds of constraint sequences inspired by Church games with register automata.
First, we show that the satisfiable lasso-shaped\footnote{Lasso-shaped words are also called regular words or ultimately periodic words in the literature.} constraint sequences, of the form $uv^\omega$, are recognisable by deterministic \emph{parity} automata.
Second, we show how to assign values to registers on-the-fly in order to satisfy a constraint sequence induced by a play in the Church game.
%Both results are used to solve Church games.

To solve one-sided Church games with a specification given as a register automaton $S$ for $(\bbN,\leq)$ and $(\bbQ,\leq)$,
we reduce them to certain finite-arena zero-sum games, which we call automata games.
The states and transitions of the game are those of the specification automaton $S$.
The winning condition requires Eve to satisfy the original objective of $S$
only on feasible plays, i.e. those that induce satisfiable constraint sequences.
In our example,
the play $A \cdot B \cdot (C \cdot D)^\omega$ does not satisfy the parity condition, yet it
is won by Eve in the automaton game since it is not satisfiable in $\bbN$,
and therefore there is no corresponding play in the Church game.
We show that if Eve wins the automaton game,
then she wins the Church game,
using a strategy that simulates the register automaton $S$ and simply picks one of its transitions.
It is also sufficient: if Adam wins the automaton game then he wins the Church game.
To prove this, we construct, from a winning strategy of Adam in the automaton game,
a winning strategy of Adam (that manipulates \emph{data}) in the Church game.
This step uses the previously mentioned results on satisfiability of constraint sequences. Over $(\bbN, \leq)$, we cannot solve the automaton game directly, as it is not $\omega$-regular. We instead reduce it to an $\omega$-regular approximation of it which considers \emph{quasi}-feasible sequences, a notion which is more liberal than feasibility but coincides with it on lasso-shaped words.
% % Overall, our results on one-sided Church games in $(\bbN,\leq)$ and $(\bbQ,\leq)$ are:
% \li
% % \- they are decidable in time exponential in the number of registers of the specification,
% \- they are determined: every game is either won by Eve or by Adam, and
% \- if Eve wins,
   % % % then she has a winning strategy that can be described by a register transducer with a finite number of states and which picks transitions in the specification automaton.
% \il
% % Finally, these results allow us to solve the register-transducer synthesis problem from
% % input-driven output specifications~\cite{DBLP:journals/lmcs/ExibardFR21} over ordered data.

\subparagraph{Related works} %\todo{mention Thomas' paper on bound-guess actions}

%%%%%%%% TO ADAPT %%%%%%%%%%%%%%%
This paper is an extended version of the conference
paper~\cite{DBLP:conf/stacs/ExibardFK21}. It follows a line of works
about synthesis from register automata specifications~\cite{ESK14,KMB18,KK19,DBLP:journals/lmcs/ExibardFR21}, which  focused on register
automata over data domains $(\D,=)$ equipped with
\emph{equality} tests only. The synthesis of data systems has also
been investigated in~\cite{FKPS19,KMMMV20}. They do not rely
on register automata and are also limited to equality tests or do not
study data comparison. Thus, systems that output the largest value seen so far,
grant a resource to a process with the lowest \textsc{id},
or raise an alert when a heart sensor reads values forming a dangerous
curve, are out of reach of those synthesis methods. These systems
require $\leq$.

In this paper, we consider specifications given by
\emph{deterministic} register automata. Already in the case of
infinite alphabets $(\D,=)$, dropping the determinism requirement
leads to undecidability: finding a winner of a Church game is
undecidable when specifications are given as nondeterministic or
universal register
automata~\cite{ESK14,DBLP:journals/lmcs/ExibardFR21}. To recover
decidability, in the case of universal register automata, those works
restrict Eve strategies to register transducers with an a priori fixed
number of registers. This problem is called \emph{register-bounded
  synthesis}. Recently in~\cite{DBLP:conf/icalp/ExibardF022},
register-bounded synthesis have been extended to various data domains
such as $(\bbN,\leq)$, $(\mathbb{Z},\leq)$, or $(\Sigma^*,\preceq)$
where $\Sigma$ is an arbitrary finite alphabet and $\preceq$ is the prefix relation. The results
of~\cite{DBLP:conf/icalp/ExibardF022} are
orthogonal to the results of this paper, although they rely on the study of
constraint sequences we conduct here.

The paper~\cite{FK20} studies synthesis from variable automata with arithmetic. Those automata are incomparable with register automata: on the one hand, they allow addition on top of a dense order predicate, but on the other hand they do not allow updating the content of the registers along the run. Note that they do not consider the case of a discrete order.
The paper~\cite{FK17b} studies strategy synthesis but, again, mainly over a dense domain.
A one-sided setting similar to ours was studied in~\cite{DBLP:journals/lmcs/FigueiraMP20} for Church games whose winning condition is given by formulas of the Logic of Repeating Values (a fragment of LTL with the freeze quantifier~\cite{DL09}), but only for $(\D,=)$.
That work was extended to domain $(\bbZ,\leq)$ in~\cite{BP22}.
There,
the authors show that the realisability problem in one-sided setting on $(\bbZ,\leq)$ for Constraint LTL and its prompt variant are 2EXPTIME-complete.
Deterministic register automata are more expressive than Constraint LTL,
so our work subsumes their decidability result,
yet the lower expressivity of Constraint LTL enables simpler arguments.
We note that our proof ideas ---~abstracting data words by finite-alphabet words and utilising regularity of abstracted words~--- are somewhat similar to those in papers on Constraint LTL~\cite{DD07,BP22}.
The work on automata with atoms \cite{KL19} implies our decidability result for $(\bbQ,\leq)$,
even in the two-sided setting,
but not the complexity result, and it does not apply to $(\bbN,\leq)$.
%Our proof techniques -- finite abstraction of data models and utilising lasso-shaped sequences -- resemble those from the study of satisfiability of Constraint LTL~\cite{DD08}, but considering the synthesis setting and a slightly more expressive specification formalism pose additional challenges.
Our setting in $\bbN$ is loosely related to monotonic games~\cite{ABd03}:
they both forbid infinite descending behaviours,
but the direct conversion is unclear.
Games on infinite arenas induced by
pushdown automata~\cite{Wal00,BSW03,DBLP:conf/csl/AbdullaAHMKT14}
or one-counter systems~\cite{DBLP:conf/fossacs/Serre06,DBLP:conf/lics/GollerMT09}
are orthogonal to our games.

\subparagraph{Outline}
In Section~\ref{sec:prelim}, we introduce preliminary notions.
Section~\ref{sec:synt-games} introduces Church synthesis games along with the main tools and results (with proofs postponed).
Section~\ref{sec:Ncase} presents the postponed proofs for Church synthesis,
relying on results about satisfiability of constraint sequences over $(\bbN,\leq)$ described in Section~\ref{sec:constraints}.

\manufixed{R2: - If your space restrictions allow, please expand on the 'simple' and omitted parts of the proofs and provide some details of the required steps for completeness.}{WON'T fix, but we improved the proofs and the paper, as an implication this should make those omitted details clear}

\section{Preliminaries}\label{sec:prelim}

In this paper, $\mathbb{N} = \{0,1,\dots\}$ is the set of natural numbers
(including 0). We assume some knowledge of $\omega$-regular languages and $\omega$-automata, and refer to e.g.~\cite{Cac02b} for an introduction.

\subparagraph{$\omega$-data words} In this paper, an \emph{ordered data domain},
or simply \emph{data domain}, $\D$ is an infinite
countable set of elements called \emph{data}, linearly ordered by some order
denoted $<$.
We consider two data domains, $\bbN$ and $\bbQ$, with their usual
order. An $\omega$-data word over $\D$ is an infinite sequence
$d_0d_1\dots$ of data in $\D$. We denote by $\D^\omega$ the set of
$\omega$-data words. Similarly, we denote by $\D^*$ the set of \emph{finite}
sequences (possibly empty) of elements in $\D$.

\subparagraph{Registers}
Let $R$ be a finite set of elements called \emph{registers},
intended to contain data values, i.e.\ values in $\D$.
A \emph{register valuation} is a mapping $\v : R \to \D$ (also written
$\v \in \D^R$). For any data $\d\in\D$, we write $\d^R$ to denote the
constant valuation $\v_\d(r) = \d$ for all $r\in R$.
%\parbf{Tests, assignments, and $\textit{update}$}\label{def:test-assign-update}

A \emph{test} is a maximally consistent set of atoms of the form
$* \bowtie r$ for $r \in R$ and ${\bowtie} \in \{=,<,>\}$.
We may represent tests as conjunctions of atoms instead of sets.
The symbol `$*$' is used as a placeholder for incoming data.
For example, for $R = \{r_1,r_2\}$, the expression $r_1<*$ is not a test because it
is not maximal, but $(r_1<*)\wedge (*<r_2)$ is a
test.  We denote $\Tst_R$
the set of all tests and just $\Tst$ if $R$ is clear from the context.
A register valuation $\v \in \D^R$ and data $\d \in \D$ \emph{satisfy} a test $\tst \in \Tst$,
written $\tup{\v,\d} \models \tst$,
if all atoms of $\tst$ get satisfied when we replace the placeholder $*$ by $\d$
and every register $r \in R$ by $\v(r)$.
%Given $\v \in \D^R$ and $\d \in \D$, let $\v\sim\d$ denote the unique test $\tst$ such that $\tup{\v,\d} \models \tst$.
An \emph{assignment} is a subset $\asgn\subseteq R$.
Given an assignment $\asgn$, a data $\d\in\D$, and a valuation $\v$,
we define $\update(\v,\d,\asgn)$ to be the valuation $\v'$
s.t.\ $\forall r \in \asgn\: \v'(r) = \d$ and $\forall r \not\in\asgn\: \v'(r) = \v(r)$.

\subparagraph{Register automata}\label{sec:def:automata}
A \emph{specification deterministic register automaton}, or simply \emph{deterministic register automaton} is a tuple
$S = (Q, \qinit, R, \delta, \alpha)$
where
$Q = Q_A \uplus Q_E$ is a set of \emph{states} partitioned into Adam and Eve states, the state $\qinit \in Q_A$ is \emph{initial},
$R$ is a set of \emph{registers},
$\delta = \delta_A\uplus\delta_E$ is a (total and deterministic)
\emph{transition function} where, for $P \in \{A,E\}$, we have, by setting
$\otherPlayer{A} = E$ and $\otherPlayer{E} = A$:
$\delta_P : (Q_P \x \Tst  \to  \Asgn \x Q_{\otherPlayer{P}})$; and
$\alpha : Q \to \{1,...,c\}$ is a \emph{priority function} where $c$ is the \emph{priority index}.

% We now define the notions of run and language semantics of $A$ over a data domain $\D$.
A \emph{configuration} of $A$ is a pair $\tup{q,\v}\in Q\x\D^R$,
describing the state and register content;
the \emph{initial configuration} is $\tup{\qinit,0^R}$.
A \emph{run} of $S$ on a word
$w = \d_0 \d_1 ... \in \D^\omega$
is a sequence of configurations $\rho = (q_0,\v_0) (q_1,\v_1)... \in ((Q_A \times \Val) (Q_E \times \Val))^\omega$
starting in the initial configuration ($(q_0,\v_0) = (\qinit, 0^R)$)
and such that for every $i\geq0$:
by letting $\tst_i$ be a unique test for which $(\v_i,\d_i)\models \tst_i$,
we have $\delta(q_i,\tst_i) = (\asgn_i, q_{i+1})$ for some $\asgn_i$
and
$\v_{i+1} = \update(\v_i, \d_i, \asgn_i)$.
Because the transition function $\delta$ is deterministic and total,
every word induces a unique run in $S$.
The run $\rho$ is \emph{accepting} if
the maximal priority visited infinitely often is even.
% i.e.\ $\max \{\alpha(q_i) \mid q_j=q_i\text{ for infinitely many }i\}$ is even.
A word is \emph{accepted} by $S$ if it induces an accepting run.
The \emph{language} $L(S)$ of $S$ is the set of all words it accepts.

\subparagraph{Interleavings} Specification register automata are meant to recognise
interleavings of inputs (provided by Adam) and output (provided by Eve), hence
the partitioning of states. Often, we need to combine them or conversely tell
them apart. Thus, given two words $u = u_0 u_1 \dots \in \D^\omega$ and $v = v_0
v_1 \dots \in \D^\omega$, we formally define their \emph{interleaving} $u \otimes v = u_0
v_0 u_1 v_1 \dots \in \D^\omega$. We note that given a word $w = w_0 w_1 \dots
\in \D^\omega$, it can be uniquely decomposed into $w = u \otimes v$, where $u =
w_0 w_2 \dots \in \D^\omega$ and $v = w_1 w_3 \dots \in \D^\omega$.

\subparagraph{Games} A \emph{two-player zero-sum game}, or simply a \emph{game}, is a tuple
$G = (\VA,\VE, v_0, E, W)$
where $\VA$ and $\VE$ are disjoint sets of \emph{vertices} controlled by Adam and Eve,
$v_0 \in \VA$ is \emph{initial},
$E \subseteq (\VA\x \VE) \cup (\VE\x \VA)$
is a turn-based \emph{transition relation}, and
$W\subseteq (\VA \cup \VE)^\omega$ is a \emph{winning objective}.
An \emph{Eve strategy} is a mapping
$\lambda_\exists : (\VA\VE)^+ \to \VA$
such that $(v_\exists,\lambda(v_\forall^0v_\exists^0 ... v_\forall^k v_\exists^k))\in E$
for all paths $v_\forall^0v_\exists^0 ... v_\forall^kv_\exists^k$ of $G$
starting in $v_\forall^0 = v_0$ and ending in $v_\exists^k \in \VE$ (where $k \geq 0$). Note that $\lambda_\exists$ only depends on the $\VE$ component, since the $\VA$ part is determined by the $\VE$ part, so we sometimes define it as $\lambda_\exists: \VE^+ \rightarrow \VA$.
Adam strategies are defined similarly, by inverting the roles of $\exists$ and $\forall$. A strategy is \emph{finite-memory} if it can be computed by a finite-state machine, and \emph{positional} if it only depends on the current vertex.
A \emph{play} is a sequence of vertices starting in $v_0$ and satisfying the edge relation $E$.
It is \emph{won} by Eve if it belongs to $W$ (otherwise it is won by Adam).
An infinite play $\pi = v_0v_1\dots$ is \emph{compatible} with an Eve strategy $\lambda$
when
for all $i\geq 0$ s.t.\ $v_i\in \VE$: $v_{i+1} = \lambda(v_0\dots v_i)$.
An Eve strategy is \emph{winning} if all infinite plays compatible with it are winning. A game is \emph{determined} (respectively, \emph{finite-memory determined}, \emph{positionally determined}) if either Adam or Eve has a winning strategy (resp., a finite-memory winning strategy, a positional winning strategy).

A \emph{finite-arena game} is a game whose arena is finite, i.e. where $\VA$ and $\VE$ are finite. Among them, we distinguish \emph{$\omega$-regular games}, where the winning condition is an $\omega$-regular language. In particular, a \emph{parity game} is a game whose winning condition is defined through a \emph{parity function} $\alpha : \VA \uplus \VE \to \{1,...,c\}$, where a play $v_0 v_1 \dots$ is winning for Eve if and only if the maximal priority seen infinitely often is even. It is well-known that $\omega$-regular games are finite-memory determined and reduce to parity games, which are positionally determined and can be solved in $n^c$~\cite{GTW02} (see also~\cite{CJKLS17}), where $n$ is the size of the game and $c$ the priority index.

Note that in register automata, Adam is represented as $A$ and Eve as $E$, while in games he is $\forall$ and she is $\exists$. This is to visually distinguish automata from games.

\section{Church Synthesis Games}
\label{sec:synt-games}

A \emph{Church synthesis game} is given as a tuple $G = (I,O,S)$,
where $I$ is an \emph{input} alphabet, $O$ is an \emph{output} alphabet,
and $S \subseteq (I \cdot O)^\omega$ is a specification. Its semantics is provided by the game $(\{v_0\} \cup O, I, v_0, E, S)$, where $E = ((\{v_0\} \cup O) \times I) \cup (I \times O)$, but we rephrase it to provide a stronger intuition. In particular, it is at first counter-intuitive that Adam owns $O$ vertices, and Eve $I$ vertices; this is because both players choose their move by \emph{targeting} a specific vertex.

Thus, in a Church synthesis game, two players, Adam (the environment, who provides inputs) and Eve (the system,
who controls outputs), interact. Their strategies are respectively represented as
mappings $\lambda_\forall : v_0 \cdot O^* \to I$ (often simply represented as $\lambda_\forall : O^* \to I$ for symmetry) and $\lambda_\exists : I^+ \to O$.
Given $\lambda_\forall$ and $\lambda_\exists$,
the \emph{outcome} $\lambda_\forall \| \lambda_\exists$ is the infinite sequence
$i_0 o_0 i_1 o_1 ...$ such that for all $j\geq0$:
$i_j = \lambda_\forall(o_0 ... o_{j-1})$ and $o_j = \lambda_\exists(i_0 ... i_j)$.
If $\lambda_\forall \| \lambda_\exists \in S$, the outcome is won by Eve, otherwise by Adam.
Eve wins the game if she has a strategy $\lambda_\exists$ such that
for every Adam strategy $\lambda_\forall$,
the outcome $\lambda_\forall \| \lambda_\exists$ is won by Eve.
Solving a synthesis game amounts to finding whether Eve has a winning strategy.
Synthesis games are parameterised by classes of alphabets and specifications.
A game class is \emph{determined} if every game in the class is either won by Eve or by Adam.

The class of synthesis games where $I$ and $O$ are finite and where $S$ is an $\omega$-regular language is known as \emph{Church games};
they are decidable and determined.
They also enjoy the finite-memoriness property:
if Eve wins a game then she can win it with a strategy that is represented
as a finite-state machine~\cite{BL69} (see also
\cite{DBLP:conf/fossacs/Thomas09} for a game-theoretic presentation of those
results).

We study synthesis games where $I = O = \D$ is an ordered data domain
and the specifications are described by deterministic register automata. In the
following, we let $G^{\D}_S = (\D,\D,S)$ be the Church synthesis game with input
and output alphabet $\D$ and specification $S$, and simply write $G_S$ when $\D$ is clear
from the context.

\subsection{Church games on register automata}

We start our study with a negative result, that highlights the difficulty of the
problem: over the data domain $(\bbN, \leq)$, Church games are undecidable.
% , deciding whether a given
% Church game is winning for Eve is undecidable.
%ak: undefined
%even if we restrict Eve's
%strategy to be a register transducer (i.e. a deterministic register automaton
%with outputs).
Indeed, if the two players pick data values, one can simulate a two-counter
machine as follows: one player provides the values of the counters, while the other
checks that no cheating happens on the increments and decrements. This can be done using the
fact that $c' = c+1$ whenever there does not exist any $\d$ such that $c < \d < c'$.
\begin{theorem}
  \label{thm:Church_2sided_undec}
  Deciding the existence of a winning strategy for Eve in a
  Church game whose specification is a deterministic register automaton over $(\bbN,\leq)$ is undecidable.
\end{theorem}

\begin{proof}[Proof idea]
  % Given a 2CM, we construct a deterministic automaton that checks that the
  % transducer correctly outputs counter values: if the transducer cheats, the
  % environment wins, if the automaton does not reach the halting state
  % eventually, the environment wins. Hence the automaton is realisable iff the
  % 2CM halts. Since it halts, the counters in the run are bounded---the realising
  % transducer will store all those values initially in its registers. The
  % increment gadget is in the figure below, where the automaton register $r_1$
  % tracks counter $\tt c_1$, and $r$ is temporary. Copying is easy to emulate
  % within our model.
  We reduce from the halting problem of 2-counter machines, which is
  undecidable~\cite{Min67}. We define a specification with 4 registers $r_1,r_2,
  z$ and $t$. Registers $r_1$ and $r_2$ each store the value of one counter; $z$ stores
  $0$ to conduct zero tests and $t$ is used as a buffer. We now describe how to
  increment $c_1$ (see \autoref{fig:gadget_incc1}); the cases of $c_2$ and of
  decrementing are similar. Eve suggests a value $d > r_1$, which
  is stored in $t$. Then, Adam checks that the increment was done
  correctly: Eve cheated if and only if Adam can provide a
  data $d'$ such that $r_1 < d' < d$. If he cannot, $d$ is stored in $r_1$, thus
  updating the value of the counter. The acceptance condition is then a
  reachability one, asking that a halting instruction is eventually met. Now, if
  $M$ halts, then its run is finite and the values of the counters are bounded
  by some $B$. As a consequence, there exists a strategy of Eve
  which simulates the run by providing the values of the counters along the run.
  Conversely, if $M$ does not halt, then no halting instruction is reachable by
  simulating $M$ correctly, and Adam is able to check that Eve
  does not cheat during its simulation.
\end{proof}
\begin{figure}[h]
\centering
\begin{subfigure}[t]{.47\textwidth}
  \centering
  \resizebox{\textwidth}{!}{
    \begin{tikzpicture}[->, >=stealth', auto, node distance=2.5cm]
      \tikzstyle{every state}=[text=black,font=\scriptsize]
      \tikzstyle{input}=[rectangle,fill=red!30,minimum size=0.8cm,inner sep=0cm]
      \tikzstyle{output}=[fill=green!30,minimum size=0.9cm,inner sep=0cm]
      \tikzstyle{input char}=[text=purple] \tikzstyle{output char}=[text=teal]

      \tikzset{every edge/.append style={font=\small}}

      \node[state, output] (ko) {$k$};
      \node[state, right= of ko, input] (ki) {};
      \node[state,right= of ki, output] (k'o) {$k+1$};
      \node[above right= 0.75cm and 2.5cm of ki] (l) {$\badState$};
      \node[below right= 0.75cm and 2.5cm of ki] (w) {$\goodState$};

      \path (ko) edge node[above,output char] {$* > r_1, \downarrow {} t$} (ki);
      \path (ki) edge node[above,input char,sloped] {$r_1 < * < t$} (l);
      \path (ki) edge node[above,input char] {$*=t, \downarrow{} r_1$} (k'o);
      \path (ki) edge node[above,input char,sloped] {$* \leq r_1 \vee * > t$} (w);
    \end{tikzpicture}
    }
    \caption{Gadget for instruction $\inc_1$.}
    \label{fig:gadget_incc1}
\end{subfigure} \hfill%
\begin{subfigure}[t]{.45\textwidth}
  \centering
  \resizebox{\textwidth}{!}{
    \begin{tikzpicture}[->, >=stealth', auto, node distance=2.5cm]
      \tikzstyle{every state}=[text=black,font=\scriptsize]
      \tikzstyle{input}=[rectangle,fill=red!30,minimum size=0.8cm,inner sep=0cm]
      \tikzstyle{output}=[fill=green!30,minimum size=0.9cm,inner sep=0cm]
      \tikzstyle{input char}=[text=purple] \tikzstyle{output char}=[text=teal]

      \tikzset{every edge/.append style={font=\small}} % does not seem to matter

      \node[state, output] (ko) {$k$};
      \node[state, above right=0.5cm and 2.5cm of ko, input] (kiz) {};
      \node[state, below right=0.5cm and 2.5cm of ko, input] (kin) {};
      \node[state, right= of kiz, output] (k') {$k'$};
      \node[state, right= of kin, output] (k'') {$k''$};

      \path (ko) edge node[above,sloped,output char] {$* = r_1 \wedge * = z$} (kiz);
      \path (ko) edge node[above,sloped,output char] {$* = r_1 \wedge * > z$} (kin);
      \path (kiz) edge node[above,input char] {$\top$} (k');
      \path (kin) edge node[above,input char] {$\top$} (k'');
    \end{tikzpicture}
    }
    \caption{Gadget for instruction $\ifz_1(k',k'')$.}
    \label{fig:gadget_ifzc1}
\end{subfigure}
\caption{Gadgets for 2CM instructions. The instruction number $k$ is stored in the state of the automaton.
        The state $\badState$ (resp.\ $\goodState$) is a rejecting sink (resp.\ accepting sink). Non-depicted transitions go to the sink state that is losing for the player that takes them.}
% \label{fig:test}
\end{figure}

\begin{proof}
  We reduce from the halting problem of deterministic 2-counter machines, which
  is undecidable~\cite{Min67}. Among multiple formalisations of counter
  machines, we pick the following one: a 2-counter machine has two counters
  which contain integers, initially valued 0. It is composed of a finite set of
  instructions $M = (I_1, \dots, I_m)$, each instruction being of the form
  $\inc_j,\dec_j,\ifz_j(k',k'')$ for $j = 1,2$ and $k',k'' \in \{1, \dots, m\}$,
  or $\halt$. The semantics are defined as follows: a configuration of $M$ is a
  triple $(k,c_1,c_2)$, where $1 \leq k \leq m$ and $c_1, c_2 \in \bbN$. The
  transition relation (which is actually a function, as $M$ is deterministic) is
  then, from a configuration $(k,c_1,c_2)$:
  \begin{itemize}
  \item If $I_k = \inc_1$, then the machine increments $c_1$ and jumps to the
    next instruction $I_{k+1}$: $(k,c_1,c_2) \rightarrow (k+1,c_1+1,c_2)$.
    Similarly for $\inc_2$.
  \item If $I_k = \dec_1$ and $c_1 > 0$, then $(k,c_1,c_2) \rightarrow
    (k+1,c_1-1,c_2)$. If $c_1 = 0$, then the computation fails and there is no
    successor configuration. Similarly for $\dec_2$.
  \item If $I_k = \ifz_1(k',k'')$, then $M$ jumps to $k'$ or $k''$ according to
    a zero-test on $c_1$: if $c_1 = 0$, then $(k,c_1,c_2) \rightarrow
    (k',c_1,c_2)$, otherwise $(k,c_1,c_2) \rightarrow
    (k'',c_1,c_2)$. Similarly for $\ifz_2$.
  \end{itemize}
  A run of the machine is then a finite or infinite sequence of successive
  configurations, starting at $(1,0,0)$. We say that $M$ \emph{halts} whenever
  it admits a finite run which ends in a configuration $(k,c_1,c_2)$ such that
  $I_k = \halt$.

  Let $M = (I_1, \dots, I_m)$ be a 2-counter machine.
  We associate to it the following specification deterministic register automaton:
  $S$ has states $Q = Q_A \uplus Q_E$,
  where, for $P \in \{A,E\}$,
  $Q_P = \big(\{0, \dots, m+1\} ~\cup~ (\{0, \dots, m+1\} {\times} \{y,n\}) ~\cup~ \{\badState, \goodState\}\big) \times \{P\}$.
  The letters $y$ and $n$ are used to remember whether an $\ifz$ test evaluated to true or false;
  they are only used by $A$, but we included them in $Q_E$ for symmetry.
  The initial state of $S$ is $(0,A)$.
  The automaton has four registers $r_1, r_2, t, z$.
  The acceptance is defined by the reachability condition $F = \{(\goodState,A)\}$,
  while $\badState$ signals rejecting sink states.
  The transitions of $S$ are defined by the following procedure:
  \begin{itemize}
  \item Initially, there is a transition $(0,A) \xrightarrow{\top} (1,E)$ so
    that the implementation can start the simulation.
  \item Then, for each $k \in \{1, \dots, m\}$:
    \begin{itemize}
    \item If $I_k = \inc_j$ for $j = 1,2$, then we add to the transitions of $S$ the gadget from
      \autoref{fig:gadget_incc1}, i.e. output transition $(k,E) \xrightarrow{* > r_1,
      \downarrow{} t} (k,A)$ and input transitions $(k,A) \xrightarrow{r_1 < *
        < t} (\badState,E)$, $(k,A) \xrightarrow{* = t, \downarrow{} r_1} (k+1,E)$
        and $(k,A) \xrightarrow{* \leq r_1} (\goodState,E)$, $(k,A) \xrightarrow{* >
      t} (\goodState,E)$.
  \item The case $I_k = \dec_j$ for $j=1,2$ is similar: we add output transition
    $(k,E) \xrightarrow{* < r_1, \downarrow{} t} (k,A)$ and input transitions
    $(k,A) \xrightarrow{t < * < r_1} (\badState,E)$, $(k,A) \xrightarrow{* = t,
      \downarrow{} r_1} (k+1,E)$ and $(k,A) \xrightarrow{* \geq r_1}
    (\goodState,E)$, $(k,A) \xrightarrow{* < t} (\goodState,E)$. Note that in our
    definition, if $c_j = 0$, then the instruction $\dec_j$ should be blocking,
    i.e. the computation should fail, which is consistent with the fact that in
    that case, the implementation cannot provide $d < r_1$.
    \item If $I_k = \ifz_j(k',k'')$, then we add the gadget of
      \autoref{fig:gadget_ifzc1}, i.e. output transitions $(k,E) \xrightarrow{*
        = r_1 \wedge * = z} (k,y,A)$, $(k,E) \xrightarrow{*
        = r_1 \wedge * > z} (k,n,A)$ and input transitions $(k,y,A)
      \xrightarrow{\top} (k',E)$ and $(k,n,A)
      \xrightarrow{\top} (k'',E)$.
      \item If $I_k = \halt$, we add a transition $(k,E) \xrightarrow{\top} (\goodState,A)$.
    \end{itemize}
  \item Finally, $(\goodState, P) \xrightarrow{\top} (\goodState,\otherPlayer{P})$ and
    $(\badState, P) \xrightarrow{\top} (\badState,\otherPlayer{P})$ for $P \in
    \{A,E\}$, so that both $\goodState$ and $\badState$ are sink states alternating between the players.
    In the following,
    we sometimes write $\goodState$ for $(\goodState, P)$ and $\badState$ for $(\badState, P)$,
    since the owner of the state does not matter.
  \end{itemize}
  %TODO: be very careful with indices
  Now, assume that $M$ admits an accepting run $\rho = (k_1,c^1_1, c^1_2) \rightarrow
  \dots \rightarrow (k_n,c^n_1, c^n_2)$, where $n \in \bbN$, $k_1 = 1$, $c^1_1 =
  c^1_2 = 0$ and $I_{k_n} = \halt$. The values of the counters are bounded
  by some $B \leq n$.
  % We can then define a register transducer $I$ with $n$ states
  % and $B+1$ registers which ignores its input and outputs
  Then, let $\lambda^\rho$ be the strategy of Eve which ignores the input provided
  by Adam and plays the output $w_\rho = c_0^{j_0} \dots
  c_{n-1}^{j_{n-1}} 0^\omega$, where for $1 \leq l < n$, $j_l$ is the index of
  the counter modified or 
  tested at step $l$ (i.e. $j_l = 1,2$ is such that $I_{k_l} = \inc_{j_l},
  \dec_{j_l}$ of $\ifz_{j_l}(k',k'')$). Formally, for all $u \in \bbN^+$ of length
  $l \geq 0$, we let $\lambda^\rho(u) = c_l^{j_l}$ if $l \leq n-1$ and
  $\lambda^\rho(u) = 0$ otherwise.

  Let us show that $\lambda^\rho$ is a
  winning strategy for Eve. Let $u \in \bbN^\omega$ be an input word provided by
  Adam. We show by induction
  on $l$ that in $S$ the partial run over $(u \otimes w)[{:}2l+1]$ is either in state
  $\goodState$ or $S$ is in configuration $((k_l,E),\tau_l)$, where $\tau_l(r_1) =
  c_l^1$ and $\tau_l(r_2) = c_l^2$.

  Initially, $S$ is in configuration $((0,A),\tau^0_R)$. Then, whatever Adam
  plays, it transitions to $((1,E), \tau^0_R)$, so the invariant holds.
  Now, assume it holds up to step $l$. If $S$ is in $(\goodState,E)$, the only
  available transition is $(\goodState,E) \xrightarrow{\top} (\goodState,A)$,
  and then $(\goodState,A) \xrightarrow{\top} (\goodState,E)$, so the invariant
  holds at step $l+2$ ($\goodState$ is a sink state). Otherwise, necessarily $l
  < n$, $S$ is in configuration $((k_l,E),\tau_l)$ and there are four cases:
  \begin{itemize}
  \item $I_{k_l} = \inc_j$. By definition, $j = j_l$. We treat the case $j = 1$, the
    other case is similar. Then, Eve plays $c_l^1 = c_{l-1}^1 + 1$, which is
    such that $c_l^1 > \tau_l(r_1)$.
    Then, there does not exist $d$ such that $\tau_l(r_1) < d < \tau_l(t)$ since
    $\tau_l(r_1) = c_{l-1}^1$ and $\tau_l(t) = c_{l-1}^1 + 1$, so the play cannot 
    transition to $(\badState,E)$. Now, either Adam plays $u_{l+1} = \tau_l(t)
    = c_{l-1}^1 + 1$, in which case $S$ evolves to configuration $((k_{l+1},E),
    c^1_{l+1},c^2_{l+1})$, and the invariant holds. Otherwise, $u_{l+1} \neq
    \tau_l(t)$ and $S$ goes to $(\goodState,E)$ and the invariant holds as well.
  \item The case of $I_{k_l} = \dec_j$ is similar. Let us just mention that
    the computation does not block at this step, otherwise $\rho$ is not a run
    of $M$, so the transition $d < r_j$ can indeed be taken by Eve.
  \item $I_{k_l} = \ifz_j(k',k'')$. Again, $j = j_l$, and we treat the case $j=1$.
    Eve plays $c_l^1$; there are two cases. If $c_l^1 = 0$, the
    transition $* = r_1 \wedge * = z$ is taken in $S$, since at every step, $\tau_l(z) =
    0$ (this register is never modified). If $c_l^1 \neq 0$, then the transition $*
    = r_1 \wedge * > z$ is taken. In both cases, whatever Adam plays, $S$ then evolves
    to $((k_{l+1},E),\tau_{l+1})$ (where $\tau_{l+1} = \tau_l$) and the invariant
    holds.
  \item Finally, if $I_{k_l} = \halt$, then whatever Eve plays, $S$ transitions
    to $(\goodState,A)$, and whatever Adam plays, the automaton transitions to $(\goodState,E)$.
  \end{itemize}
  As a consequence, $\goodState$ is eventually reached whatever the input, which
  means that for all $u \in \bbN^\omega, u \otimes I(u) \in S$, i.e. $I$ is indeed
  an implementation of $S$.

  Conversely, assume that Eve has a winning strategy $\lambda_\exists$ in $G_S$. Let $\rho$ be the
  maximal run of $M$ (i.e. either $\rho$ ends in a configuration with no
  successor, or it is infinite). It is unique since $M$ is deterministic. Let $n
  = \length{\rho}$, with the convention that $n = \infty$ if $\rho$ is infinite.
  Let us build by induction a play of a strategy\footnote{We only construct the given play, since the rest of the strategy does not matter.} of Adam $\lambda_\forall$ such that for all $l < n, (\lambda_\forall \| \lambda_\exists)[{:}2l] =
  c_l^{j_l}$. and the configuration reached by $S$ over $(\lambda_\forall \otimes \lambda_\exists)[{:}2l]$ is $((k_l,E),\tau_l)$.
  Initially, let $u_0 = 0$. As the initial test is $\top$, $S$ anyway evolves
  to state $(1,E)$, with $\tau(r_1) = \tau(r_2) = 0$.

  Now, assume we built such input $u$ up to $l$. There are again four cases:
  \begin{itemize}
  \item $I_{k_l} = \inc_j$. Then, Eve provides some output data $d_E >
    \tau_l(r_j)$. Assume by contradiction that $d_E > \tau_l(r_j) + 1$. Then,
    $\lambda_\exists$ is not winning because if Adam plays $d_A = \tau_l(r_j) + 1$,
    $S$ goes to state $(\badState,E)$, which is a sink rejecting state, so the
    play is losing irrelevant of what both players play after this move. So, necessarily,
    $d_E = \tau_l(r_j) + 1 = c_l^{j_l}$, and $S$ evolves to configuration
    $(k_{l+1}, \tau_{l+1})$.
  \item The case $I_{k_l} = \dec_j$ is similar. Necessarily, $c_j^l > 0$,
    otherwise Eve cannot provide any output data and the play is losing for Eve,
    which contradicts the fact that $\lambda_\exists$ is winning. Thus, the
    computation does not block here.
  \item $I_{k_l} = \ifz_j(k',k'')$. The output transitions of the gadget
    constrain Eve to output $d_E = \tau_l(r_j) = c_l^{j_l}$, and irrelevant of
    what Adam plays $S$ then evolves to
    configuration $((k_{l+1},E), \tau_{l+1})$.
  \item $I_{k_l} = \halt$. Then, it means that $n < \infty$ and $l=n$, so the
    invariant vacuously holds.
  \end{itemize}
  Now, $\rho$ cannot be infinite, otherwise $\lambda_\forall \| \lambda_\exists$ is not
  accepted by $S$ because $\goodState$ is never reached and Eve would not win. It moreover cannot block on some
  $\dec_j$ instruction, as demonstrated in the induction. Thus, a $\halt$ instruction is eventually reached, which
  means that $\rho$ is a halting run of $M$: $M$ halts.
\end{proof}
%%\begin{remark}
%%  Note that register $z$ can actually be omitted with a slightly smarter
%%  construction. Indeed, for the $\ifz$ test, the implementation can prove that $c_j > 0$
%%  by providing some data $d \leq c_j$, and otherwise claim that $c_j = 0$. If
%%  this is not the cae (i.e. the implementation tried to cheat), the environment
%%  can prove her wrong by providing some data $d \leq c_j$. This shows that even
%%  with $3$ registers, the synthesis problem considered above is undecidable.
%%
%%  Note also that more generally, for any $k$-counter machine $M$, we can define a DRA
%%  specification with $k+1$ registers which is realisable iff $M$ halts, hence
%%  establishing that, as soon as the above synthesis problem is concerned, a
%%  register is at least as expressive as a counter.
%%\end{remark}
%%
%%\begin{remark}
%%  Is the synthesis problem decidable when
%%  \li
%%  \- ...\ all transducer registers are initialised to the same $\d_0$?
%%  \- ...\ automata cannot save an output data value?
%%  \- ...\ the number of registers in the implementation transducer is bounded?
%%  \- ...\ we consider $\tup{\bbQ,<}$ instead of $\tup{\bbN,<}$?
%%  \il
%%\end{remark}

\subsection{Church games on one-sided register automata}

In light of this
undecidability result, we consider one-sided synthesis games, where Adam
provides data but Eve reacts with labels from a \emph{finite} alphabet (a
similar restriction was studied in~\cite{DBLP:journals/lmcs/FigueiraMP20} for
domain $(\D,=)$). Specifications are now given as a language $S \subseteq
(\D\cdot\Sigma)^\omega$, recognised by a one-sided deterministic register automaton.
% Such games are still quite expressive, as they enable
% the synthesis of `relaying' register transducers, which can only output data
% that is present in the specification automaton;\manu{the latter
%   sentence is very unclear} we elaborate on this in
% Section~\ref{sec:transducer_synthesis}.

\begin{definition}
A \emph{one-sided deterministic register automaton}, or simply \emph{one-sided register automaton} $S = (\Sigma, Q, \qinit, R, \delta, \alpha)$ is a deterministic register automaton that additionally has a finite alphabet $\Sigma$ of Eve \emph{labels}. Its states are again partitioned into Adam and Eve states $Q = Q_A \uplus Q_E$, and it has an \emph{initial state} $\qinit \in Q_A$. Its \emph{transition function} $\delta = \delta_A \uplus \delta_E$ is again total, but now has $\delta_E : Q_E \x \Sigma  \to  Q_A$. The rest is defined as for deterministic register automata: $\delta_A : Q_A \x \Tst \to \Asgn \x Q_E$; $R$ is a set of \emph{registers}, and finally $\alpha : Q \to \{1,...,c\}$ is a \emph{priority function} where $c$ is the \emph{priority index}.

The notions of configurations and runs are defined analogously, except for the asymmetry between input and output: a \emph{configuration} of $A$ is a pair $\tup{q,\v}\in Q\x\D^R$,
describing the state and register content;
the \emph{initial configuration} is $\tup{\qinit,0^R}$.
A \emph{run} of $S$ on a word
$w = \d_0 a_0 \d_1 a_1 ... \in (\D \Sigma)^\omega$ (note the interleaving of $\D$ and $\Sigma$)
is a sequence of configurations $\rho = (q_0,\v_0) (p_0, \v_1) (q_1,\v_1) (p_0, \v_2)... \in ((Q_A \times \Val) (Q_E \times \Val))^\omega$
starting in the initial configuration (i.e. $(q_0,v_0) = (\qinit,0^R)$)
and such that for every $i\geq0$:
\begin{itemize}
\item (reading an input data value) by letting $\tst_i$ be a unique test for which $(\v_i,\d_i)\models \tst_i$, we have $\delta(q_i,\tst_i) = (\asgn_i, p_i)$ for some $\asgn_i$ and $\v_{i+1} = \update(\v_i, \d_i, \asgn_i)$, as for deterministic register automata;
\item (reading an output letter from $\Sigma$) $\delta(p_i, a_i) = q_{i+1}$, as for finite-state automata.
\end{itemize}
Again, because the transition function $\delta$ is deterministic and total,
every word induces a unique run in $S$.
The run $\rho$ is \emph{accepting} if
the maximal priority visited infinitely often is even.
% i.e.\ $\max \{\alpha(q_i) \mid q_j=q_i\text{ for infinitely many }i\}$ is even.
A word is \emph{accepted} by $S$ if it induces an accepting run.
The \emph{language} $L(S)$ of $S$ is the set of all words it accepts.

Figure~\ref{fig:UpWeGo} shows an example of a one-sided automaton.
For instance, it rejects the words $3a1b2(\Sigma\D)^\omega$ and accepts the
words $3a1a2b(\D\Sigma)^\omega$.
\end{definition}

The rest of this paper is dedicated to showing that Church games whose specification are defined by one-sided register automata over $(\bbQ, \leq)$ or $(\bbN, \leq)$ are decidable in exponential time, and that those games are determined. Formally,
\begin{theorem}\label{thm:synt-games}
  % For every Church game $G_S$ on a one-sided automaton $S$
  Let $S = (\Sigma, Q, \qinit, R, \delta, \alpha)$ be a one-sided register automaton over $(\bbN,\leq)$ or $(\bbQ,\leq)$.
  \lo[1.~]
  \- The problem of determining the winner of the Church synthesis game $G = (\D,\D,S)$ is decidable in time polynomial in $\size{Q}$ and exponential in $c$ and $\size{R}$. % has a winning strategy in $G$.
  % \- If Eve wins, she can win using finite-memory
  % strategy.\todo{define}
  \- $G_S$ is determined, i.e. either Eve or Adam has a winning strategy in $G_S$.
  \ol
\end{theorem}
The above is a wrapper theorem, that aggregates Theorems~\ref{thm:synt-games-Q} for $(\bbQ, \leq)$ and~\ref{thm:synt-games-N} for $(\bbN, \leq)$. We defer the proof to Section~\ref{sec:Ncase}. The result for $(\bbQ, \leq)$ can be derived from~\cite{DD07} or~\cite[Section~7]{KL19}, but we include it for pedagogical reasons, as it allows us to introduce the main tools in a simple setting and to highlight the difficulties that creep up when we shift to $(\bbN, \leq)$.

In the case of a finite alphabet, the game-theoretic approach to solving Church games whose specification is given by a deterministic finite-state automaton consists in playing \emph{on} the automaton, in the following sense: the arena consists of the automaton, and Adam and Eve alternately choose an input (respectively, output) letter, or equivalently (since the automaton is deterministic) an input (resp., output) transition of the automaton. Then, Eve wins whenever the word they jointly produced is accepted by the automaton.

Here, we follow the same approach, with the additional difficulty that the players manipulate data values from an infinite alphabet. Thus, it is not immediate to relate the data values they choose with the corresponding transitions of the automaton. To that end, we study the link between the automaton game (where players pick transitions in the automaton) and the corresponding Church game. This is done through the key notion of \emph{feasible} action words: a sequence of transition labels is feasible whenever it labels a run over some data word. Adam is then asked to provide feasible action words, otherwise he loses. To show that the automaton game is equivalent with the Church game, it remains to show that a strategy of Adam in the automaton game can be translated to a strategy in the Church game. The key ingredient is to be able to instantiate a given action by a data value on-the-fly, while the play unfolds.

Over $(\bbQ, \leq)$, as we demonstrate, the set of feasible action words is $\omega$-regular, so the automaton game is $\omega$-regular as well. Moreover, from a given configuration, one can locally determine whether an action can be instantiated with a data value, and pick it accordingly, which yields the sought strategy translation. Thus, both games are equivalent, and we get decidability since $\omega$-regular games are decidable. The case of $(\bbN, \leq)$ is much more involved and requires further developments, so we start the presentation with $(\bbQ, \leq)$ to sharpen our tools.

\subsection{The automaton game}
\label{sec:automaton_game}

For the rest of this section, fix a one-sided
register automaton $S = (\Sigma,Q,\qinit,R,\delta,\alpha)$ over an ordered data domain $\D$ (it can be either $(\bbQ,\leq)$ or $(\bbN,\leq)$).

Before introducing the game itself, we define the main technical notion, which relates the syntax and semantics of register automata.
\begin{definition}
  \label{def:action_word}
An \emph{action word} is a sequence
$(\tst_0,\asgn_0)(\tst_1,\asgn_1)...$ from $(\Tst\x\Asgn)^{*,\omega}$.
It is \emph{$\D$-feasible} (or simply \emph{feasible} when $\D$ is clear from the context)
if
there exists a sequence $\v_0\d_0\v_1\d_1\dots$ of register valuations $\v_i$ and data $\d_i$ over $\D$
such that $\v_0 = 0^R$
and for all $i$:
$\v_{i+1} = \update(\v_i,\d_i,\asgn_i)$ and $(\v_i,\d_i) \models \tst_i$.

We denote by $\Feasible_\D(R)$ the set of action words over $R$ feasible in $\D$.
\end{definition}

With the Church game $(\D,\D,S)$,
we associate the following \emph{automaton game}, which is a finite-arena game $G_S^f = (\VA,\VE,v_0,E, W_S^f)$.
Essentially, it memorises the transitions taken by the automaton $S$ during the play of Adam and Eve.
It has
$\VA = \{\qinit\} \cup (\Sigma \x Q_A)$,
$\VE = \Tst\x\Asgn \x Q_E$,
$v_0 = \qinit$,
$E = E_0 \cup E_\forall \cup E_\exists$
where:
\li
\-
$E_0 = \big\{\big(v_0, (\tst,\asgn,u_0)\big) \mid \delta(v_0,\tst)=(\asgn,u_0)\big\}$,
\-
$E_\forall = \big\{ \big((\sigma,v),(\tst,\asgn,u)\big)\mid \delta(v,\tst)=(\asgn,u) \big\}$, and
\-
$E_\exists = \big\{ \big((\tst,\asgn,u),(\sigma,v)\big)\mid \delta(u,\sigma)=v \big\}$.
\il

We let:
% The winning condition $W_f \subseteq (\VA\cup\VE)^\omega$ of $G_f$ is given as the set of words:
$$
W_S^f = \left\{ v_0(\tst_0,\asgn_0,u_0)(\sigma_0,v_1) \ldots\;\middle|\;
  \begin{array}{l}
               (\tst_0\asgn_0) \ldots \in \Feasible_\D(R) \\
               \Impl
    v_0u_0v_1u_1\dots\models \alpha
    \end{array}\right\}
$$%
The strategies of Adam and Eve in the automaton game are of the form
$\lambda_\forall^f : \VA(\VE\VA)^* \to \VE$ and
$\lambda_\exists^f : (\VA\VE)^+ \to \VA$.
Since the automaton $S$ is deterministic, they can equivalently be expressed as
$\lambda_\forall^f : \Sigma^* \to \Tst$ and $\lambda_\exists^f : \Tst^+ \to \Sigma$.

Let us show that $G_S^f$ is a sound abstraction of $G_S$, in the sense that a winning strategy of Eve in $G_S^f$ can be translated to a winning strategy of Eve in $G_S$, for both $(\bbQ, \leq)$ and $(\bbN, \leq)$:
\begin{proposition}
  \label{prop:Gf_GS_sound}
  Let $S$ be a deterministic register automaton. If Eve has a winning strategy in $G_S^f$, then she has a winning strategy in the Church game $G_S$.
\end{proposition}
\begin{proof} The main idea of the proof is that is $G_S$, Eve has more information than in $G_S^f$, since she knows what data values Adam played, while in $G_S^f$ she can only access the corresponding tests.

  Formally, let $\lambda^f_\exists : (\VA\VE)^+ \to \VA$ be a winning Eve strategy in $G_S^f$.
  We construct a winning Eve strategy $\lambda_\exists : \Tst^+ \to \Sigma$ in $G_S$ as follows\footnote{%
    What we really need is a winning Eve strategy of the form $\lambda^\D_\exists : \D^+ \to \Sigma$.
    The strategy $\lambda_\exists : \Tst^+ \to \Sigma$ that we construct encodes $\lambda_\exists^\D$ as follows:
    it has the same set $R$ of registers as the automaton $G_S$,
    and performs the same assignment actions as the automaton.
    Then, on seeing a new data value, the strategy compares it with the register values,
    which induces a test, and passes this test to $\lambda_\exists$.%
  }.
  Fix an arbitrary sequence $\tst_0...\tst_k$;
  we define $\lambda_\exists(\tst_0...\tst_k)$.
  First, for all $0 \leq i \leq k-1$,
  we inductively define $v_0,u_0,v_1,u_1,\dots,v_k\in (Q_A \cup Q_E)$, $\asgn_0,...,\asgn_k$, and $\sigma_1,\dots,\sigma_{k} \in \Sigma$:
  \li
  \- The state $v_0 = \qinit$ is the initial state of $S$.
  \- For all $0\leq i\leq k$,
     define
     $u_i \in Q_E$ and $\asgn_i$ to be such that $(\asgn_i,u_i) = \delta(v_i,\tst_i)$,
     $\sigma_{i+1} = \lambda^f_\exists\big(v_0(\tst_0,\asgn_0,u_0)(\sigma_1,v_1) \dots (\tst_i,\asgn_i,u_i)\big)$, and
     $v_{i+1} = \delta(u_i,\sigma_i)$.
  \il
  We then set $\lambda_\exists(\tst_0...\tst_k) = \sigma_{k+1}$.
  We now show that the constructed Eve strategy $\lambda_\exists$ is winning in $G_S$.
  Consider an arbitrary Adam data strategy $\lambda^\D_\forall$,
  and let $(v_0,\v_0)(u_0,\v_1)(v_1,\v_1)(u_1,\v_2)...$ be an infinite run in $G_S$ on reading the outcome $\lambda^\D_\forall \| \lambda_\exists$;
  it is enough to show that $v_0 u_0 v_1 u_1 ...$ satisfies the parity condition.
  Let $\d_0 \d_1...$ be the sequence of data produced by Adam during the play,
  let $\sigma_0 \sigma_1 ...$ be the labels produced by Eve strategy $\lambda_\exists$, and
  let $\overline{a} = (\tst_0,\asgn_0)(\tst_1,\asgn_1)...$ be the tests and assignments performed by the automaton during the run.
  Then, the sequence
  $v_0 (\tst_0,\asgn_0,u_0) (\sigma_0,v_1) (\tst_1,\asgn_1,u_1) ...$
  constitutes a play in $G_S^f$,
  which is compatible with $\lambda^f_\exists$.
  Moreover, as witnessed by $\v_0\d_0\v_1\d_1...$, the action word $\overline{a}$ is feasible.
  Therefore, since $\lambda^f_\exists$ is winning,
  the sequence $v_0 u_0 v_1 u_1 ...$ satisfies the parity condition.
\end{proof}

The converse direction of the above proposition is in general harder, as it amounts to showing that the information provided by tests is enough. For the case of $(\bbQ, \leq)$, the density of the domain allows to instantiate tests on-the-fly, in a way that does not jeopardise the feasibility of the overall sequence (Section~\ref{sec:synt-games-Q}). The case of $(\bbN,\leq)$ is much harder, and is the subject of most of Section~\ref{sec:Ncase}.

\subsection{Application to transducer synthesis}
\label{sec:transducer_synthesis}

The Church synthesis game
models the reactive synthesis problem:
$S$ is a specification, and a winning strategy in $G$ corresponds to a
reactive program which implements $S$, i.e.\ whose set of behaviours abides by $S$.

In the finite alphabet case, Church synthesis games are $\omega$-regular. Since those games are finite-memory determined, it means that if a specification admits an implementation, then it admits a finite-state one~\cite{BL69}, that can be modelled as a finite-state transducer (i.e., a Mealy machine). In this section, we study at which conditions we can get an analogue of this result for specifications defined by input-driven register automata~\cite{DBLP:journals/lmcs/ExibardFR21}.
Those specifications consist in two-sided automata where the output data values are restricted to be the content of some
register (in other words, the implementation is not allowed to generate data).
Input-driven automata can be simulated by one-sided automata, in that
output registers can be seen as finite labels.
Correspondingly, we target register transducers, which generalise finite-state transducers to data domains in the same way as register automata generalise finite-state automata. We then show that finite-memory strategies in the automaton game induce register transducer implementations. Indeed, a finite-memory strategy corresponds to a sub-automaton of $S$, which picks output
transitions in $S$ with the help of its memory. This sub-automaton can then be interpreted as a
register transducer with $R$ registers. Note that this result is reminiscent of
Proposition~5~in~\cite{DBLP:journals/lmcs/ExibardFR21}.

We now define input-driven register automata, register transducers,
and then define the synthesis problem and show that it is decidable.

\subparagraph{Input-driven register automata}
An input-driven deterministic register automaton is a two-sided register
automaton whose output data are required to be the content of some register.
Formally, it is a tuple $S = (Q, \qinit, R, \delta, \alpha)$
where $Q = Q_A \uplus Q_E$, $\qinit \in Q_A$ and the transition function is
$$
\delta :
(Q_A \x \Tst \to \Asgn \x Q_E) \cup (Q_E \x \Tst_= \to
\Asgn_\emptyset \x Q_A),$$%
where $\Tst_=$ consists of tests which contain at least one atom of the form
$*=r$ for some $r \in R$, i.e.\ the output data value must be equal to some
specification register, and $\Asgn_\emptyset = \{\emptyset\}$ meaning that output data values
are never assigned to any register. This is without loss of generality, given
that the output value has to be equal to the content of some register.

\subparagraph{Correspondence with one-sided register automata}
  To an input-driven register automaton specification, we associate a
  one-sided register automaton by treating output registers as finite
  labels. Formally, let $S = (Q, \qinit, R, \delta, \alpha)$ be an input-driven register automaton. Its \emph{associated one-sided automaton} is $S' =
  (\Tst_=, Q, \qinit, R, \delta', \alpha)$ (note that the finite output alphabet is
  $\Tst_=$). Up to remembering equality relations between registers, we
  can assume that from an output state, all outgoing
  transitions can be taken, independently of the registers' configuration, i.e.
  that from a reachable output configuration $(q_E, \tau)$, for all transitions
  $t=q_E \xrightarrow{\tst_=, \varnothing} q'_A$, there exists $\d$ such that
  $q_E \xrightarrow[t]{\d} q'_A$. This however induces a blowup of
  $Q$ exponential in $\size{R}$.

  The transition function is
  $\delta'_A = \delta_A$, and $\delta'_E(q_E,\tst) = q'_A$ if and only if
  $\delta_E(q_E,\tst) = (\varnothing, q'_A)$. Overall, the size of $S'$ is
  exponential in $\size{R}$ (because of the assumption we made on output
  transitions) and polynomial in $\size{Q}$.

\subparagraph{Register transducers}\label{sec:def:transducers}
A \emph{register transducer} (RT) is a tuple
$T = (Q, \qinit, R, \delta)$,
where
$Q$ is a set of \emph{states} and $\qinit \in Q$ is \emph{initial},
$R$ is a finite set of \emph{registers}.
The \emph{transition function} $\delta$ is a (total) function
$\delta\:Q \x \Tst  \to \Asgn \x R  \x  Q$.

The semantics of $T$ are provided by the associated register automaton $S_T$. It
has states $Q' = (Q_A \cup \{\badState_A\}) \uplus (Q_E \cup \{\badState_E\})$, where $Q_A$ and $Q_E$ are two
disjoint copies of $Q$ and $\badState_A$, $\badState_E$ jointly form a rejecting sink.
It has initial state $\qinit$ and set of registers $R$. Its transition function is
defined as $q_A \xrightarrow[S_T]{\tst, \asgn} q_E \xrightarrow[S_T]{r^=,
  \varnothing} q'_A$ and $q_E \xrightarrow[A_T]{r^{\neq}, \varnothing} \badState_A$ whenever $q \xrightarrow[T]{\tst \mid \asgn, r} q'$, where $q \xrightarrow[T]{\tst \mid \asgn, r} q'$ stands for $\delta(q,\tst)
= (\asgn, r, q')$ (similarly for $A_T$). Additionally, we let $\badState_A \xrightarrow[A_T]{\top,\varnothing} \badState_E \xrightarrow[A_T]{\top,\varnothing} \badState_A$. The priority function is defined as $\alpha
: q \in Q' \mapsto 2$ and $\badState_A, \badState_E \mapsto 1$, i.e. all states but $\badState_A, \badState_E$ are accepting. Then, $T$ recognises the (total)
function $f_T: \d^A_0 \d^A_1 \dots \mapsto \d^E_0 \d^E_1 \dots$ such
that $\d^A_0 \d^E_0 \d^A_1 \d^E_1 \dots \in L(A_T)$. For each input $\omega$-data word, the associated output $\omega$-data word exists since
all states but $\badState_A, \badState_E$ are accepting. It is moreover unique since the output transitions that avoid the sink state are
determined by the input ones, and they only contain equality tests so the
corresponding output data values are unique.

\paragraph{Synthesis for input-driven output specifications}
Given a specification $S$, we say that a function $f$ realises $S$ if they have
the same domain and its graph
is included in $S$, i.e. $\dom(f) = \dom(S)$ and for all input $x \in \dom(S)$,
$(x,f(x)) \in S$. We then say that a register transducer $T$ \emph{realises} the
register automaton specification $S$ if $f_T$ does, i.e. $L(T) \subseteq L(S)$.

The \emph{register transducer synthesis problem} then asks to produce a $T$ that realises $S$
when such $T$ exists, otherwise output ``unrealisable''.
Note that $T$ and $S$ can have different sets of registers.

\begin{proposition}
  \label{prop:RT_implem_win_strat_GS}
  Let $S = (Q, \qinit, R, \delta, \alpha)$ be an input-driven register automaton, and $S'$ its associated one-sided register automaton.
  If $S$ admits a register transducer implementation, then Eve has a winning strategy in the Church game $G_{S'}$ associated with $S'$.
\end{proposition}
\begin{proof}
  Assume that there exists a register transducer $T$ which realises $S$.
  From $T$, we define a strategy $\lambda^T$ in $G$, which simulates $T$ and $S$ in
  parallel. Given a history $\d^\inp_0 \dots \d^\inp_n$, let $\d^\outp_n$ be the data output by $T$.
  As $S$ is deterministic, there exists a unique run over the history $\d^\inp_0
  \d^\outp_0 \dots \d^\inp_n \d^\outp_n$; let $t = q_E \xrightarrow{\tst_=,
    \varnothing} q'_A$ be the transition taken by $S$ on reading $\d^\outp_n$.
  Then, define $\lambda^T(\d^\inp_0 \dots \d^\inp_n) = \tst_=$. Now, for a play in
  $G$ consistent with $\lambda^T$, consider the associated run in $S'$. As $T$
  is an implementation and the sequence of transitions is feasible (as witnessed
  by the data given as input), this run is necessarily accepting, so $\lambda^T$ is
  indeed a winning strategy in $G$.
\end{proof}
\begin{proposition}
  \label{prop:FM_Gf_RT_implem}
  Let $S = (Q, \qinit, R, \delta, \alpha)$ be an input-driven register automaton, and $S'$ its associated one-sided register automaton. If Eve wins $G^f_{S'}$ with a finite-memory strategy, then $S$ admits a register transducer implementation.
\end{proposition}
\begin{proof}
  Let $S = (Q, \qinit, R, \delta, \alpha)$ be an input-driven register automaton, and $S'$ its associated one-sided register automaton.
  Assume that Eve has a finite-memory winning strategy in $G_S^f$ that is computed by a finite-state automaton $M$ with states $P$,
  initial memory $p_0$, transition function $\mu : P \times V_\exists \rightarrow
  P$ and move selection $s: P \rightarrow V_\forall$. Thus, given a history $h = v_0 \dots v_n \in V_\exists^+$, $\lambda_\exists(h) : V_\exists^+ \rightarrow \Tst_=$ is defined as $s(p)$, where $p_0 \xrightarrow[M]{h} p$.
  Then, consider $T = (Q \times P, (\qinit,p_0), R, \delta')$. We define $\delta'$
  as follows: assume the transducer is in state $(q,p)$. Then, the transducer
  receives input satisfying some test $\tst$. In $S$, it corresponds to some
  input transition $\delta(q,\tst) = (\asgn, q')$. The memory is updated to
  $\mu(p,(\tst,\asgn)) = p'$, and $s(p') = \tst_=$. Let $r$ be such that
  $\tst_= \Rightarrow r^=$ (such $r$ necessarily exists by definition of
  $\Tst_=$). Then, we let
  $\delta((q,p), \tst) = (\asgn, r, (q',p'))$. Now, let $w = \d^A_0 \d^A_1
  \dots$ be an input data word, and $T(w) = \d^E_0 \d^E_1 \dots$. By
  construction, the run of $S$ over $w \otimes T(w) = \d^A_0 \d^E_0 \d^A_1 \d^E_1
  \dots$ corresponds to a play consistent with $\lambda_\exists$, so it is accepting
  (since it is feasible, as witnessed by $w \otimes T(w)$). As a consequence, $w
  \otimes T(w) \in L(S)$, which means that $T$ is indeed a register transducer
  implementation of $S$.
\end{proof}

In the proof of Theorem~\ref{thm:Church_2sided_undec}, Eve's strategy consists in outputting a finite data word with $B \geq 0$ distinct data values, and then only zeroes. Thus, it can be implemented with a register transducer with $B$ registers, provided that its registers can be initialised with non-zero data values (in our setting, we assume all registers are initialised to $0$). As a consequence, we get:
  \begin{theorem}\label{thm:synt-undecidable-N}
    For specifications defined by two-sided deterministic register
    automata over data domains $(\bbQ,\leq)$,
    the register transducer synthesis problem is undecidable, provided that registers can be initialised to an arbitrary valuation.
  \end{theorem}
  \begin{remark}
    The decidability status of the synthesis problem for register transducers with a fixed initial valuation $0^R$ is open.
  \end{remark}

\section{Solving Church Synthesis Games on $(\bbN,\leq)$}\label{sec:Ncase}

We now have the main tools in hand to solve Church synthesis games over ordered data domains. As an introduction, before the case of $(\bbN, \leq)$, we apply those tools to $(\bbQ, \leq)$.

\subsection{Warm-up: the case of $(\bbQ,\leq)$}
\label{sec:synt-games-Q}
 First, let us observe that in that case, the automaton game is $\omega$-regular:
\begin{proposition}
  \label{prop:Gf_Q_omega_regular}
  Let $S$ be a one-sided register automaton over $(\bbQ,\leq)$. Then $G_S^f$ is an $\omega$-regular game.
\end{proposition}
\begin{proof}
  Let $S = (\Sigma, Q, \qinit, R, \delta, \alpha)$ be a one-sided register automaton over $(\bbQ,\leq)$, and let $G_S^f = (\VA,\VE,v_0,E, W_S^f)$ be its associated automaton game. $G_S^f$ is a finite-arena game; it remains to show that it is $\omega$-regular, i.e. that $W_S^f$ is $\omega$-regular. Recall that $W_S^f = \big\{ v_0(\tst_0,\asgn_0,u_0)(\sigma_0,v_1) \ldots \mid (\tst_0\asgn_0) \ldots \in \Feasible_\D(R) \Impl v_0u_0v_1u_1\dots\models \alpha \big\}$. By Theorem~\ref{thm:0-satisf-Q} (on page~\pageref{thm:0-satisf-Q}), we know that $\Feasible_\D(R)$ is $\omega$-regular; since $\alpha$ is a parity condition, one can then build an $\omega$-regular automaton recognising $W_S^f$ using standard automata constructions.
\end{proof}

From Proposition~\ref{prop:Gf_GS_sound}, we already know that for all one-sided register automata $S$ (over $(\bbQ, \leq)$ or $(\bbN, \leq)$), $G^f_S$ soundly abstracts $G_S$. We now show the converse for $(\bbQ, \leq)$:
\begin{proposition}
  \label{prop:Gf_GS_complete}
Let $S$ be a one-sided register automaton over $(\bbQ,\leq)$.
If Eve has a winning strategy in $G_S$, then she has a winning strategy in the Church game $G_S^f$.
\end{proposition}
\begin{proof}
We show the result by contraposition. Assume that Eve does not win $G_S^f$.
As $G_S^f$ is $\omega$-regular (Proposition~\ref{prop:Gf_Q_omega_regular}), it is determined, so
Adam has a winning strategy $\lambda^f_\forall : \VA(\VA\VE)^* \to \VE$ in $G_S^f$.
  We construct the winning Adam data strategy $\lambda_\forall^\bbQ$ in $G_S$ step-by-step, by instantiating the tests on-the-fly. When the test is an equality, pick the corresponding data, and when it is of
  the form $r < * < r'$, take some rational number strictly in the
  interval.

  Formally, suppose we are in the middle of a play:
  $\d_0 ... \d_{k-1}$ has been played by Adam $\lambda_\forall^\bbQ$
  and $\sigma_0 ... \sigma_{k-1}$ has been played by Eve;
  both sequences are empty initially.
  We want to know the value $\d_k$ for $\lambda_\forall^\bbQ(\sigma_0 ... \sigma_{k-1})$.
  Let $(v_0,\v_0)(u_0,\v_1)(v_1,\v_1)(u_1,\v_2)...(v_k,\v_k)$ be the current run prefix of the register automaton $G_S$ (initially $(v_0,\v_0)$).
  We construct the corresponding play prefix
  $v_0 (\tst_0,\asgn_0,u_0)(\sigma_0,v_1)(\tst_1,\asgn_1,u_1)(\sigma_1,v_2)...(\sigma_{k-1},v_k)$ of $G_f$
  (initially $v_0$).
  We assume that this play prefix adheres to $\lambda^f_\forall$ (this holds initially).
  We now consult $\lambda^f_\forall$:
  let $(\tst_k,\asgn_k,u_k) = \lambda^f_\forall(\sigma_{k-1},v_k)$.
  Using $\tst_k$ and $\v_k$,
  we construct $\d_k$ as follows.
  \li
  \- If $\tst_k$ contains $*=r$ for some $r \in R$,
     we set $\d_k = \v_k(r)$.

  \- If $\tst_k$ is of the form $r<*$ for all $r \in R$,
     then set $\d_k = \max(\v_k)+1$,
     i.e.\ take the largest value held in the registers plus $1$.

  \- Similarly, if $\tst_k$ is of the form $*<r$ for all $r \in R$,
     then set $\d_k = \min(\v_k) - 1$.

  \- Otherwise, for every $r \in R$,
     the test $\tst_k$ has either $r<*$ or $*<r$.
     We now pick two registers $r,s$
     such that
     the test contains $r<*$ and $*<s$ and no register holds a value between $\v_k(r)$ and $\v_k(s)$.
     Then we set $\d_k = \frac{\v_k(r) + \v_k(s)}{2}$.
  \il
  It is easy to see that $\d_k$ satisfies $\tst_k$, i.e.\ $(\v_k,\d_k) \models \tst_k$.
  Finally, define $\v_{k+1} = \update(\v_k,\d_k,\asgn_k)$.
  Thus, the next configuration of the run in the register automaton is $(u_k,\v_{k+1})$.
  In $G_f$, the play is extended by $(\tst_k,\asgn_k,u_k)$;
  notice that the resulting extended play again adheres to the winning Adam strategy $\lambda^f_\forall$.
  Therefore,
  starting from the empty sequences of Adam data choices and Eve label choices,
  step-by-step we construct the values for $\lambda_\forall^\bbQ$.

  Then, each play consistent with this strategy in $G_S$ corresponds to a
  unique run in $S$, which is also a play in $G_f$. As $\lambda_\forall^f$ is winning,
  such a run is accepting, so $\lambda_\forall$ is winning: Eve does not win $G_S$.
\end{proof}
We are now ready to show:
\begin{theorem}\label{thm:synt-games-Q}
  % For every Church game $G_S$ on a one-sided automaton $S$
  Let $S = (\Sigma, Q, \qinit, R, \delta, \alpha)$ be a one-sided register automaton over ${(\bbQ,\leq)}$.
  \lo[1.~]
  \- \label{itm:synt-games-Q-dec} The problem of determining if Eve wins the Church synthesis game $G = (\D,\D,S)$ is decidable in time polynomial in $\size{Q}$ and exponential in $c$ and $\size{R}$. % has a winning strategy in $G$.
  % \- If Eve wins, she can win using finite-memory
  % strategy.\todo{define}
  \- \label{itm:synt-games-Q-det} $G_S$ is determined, i.e. either Eve or Adam has a winning strategy in $G_S$.
  \ol
\end{theorem}
\begin{proof}[Proof of Theorem~\ref{thm:synt-games-Q}]
  First, by Propositions~\ref{prop:Gf_GS_sound} and~\ref{prop:Gf_GS_complete}, we know that Eve $G_S$ iff she wins $G_S^f$.

  By analysing the constructions of Propositions~\ref{prop:Gf_Q_omega_regular} and Theorem~\ref{thm:0-satisf-Q}, we get that the automaton game $G_S^f$ is of size polynomial in $\size{Q}$ and
  exponential in $\size{R}$, and has a number of priorities linear in $c$, so it
  can be solved in $O((poly(\size{Q})2^{poly(\size{R})})^c)$, which yields item~\ref{itm:synt-games-Q-dec} of the theorem.

  Then, determinacy (item~\ref{itm:synt-games-Q-det}) follows from the determinacy of $G_S^f$, since it is equivalent with $G_S$.
\end{proof}
As a consequence of Propositions~\ref{prop:RT_implem_win_strat_GS} and~\ref{prop:FM_Gf_RT_implem}, we also get:
\begin{proposition}
  Let $S$ be an input-driven register automaton, and $S'$ its associated one-sided register automaton. The following are equivalent:
  \begin{itemize}
  \item Eve has a winning strategy in $G_{S'}$
  \item Eve has a winning strategy in $G_{S'}^f$
  \item Eve has a finite-memory winning strategy in $G_{S'}^f$
  \item $S$ admits a register transducer implementation
  \item $S$ admits an implementation
  \end{itemize}
\end{proposition}
Thus, we have:
\begin{theorem}\label{thm:synt-decidable-Q}
  For specifications defined by deterministic input-driven output register
  automata over data domains $(\bbQ,\leq)$,
  the register transducer synthesis problem is equivalent with the synthesis problem (for arbitrary implementations) and can be solved in time polynomial in
  $\size{Q}$ and exponential in $c$ and $\size{R}$.
\end{theorem}

\begin{remark}
% I use \parit because \begin{remark} generates an unnecessary numbering
% Leo: if FMSD wants to number remarks, let's do at it pleases.
For data domain $(\bbQ, \leq)$, the synthesis problem for specifications defined by
two-sided register automata is also decidable, if the target implementation is
any program, as the Church game again reduces to a parity game: checking
feasibility is still doable using a parity automaton. However, in general,
register transducers might not suffice; e.g.\ the environment can ask the
system to produce an infinite sequence of data values in increasing order. Yet,
it can be shown that implementations can be restricted to simple programs,
which can be modelled by register transducers which have the additional
ability to pick a data between two others, e.g.\ by computing
$\frac{\d_1+\d_2}{2}$: such ability suffices to translate a finite-memory
strategy in the automaton game to an implementation.
\end{remark}

We now shift to the main result of the paper, namely that Church synthesis games are decidable over $(\bbN, \leq)$. We start by providing some results on actions sequences over $(\bbN,\leq)$ that highlight the difficulties and hint at how to overcome them (Section~\ref{sec:AS_N}). We then use those results to define an $\omega$-regular approximation of the automaton game that we show to be sound and complete (Section~\ref{sec:Greg_N}).
% This section is devoted to the proof of Theorem~\ref{thm:synt-games} for $(\bbN,\leq)$.

\subsection{Action sequences over $(\bbN,\leq)$}
\label{sec:AS_N}
\paragraph{Action sequences over $(\bbN,\leq)$ are not $\omega$-regular}
First, contrary to $(\bbQ, \leq)$, one needs a global condition on action sequences to check whether they are feasible. To get an intuition, consider the action sequence $(\top\{r\})((r > *){r})^\omega$, that asks for an initial data value (stored in $r$), and then repeatedly asks to provide smaller and smaller data values. While feasible in $(\bbQ, \leq)$, such a sequence is not feasible in $(\bbN,\leq)$, as it would yield an infinite descending chain in $\bbN$. And, actually, the discreteness of $(\bbN,\leq)$ implies that the set of feasible action sequences is not $\omega$-regular in $(\bbN, \leq)$ (see, e.g., \cite[Corollary~6.5]{DD07} or~\cite[Appendix~C]{ST11}). We provide an example, for self-containedness.
\begin{example}
  \label{ex:AS_N_not_reg}
consider the automaton of Figure~\ref{fig:UpWeGoAgain}, which essentially consists in that of Figure~\ref{fig:UpWeGo} (on page~\pageref{fig:UpWeGo}) where we allow Adam to repeatedly try his luck by taking the transition from $C$ to $B$. Note that the priorities (written above the states) ensure that if he does so, he loses.
\begin{figure}[ht]
  \resizebox{\textwidth}{!}{%
    ~~~~\begin{tikzpicture}[auto, node distance=2.5cm]
      \tikzstyle{every state}=[text=black,font=\scriptsize]
      \tikzstyle{input}=[rectangle,fill=red!30,minimum size=0.7cm,inner sep=0cm]
      \tikzstyle{output}=[fill=green!30,minimum size=0.8cm,inner sep=0cm]
      \tikzstyle{input char}=[text=purple] \tikzstyle{output char}=[text=teal]

      \tikzset{every edge/.append style={font=\small}}

      \node[state, initial, initial text={}, input,label={above:{$2$}}] (1) {$A$};
      \node[state, right= of 1, output,label={above:{$2$}}] (2) {$B$};
      \node[state, right= of 2, input,label={above:{$1$}}] (3) {$C$};
      \node[state, below=1.8cm of 3, output,label={above:{$1$}}] (4) {$D$};
      \node[state, right= of 4, input,label={above left:{$0$}}] (5) {$E$};
      \node[state, output,label={above:{$2$}}] at (3-|5) (6) {$F$};
      \node[state, right=0.75cm of 6, input,label={above:{$2$}}] (6b) {$F'$};
      \node[state, right= of 5, output,label={above:{$1$}}] (7) {$G$};
      \node[state, above=0.75cm of 7, input,label={above:{$1$}}] (7b) {$G'$};
      % \node (8) at ($(3)!0.52!(4)$) {\footnotesize\begin{tabular}{c}Infinite: \\ Eve \\ loses \end{tabular}};

      \path[->, >=stealth'] (1) edge node[above,input char] {$\top / \da r_M$} (2);
      \path[->, >=stealth'] (2) edge[bend left] node[above,output char] {$a,b$} (3);
      \path[->, >=stealth'] (3) edge[bend right=35] node[left,input char] {$r_l < \indata < r_M / \da r_l$} (4);
      \path[->, >=stealth'] (3) edge[bend left] node[above,input char] {$else$} (2);
      \path[->, >=stealth'] (4) edge[bend right=35] node[right,output char] {$a$} (3);
      \path[->, >=stealth'] (4) edge node[above,output char] {$b$} (5);
      \path[->, >=stealth'] (5) edge node[right,input char] {$else$} (6);
      \path[->, >=stealth'] (5) edge node[above,input char] {$r_l < \indata < r_M$} (7);
      \path[->, >=stealth'] (6) edge[bend left] node[above,output char] {$\top$} (6b);
      \path[->, >=stealth'] (6b) edge[bend left] node[above,input char] {$\top$} (6);
      \path[->, >=stealth'] (7) edge[bend left] node[left,output char] {$\top$} (7b);
      \path[->, >=stealth'] (7b) edge[bend left] node[right,input char] {$\top$} (7);
      % \draw [decorate, decoration = {brace,raise=10em}] (4) --  (3);
    \end{tikzpicture}~~~~
    }
    \caption{Eve wins this game in $\bbN$ (but loses in $\bbQ$).}
    \label{fig:UpWeGoAgain}
\end{figure}
Then, consider sequences of states in $A (B C (DC)^*)^\omega$, where Adam initially picks a value, the game transitions to $B$ then $C$, then Adam and Eve loop between $B$ and $C$ for some time, until at some point Adam transitions back to $B$, and so on. To check whether such a sequence actually corresponds to a play, one needs to check that there exists a uniform bound (the content of $r_M$) over the iterations of $DC$. Formally, plays in $A (B C (DC)^*)^\omega$ are of the form $A (B C (DC)^{n_0})(B C (DC)^{n_1}) \dots$ where there exists $\Bound \geq 0$ such that for all $i \geq 0$, $n_i \leq \Bound$. By an elementary pumping argument, one can show that this language is not $\omega$-regular~\cite{BC06}.
\end{example}
This implies that $\Feasible_{\bbN}(R)$ is not $\omega$-regular whenever $\size{R} \geq 2$, and neither is the automaton game.
We thus consider an $\omega$-regular over-approximation of the automaton game, and show that both games are actually equivalent.

\paragraph*{Constraint sequences, consistency and satisfiability}
\label{sec:def:constraints}

To introduce the said approximation, we first require a further study of $\Feasible_{\bbN}(R)$, that we conduct through the notion of constraint sequences. To ease the comparison between $(\bbQ, \leq)$ and $(\bbN, \leq)$, we define them for both domains. Thus, in this section, fix an ordered domain $\D$.
% Recall that $0$ denotes the minimal element of $\bbN$; for $\bbQ$ its choice is irrelevant.

Given a set of registers $R$ (which can also be thought of as variables),
we let $R' = \{r'\mid r \in R\}$ be the set of their \emph{primed} versions.
Given a valuation $\v \in \D^R$, define $\v' \in \D^{R'}$ to be the valuation that maps $\v'(r') = \v(r)$ for every $r \in R$.

\begin{definition}
  \label{def:constraint}
A \emph{constraint} over $R$ is a total non-strict preorder over $R \cup R'$,
i.e.\ a total order with ties allowed.
It can be represented as a maximally consistent set of atoms of the form
$t_1 \bowtie t_2$ where $t_1,t_2 \in R\cup R'$,
where the symbol $\bowtie$ denotes one of $>$, $<$, or $=$.

Given a constraint $C$,
the writing $C_{|R}$ denotes the subset of its atoms $r \bowtie s$ for $r,s \in R$,
and $C_{|R'}$ denotes the subset of atoms over primed registers.
Given a set $S$ of atoms $r' \bowtie s'$ over $r',s' \in R'$,
let $unprime(S)$ be the set of atoms derived by replacing every $r' \in R'$ by $r$.

A \emph{state constraint} relates registers in the current moment only:
it contains atoms over non-primed registers, so it has no atoms over primed registers.
Note that both $C_{|R}$ and $unprime(C_{|R'})$ are state constraints.
\end{definition}

A constraint describes how register values change in one step:
their relative order at the beginning (when $t_1,t_2 \in R$),
at the end (when $t_1,t_2 \in R'$), and in between (with $t_1 \in R$ and $t_2 \in R'$).

\begin{example}
For instance, the ordering $r_1 < r'_1 < r'_2 < r_2$ is a
constraint over $R=\{r_1,r_2\}$ and can be represented by
$\{r_1<r_2, r_1<r'_1, r_2>r'_2, r'_1<r'_2\}$;
it is satisfied e.g.\ by the two successive valuations $\v_a\: \{r_1 \mapsto 1, r_2 \mapsto 4
\}$ and $\v_b\:\{r_1 \mapsto 2, r_2 \mapsto 3 \}$. Similarly, $r_1 =
r'_1 < r'_2 = r_2$ is a constraint corresponding to the set
$\{r_1<r_2, r_1=r'_1, r_2=r'_2, r'_1<r'_2\}$.
Note that the set $\{r_1<r_2, r_1>r'_1, r_2<r'_2, r'_1>r'_2\}$
does not represent a constraint:
it is not consistent since $r_1 > r'_1 > r'_2 > r_2 > r_1$
implies $r_1 > r_1$, violating irreflexivity,  and thus does not correspond to any total non-strict preorder.
\manufixed{R1: Strictly speaking, this violates *irreflexvity*, not reflexivity.}{changed; also added "total non-strict preorder" and another non-example (non-total)}
Another counter-example is $r \leq r'$ for $R=\{r\}$: it is not a constraint since it is not total.
%In words: if
%we start with a valuation $\v_a\in\D^R$ such that $\v_a(r_1) < \v_a(r_2)$ (written $r_1<r_2$ in the constraint),
%then $r_1$ increases and $r_2$ decreases (written $r_1<r_1'$ and $r_2>r_2'$) in such a way that a new valuation $\v_b\in\D^R$ has $\v_b(r_1) < \v_b(r_2)$ (written $r_1' < r_2'$).
%An example of such valuations could be .
\end{example}

\begin{definition}
  \label{def:constraint_sequence}
A \emph{constraint sequence} is then an infinite sequence of constraints $C_0 C_1 \dots$
(when a sequence is finite, we explicitly state it).

It is \emph{consistent} if for every $i$:
$unprime({C_i}_{|R'}) = {C_{i+1}}_{|R}$,
i.e.\ the register order at the end of step $i$ equals the register order at the beginning of step $i+1$.

A valuation $\mc w \in \D^{R\cup R'}$ \emph{satisfies} a constraint $C$, written $\mc w \models C$,
if every atom holds when we replace every $r \in R \cup R'$ by $\mc w(r)$.
A constraint sequence is \emph{satisfiable} if there exists a sequence of valuations
$\v_0 \v_1 ... \in (\D^R)^\omega$ such that
$\v_i\cup\v'_{i+1} \models C_i$ for all $i\geq0$.
If, additionally\footnote{Recall that over $(\bbN,\leq)$, $0$ denotes its minimal element. Over $(\bbQ, \leq)$, its choice is irrelevant.}, $\v_0 = 0^R$, then it is \emph{$0$-satisfiable}.
Note that satisfiability implies consistency, but not vice versa, as we show below.
\end{definition}
Note also that the notions of constraints and constraint sequences over $(\bbN,\leq)$ and over $(\bbQ,\leq)$ syntactically coincide. This is done on purpose, to ease the comparison between the two domains. When this matters, we always make it clear on which domain a constraint sequence is meant to be interpreted.

Finally, remark that consistency also coincides for both domains, while satisfiability does not, as witnessed by the constraint sequence $(\{r>r'\})^\omega$ over $R = \{r\}$: it is satisfiable in $\bbQ$ but not in $\bbN$.
\begin{example}
  We give a richer example. Let $R = \{r_1,r_2,r_3,r_4\}$.
  Let a consistent constraint sequence $C_0 C_1\dots$ start with
  $$
  \{r'_2 < r_1 = r'_1 < r_2 < r_3 = r'_4 < r_4 = r'_3\} \{r'_1 < r_2 = r_2' < r_1 < r_4  = r'_3 < r_3 = r'_4\}
  $$
  % $$
  %   \scalebox{0.9}{$\big\{ r_1\!<\!r_2\!<\!r_3\!<\!r_4, r_4\!=\!r'_3, r_3\!=\!r'_4,r_1\!=\!r'_1, r_1 \!>\! r'_2 \big\}
  %   \big\{ r_2\!<\!r_1\!<\!r_4\!<\!r_3, r_4\!=\!r'_3, r_3\!=\!r'_4,r_1\!=\!r'_1, r_2 \!>\! r'_1 \big\}$}
  % $$%
  % Note that we omit some atoms in $C_0$ and $C_1$ for readability:
  % although they are not maximal (e.g.\ $C_0$ does not contain $r'_2<r'_1<r'_4<r'_3$),
  % they can be uniquely completed to maximal sets.
  Figure~\ref{fig:constraint-seq} visualises $C_0 C_1$ plus a bit more constraints.
  The black lines represent the evolution of the same register;
  ignore the colored paths for now.
  The constraint $C_0$ describes the transition from moment $0$ to $1$,
  and $C_1$ the transition from moment $1$ to $2$.
  This finite constraint sequence is satisfiable in $\bbQ$ and in $\bbN$.
  For example, the valuations can start with
  $\v_0 = \{ r_4 \mapsto 6, r_3 \mapsto 5, r_2 \mapsto 4, r_1 \mapsto 3 \}$.
  In $\bbN$, no valuations starting with $\v_0(r_3)<5$ can satisfy the sequence.
  Further, since the constraint $C_0$ requires all registers in $R$ to differ,
  the sequence is not $0$-satisfiable in $\bbQ$ nor in $\bbN$.
\begin{figure}[htb]
  \centering
    \begin{tikzpicture}[->, >=stealth', bend angle=45, auto, node distance=2.75cm,yscale=0.9,scale=0.4,nodes={font=\scriptsize}]
      \tikzstyle{every state}=[text=black]
      \tikzstyle{dot}=[circle,draw=black,fill=black,minimum size=0.8mm,inner sep=0]

      % 16
      % |   |   |   |   |
      % |   |   |   |   |
      % |   |   |   |   |
      % |   |   |   |   |
      % 0---4---8--12--16------28

      % canvas
      \draw[->,black!30] (0, 8) -- (0, 16.5);
      \draw[->,black!30] (0, 8) -- (26, 8);
      \node at (-1.5,16.3) {\footnotesize order};
      \node at (26.3,7.3)  {\footnotesize time};
      \node at (0,7.3) {\scriptsize $0$};

      \foreach \x in {1,...,6}
      {
        \draw[-,black!10] (\x*4,8)--(\x*4,15);
        \node at (\x*4,7.3) {\scriptsize $\x$};
      }

      % threads
      % step 0
      \node[dot] (n03) at (0,15){}; \node at (-1,15) {$r_4$};
      \node[dot] (n02) at (0,14){}; \node at (-1,14) {$r_3$};
      \node[dot] (n01) at (0,13){}; \node at (-1,13) {$r_2$};
      \node[dot] (n00) at (0,12){}; \node at (-1,12) {$r_1$};

      % step 1
      \node[dot] (n12) at (4,15){}; \draw[color=black!70,-] (n02) -- (n12);
      \node[dot] (n13) at (4,14){}; \draw[color=black!70,-] (n03) -- (n13);
      \node[dot] (n11) at (4,11){}; \draw[color=black!70,-] (n01) -- (n11);
      \node[dot] (n10) at (4,12){}; \draw[color=black!70,-] (n00) -- (n10);

      % step 2
      \node[dot] (n23) at (8,15){}; \draw[color=black!70,-] (n13) -- (n23);
      \node[dot] (n22) at (8,14){}; \draw[color=black!70,-] (n12) -- (n22);
      \node[dot] (n21) at (8,11){}; \draw[color=black!70,-] (n11) -- (n21);
      \node[dot] (n20) at (8,10){};  \draw[color=black!70,-] (n10) -- (n20);

      % step 3
      \node[dot] (n33) at (12,14){}; \draw[color=black!70,-] (n23) -- (n33);
      \node[dot] (n32) at (12,15){}; \draw[color=black!70,-] (n22) -- (n32);
      \node[dot] (n31) at (12,9){};  \draw[color=black!70,-] (n21) -- (n31);
      \node[dot] (n30) at (12,10){};  \draw[color=black!70,-] (n20) -- (n30);

      % step 4
      \node[dot] (n43) at (16,15){}; \draw[color=black!70,-] (n33) -- (n43);
      \node[dot] (n42) at (16,14){}; \draw[color=black!70,-] (n32) -- (n42);
      \node[dot] (n41) at (16,12){}; \draw[color=black!70,-] (n31) -- (n41);
      \node[dot] (n40) at (16,10){};  \draw[color=black!70,-] (n30) -- (n40);

      % step 5
      \node[dot] (n53) at (20,14){}; \draw[color=black!70,-] (n43) -- (n53);
      \node[dot] (n52) at (20,15){}; \draw[color=black!70,-] (n42) -- (n52);
      \node[dot] (n51) at (20,12){}; \draw[color=black!70,-] (n41) -- (n51);
      \node[dot] (n50) at (20,10){};  \draw[color=black!70,-] (n40) -- (n50);

      % step 6
      \node[dot] (n63) at (24,15){}; \draw[color=black!70,-] (n53) -- (n63);
      \node[dot] (n62) at (24,14){}; \draw[color=black!70,-] (n52) -- (n62);
      \node[dot] (n61) at (24,12){}; \draw[color=black!70,-] (n51) -- (n61);
      \node[dot] (n60) at (24,10){};  \draw[color=black!70,-] (n50) -- (n60);

      \node[color = black!80!yellow] (c1) at (2,15.5) {$c_1$};
      \node[color = black!80!blue] (c2) at (2,11) {$c_2$};
      \node[color = black!80!green] (c3) at (10.25,12) {$c_3$};
      \node[color = black!80!red] (c4) at (18,9.5) {$c_4$};

      % chains
      \begin{scope}[on background layer]
        \draw[-,color=yellow,line width=0.6mm,opacity=0.7] (n03) -- (n12) -- (n23) -- (n32) -- (n43) -- (n52) -- (n63);
        \draw[-,color=blue,line width=0.6mm,opacity=0.35] (n03) -- (n02) -- (n01) -- (n00) -- (n11) -- (n20) -- (n31);
        \draw[-,color=red,line width=0.6mm,opacity=0.4] (n31) -- (n40) -- (n50) -- (n51) -- (n53) -- (n52);
        \draw[-,color=green,line width=0.6mm,opacity=0.4] (n23) -- (n33) -- (n21) -- (n30) -- (n31);
      \end{scope}
    \end{tikzpicture}
    \caption{Visualisation of a constraint sequence.
             Individual register values are depicted by black dots,
             and dots are connected by black lines when they talk about the same register.
             Yellow/blue/green/red paths depict chains (cf infra).}
    \label{fig:constraint-seq}
\end{figure}
\end{example}

\paragraph{Chains}
This section describes a characterisation of satisfiable constraint sequences
that is amenable to being recognised by automata. The proofs are quite technical, so we defer them to Section~\ref{sec:constraints} and for the time being we only give an intuition.
\begin{definition}[Chains]
  \label{def:chains}
Fix $R$ and a consistent constraint sequence $C_0 C_1 \dots$ over $R$.
A (decreasing) \emph{two-way chain} is a finite or infinite sequence
$(r_0,m_0) \tr_0 (r_1,m_1) \tr_1 ... \in \big((R\x\bbN) \cdot \{=,>\}\big)^{*,\omega}$
satisfying the following (note that $m_0$ can differ from $0$).
\li
\- $m_{i+1} \!=\! m_i$, or $m_{i+1} \!=\! m_i+1$ (time flows forward), or $m_{i+1} = m_i-1$ (backwards).
\- If $m_{i+1} = m_i$ then $(r_i \triangleright_i r_{i+1}) \in C_{m_i}$.
\- If $m_{i+1} = m_i+1$ then $(r_i \triangleright_i r'_{i+1}) \in C_{m_i}$.
%\- If $m_{i+1} = m_i-1$ then $(r_{i+1} \triangleright_i r'_i) \in C_{m_i-1}$.
\- If $m_{i+1} = m_i-1$ then $(r'_i \triangleright_i r_{i+1}) \in C_{m_i-1}$.
\il
The \emph{depth} of a chain is the number of $>$;
when it is infinity, the chain is \emph{infinitely} decreasing.
Figure~\ref{fig:constraint-seq} highlights four two-way chains (there are more) with yellow, blue, green and red colors.
For instance,
the green-colored chain $c_3$, defined as $(r_4,2) > (r_3,3) > (r_2,2) > (r_1,3) > (r_2,3)$, has depth $4$.

Given a moment $i$ and a register $x$,
a (decreasing) \emph{right two-way chain starting in $(x,i)$} (\emph{r2w} for short)
is a two-way chain
$(x,i) \tr_1 (r_1,m_1) \tr_2 (r_2,m_2) \ldots $
such that $m_j\geq i$, $\tr_j \in \{=,>\}$,
for all $j$.
Thus, all elements appear to the right of the starting moment $(x,i)$.

We define \emph{one-way} chains similarly, except that time now flows forwards or stays the same, and that they can be either increasing or decreasing:
\li
\- $m_{i+1} \!=\! m_i$ (time does not flow), or $m_{i+1} \!=\! m_i+1$ (time flows forward).
\- If $m_{i+1} = m_i$ then $(r_i \bowtie_i r_{i+1}) \in C_{m_i}$.
\- If $m_{i+1} = m_i+1$ then $(r_i \bowtie_i r'_{i+1}) \in C_{m_i}$.
\il
A one-way chain is \emph{decreasing} (respectively, \emph{increasing}) if for all $i \geq 0$, $\bowtie_i \in \{>,=\}$ (resp., $\bowtie_i \in \{<,=\}$).

In Figure~\ref{fig:constraint-seq},
the blue ($c_2$) chain $(r_4,0) > (r_3,0) > (r_2,0) > (r_1,0) > (r_2,1) > (r_1,2) > (r_2,3)$ is one-way decreasing chain of depth 6;
the same sequence is also a two-way chain.
The red ($c_4$) chain $(r_2,3) < (r_1,4) = (r_1,5) < (r_2,5) < (r_4,5) < (r_3,5)$ is one-way increasing of depth 4;
if we read the sequence in reverse, it represents a two-way chain (two-way chains are always decreasing).
Sometimes we write ``chain'' omitting whether it is two- or one-way.
%\ak{define one-way chains as those that never stutter? (always $m_{i+1} = m_i+1$)}

A \emph{stable chain} is an infinite chain
%%$(r_0,m) = (r_1,m+1) = (r_2,m+2) = ... \in \big((R\x\bbN)\cdot\{=\}\big)^\omega$;
$(r_0,m) = (r_1,m+1) = (r_2,m+2) = ...$;
it can also be written as $(m,r_0r_1r_2...)$.
In Figure~\ref{fig:constraint-seq}, the yellow ($c_1$) chain $(0,(r_4r_3)^\omega)$ is stable.
Given a stable chain $\chi_r = (m,r_0r_1...)$ and a
%(finite or infinite)
chain
$\chi_s = (s_0,n_0) \bowtie_0 (s_1,n_1) \bowtie_1 ...$,
where $n_i \geq m$ for all $i$,
the chain $\chi_r$ is \emph{above} $\chi_s$ (equiv., $\chi_s$ is \emph{below} $\chi_r$) if
for all $i$ the constraint $C_{n_i}$ contains $r_{n_i-m} > s_{i}$ or $r_{n_i-m} = s_{i}$;
here we used $n_i-m$ because the register at
moment $n_i$ in the chain $\chi_r$ is $r_{n_i-m}$.
In Figure~\ref{fig:constraint-seq}, the yellow chain
$(0,(r_4r_3)^\omega)$ is above all colored chains.
A stable chain $(m,r_0r_1...)$ is \emph{maximal} if
it is above all other stable chains starting after $m$.
In Figure~\ref{fig:constraint-seq},
the yellow chain $(0,(r_4r_3)^\omega)$ is maximal (assuming the
sequence evolves in a similar fashion). %
%
%
%Notice that a constraint sequence can have from zero up to $|R|$ non-equivalent stable chains,
%where equivalence is in terms of equality of the registers composing them.
Notice that if a sequence has a stable chain, then it has a maximal one.
A \emph{ceiled chain} is a chain that is below a maximal stable chain.
A constraint sequence can have an infinite number of ceiled chains;
it can also have zero, e.g.\ when there are no stable chains.
\end{definition}
Note that in this section, we mostly focus on one-way chains and right two-way chains, while two-way chains are used in Section~\ref{sec:chain-characterisation} as a technical intermediate. In the latter section, we show:
\begin{restatable}{lemma}{lemzsatisfNo}
  \label{lem:0-satisf-N-1}
  A consistent constraint sequence is $0$-satisfiable in $\bbN$ iff there exists $\Bound \geq 0$ such that:
  \lo[1.]
  \- it has no infinitely decreasing one-way chains,
  \- \label{itm:N_ceiled_bounded} the ceiled one-way chains have a depth at most $\Bound$
  \- it starts in $C_0$ s.t.\ ${C_0}_{|R} = \{r\!=\!s\mid r,s \in R\}$, and
  \- it has no decreasing one-way chains of depth $\geq\!1$ from $(r,0)$ for any $r$.
  \ol
\end{restatable}
In line with Example~\ref{ex:AS_N_not_reg}, the above characterisation is not $\omega$-regular; the culprit is item~\ref{itm:N_ceiled_bounded}.
% Actually, we show in Section~\ref{sec:max-automata-characterisation} that $0$-satisfiable constraint sequences are recognisable by max-automata (it was known from~~\cite[Appendix~C]{ST11} that they are recognisable by non-deterministic $\omega B$ automata~\cite{BC06}).
We thus define quasi-feasible constraint sequences, by relaxing the condition to asking that there are no infinite increasing ceiled chains.
\begin{definition}
  \label{def:QFeasible}
  A consistent constraint sequence is \emph{quasi-feasible} whenever:
  \li
  \- it has no infinitely decreasing one-way chains,
  \- it has no infinitely increasing ceiled one-way chains,
  \- it starts in $C_0$ s.t.\ ${C_0}_{|R} = \{r\!=\!s\mid r,s \in R\}$, and
  \- it has no decreasing one-way chains of depth $\geq\!1$ from $(r,0)$ for any $r$.
  \il
\end{definition}
In Section~\ref{sec:lasso-characterisation} on page~\pageref{lem:N-satisf-lasso}, we show:

\medskip

\textbf{Lemma~\ref{lem:N-satisf-lasso}.} \emph{A lasso-shaped consistent constraint sequence is $0$-satisfiable if and only if it is quasi-feasible.}

\medskip

We conclude the section by formally relating action words (see Definition~\ref{def:action_word}) with constraint sequences.
\paragraph{Action words and constraint sequences}
Every action word naturally induces a unique constraint sequence.
For instance, for registers $R = \{r,s\}$,
an action word starting with $(\{r<*, s<*\},\{s\})$
(test whether the current data $\d$ is above the values of $r$ and $s$, store it in $s$)
induces a constraint sequence starting with $\{ r=s, r=r', s<s', r'<s'\}$
(the atom $r=s$ is due to all registers being equal initially).
This is formalised in the next lemma, which is notation-heavy but says a simple thing:
given an action word,
we can construct, on the fly, a constraint sequence that is $0$-satisfiable iff the action word is feasible.
For technical reasons, we need a new register $r_d$ to remember the last Adam data. The proof is on page~\pageref{page:constr-description}, so as not to break the flow of the argument.
\begin{restatable}{lemma}{lemMappingConstr}
  \label{lem:mapping-constr}
  Let $R$ be a set of registers, $R_d = R \uplus \{r_d\}$, and
  $\D$ be $(\bbN,\leq)$ or $(\bbQ,\leq)$.
  There exists a mapping $constr : \Pi \x\Tst\x\Asgn \to \sf C$
  from state constraints $\Pi$ over $R_d$ and tests-assignments over $R$ to constraints $\sf C$ over $R_d$,
  such that for all action words $a_0 a_1 a_2 ... \in (\Tst\x\Asgn)^\omega$,
    $a_0 a_1 a_2 ...$ is feasible iff $C_0 C_1 C_2 ...$ is $0$-satisfiable,
  where
  $\forall i\!\geq\!0$: $C_i = constr(\pi_i, a_i)$, $\pi_{i+1}\!=\!unprime({C_i}_{|R'_d})$,
  $\pi_0 = \{r\!=\!s\mid r,s\in R_d\}$.
\end{restatable}
Then, given a set of registers $R$, we say that an action word $\overline{a}$ is \emph{quasi-feasible} whenever $constr(\overline{a})$ is quasi-feasible. We correspondingly denote by $\QFeasible_\bbN(R)$ the set of \emph{quasi-feasible} action words over $R$.

\subsection{The \emph{$\omega$-regular} game $G_S^{reg}$}
\label{sec:Greg_N}
% $W_f$ is not $\omega$-regular, and the known results
% over deterministic max-automata do not suffice to obtain determinacy nor
% finite-memoriness, which will both prove useful for the transducer synthesis
% problem in Section~\ref{sec:transducer_synthesis}.
% We thus define an $\omega$-regular subset $W_f^{reg} \subseteq W_f$ which is
% equi-realisable to $W_f$.

%By Lemma~\ref{lem:N-satisf-lasso}, this entails $0$-satisfiability of lasso-shaped
%constraint sequences.
After this long but necessary detour through constraint sequences, we are ready to define the $\omega$-regular game associated with the automaton game. Recall that in Section~\ref{sec:automaton_game}, given a one-sided automaton $S$, we defined $G_S^f = (\VA,\VE,v_0,E, W_S^f)$. We now let $G_S^{reg} = (\VA,\VE,v_0,E, W_S^{reg})$. Thus, it has the same vertices and edge relation:
$\VA = \{\qinit\} \cup (\Sigma \x Q_A)$,
$\VE = \Tst\x\Asgn \x Q_E$,
$v_0 = \qinit$,
$E = E_0 \cup E_\forall \cup E_\exists$
where:
\li
\-
$E_0 = \big\{\big(v_0, (\tst,\asgn,u_0)\big) \mid \delta(v_0,\tst)=(\asgn,u_0)\big\}$,
\-
$E_\forall = \big\{ \big((\sigma,v),(\tst,\asgn,u)\big)\mid \delta(v,\tst)=(\asgn,u) \big\}$, and
\-
$E_\exists = \big\{ \big((\tst,\asgn,u),(\sigma,v)\big)\mid \delta(u,\sigma)=v \big\}$.
\il
However, the winning condition is now:
$$
W_S^f = \left\{ v_0(\tst_0,\asgn_0,u_0)(\sigma_0,v_1) \ldots\;\middle|\;
  \begin{array}{l}
    (\tst_0\asgn_0) \ldots \in \QFeasible_{\bbN}(R) \\
    \Impl
    v_0u_0v_1u_1\dots\models \alpha
  \end{array}\right\}
$$%
i.e., we replaced $\Feasible_{\bbN}(R)$ with $\QFeasible_{\bbN}(R)$.

First, by Proposition~\ref{lem:DPA-for-satisf-lasso}, we know that $\QFeasible_{\bbN}(R)$ is $\omega$-regular. Thus:
\begin{proposition}
  Let $S$ be a one-sided automaton, and define $G_S^{reg}$ as above. Then, $G_S^{reg}$ is an $\omega$-regular game.
\end{proposition}
% we can build
% a deterministic parity automaton with a number of
% states exponential in $\size{R}$ and polynomial in $\size{Q}$ and a priority index linear in $c$
% recognising $W_f^{reg}$.
% Let $G_f^{reg}$ be the finite-arena game with the same arena as $G_f$, with winning condition $W_f^{reg}$.
We now show that it is equivalent with the Church game $G_S$.
\begin{proposition}
  \label{prop:GS_Gf_Gfreg}
  Let $S$ be a one-sided automaton, $G_S$ the corresponding Church game, $G_S^f$ its automaton game, and $G_S^{reg}$ its associated $\omega$-regular game. The following are equivalent:
  \begin{enumerate}
  \item \label{itm:Eve_wins_reg} Eve has a winning strategy in $G_S^{reg}$
  \item \label{itm:Eve_fin_wins_reg} Eve has a finite-memory winning strategy in $G_S^{reg}$
  \item \label{itm:Eve_fin_wins_f} Eve has a finite-memory winning strategy in $G_S^f$
  \item \label{itm:Eve_wins_f} Eve has a winning strategy in $G_S^f$
  \item \label{itm:Eve_wins} Eve has a winning strategy in $G_S$.
\end{enumerate}
\end{proposition}
\begin{proof}
  We start with the chain of implications $\ref{itm:Eve_wins_reg} \Rightarrow \ref{itm:Eve_fin_wins_reg} \Rightarrow \ref{itm:Eve_fin_wins_f} \Rightarrow \ref{itm:Eve_wins_f} \Rightarrow \ref{itm:Eve_wins}$.

  The implication $(\ref{itm:Eve_wins_reg}) \Rightarrow (\ref{itm:Eve_fin_wins_reg})$ holds because $G_S^{reg}$ is $\omega$-regular, and we know that those games are finite-memory determined~\cite{GH82}.

  Then, $(\ref{itm:Eve_fin_wins_reg}) \Rightarrow (\ref{itm:Eve_fin_wins_f})$ follows from the fact that $G_S^{reg}$ is actually harder than $G_S^f$, i.e. $W_S^{reg} \subseteq W_S^f$, because $\Feasible_{\bbN}(R) \subseteq \QFeasible_{\bbN}(R)$.

  $(\ref{itm:Eve_fin_wins_f}) \Rightarrow (\ref{itm:Eve_wins_f})$ is immediate.

  $(\ref{itm:Eve_wins_f}) \Rightarrow (\ref{itm:Eve_wins})$ is exactly Proposition~\ref{prop:Gf_GS_sound}.

  It remains to show that $(\ref{itm:Eve_wins}) \Rightarrow (\ref{itm:Eve_wins_reg})$.
  We proceed by contraposition. Thus, assume that Eve does not have a winning strategy in $G_f^{reg}$. By finite-memory determinacy of games with parity objectives,
  in $G_f^{reg}$
  Adam has a finite-memory winning
  strategy $\lambda_\forall^f : \VA(\VE\VA)^* \to \VE$
  (equiv., $\lambda_\forall^f : \Sigma^* \to \Tst$).
  We show the following:
  \begin{proposition}
    \label{prop:Adam_Greg_GS}
    If Adam has a winning strategy in $G_S^{reg}$, then he has a winning strategy in $G_S$.
  \end{proposition}
  \begin{proof}
  At first, it is not clear how to instantiate it to
  a \emph{data} strategy $\lambda_\forall^\bbN : \Sigma^* \to \bbN$ winning in $G_S$.
  For instance,
  if the strategy $\lambda_\forall^f$ in $G_f^{reg}$ dictates Adam to pick the test $*>r$,
  it is not clear which data should $\lambda_\forall^\bbN$ pick
  ($\v(r)+1$, $\v(r)+2$, more?) because for different strategies of Eve different values may be needed.
  To construct $\lambda_\forall^\bbN$ from $\lambda_\forall^f$ that beats every Eve, we show that for any finite-memory strategy of Adam, there is a uniform bound on the depth of all its r2w chains. This is formalised by the following claim (that we prove afterwards):
  \begin{claim}\label{clm:Adam-is-bounded}
    Let $\lambda^f_\forall$ be a finite-memory strategy of Adam that is winning in
    $G^{reg}_f$. There exists a bound $\Bound\geq 0$ such that for each play $\rho$ consistent with $\lambda_\forall^f$, for each right two-way chain $\gamma$ of the constraint sequence induced by $\rho$ (starting in some $(r,i) \in R \times \bbN$), $\mathit{depth}(\gamma) \leq \Bound$.
    % $$
    % \begin{array}{ll}
    %   &\forall \text{constraint sequences resulting from playing with $\lambda_A^f$}.\\
    %   &\forall x \in R. \forall i\geq 0.
    %     \text{$\forall$r2wch from $(x,i)$}\:
    %     \mathit{depth}(r2wch) \leq \Bound.
    % \end{array}
    % $$
  \end{claim}
  Thanks to existence of this uniform bound $\Bound$,
  we can construct $\lambda_\forall^\bbN$ from $\lambda_\forall^f$ as follows.
  First, translate the currently played action-word prefix
  $(\tst_0,\asgn_0) ... (\tst_m, \asgn_m)$
  into a constraint-sequence prefix using Lemma~\ref{lem:mapping-constr}.
  Then apply to it the data-assignment function from Lemma~\ref{lem:data-assign-func}.
  By construction, for each play in $G$ consistent with $\lambda_\forall^{\bbN}$, the
  corresponding run in $S$ is a play consistent with
  $\lambda_\forall^f$ in $G_f^{reg}$.
  As $\lambda_\forall^f$ is winning, this run is not accepting,
  i.e.\ the play is winning for Adam in $G_S$.

  Therefore,
  $\lambda_\forall^{\bbN}$ is a winning Adam's strategy in $G_S$.
  \hfill {\footnotesize End of the proof of Prop.~\ref{prop:Adam_Greg_GS}}
  \end{proof}
  As a consequence, Eve does not have a winning strategy in $G_S$, which means that $(\ref{itm:Eve_wins}) \Rightarrow (\ref{itm:Eve_wins_reg})$.
  \hfill {\footnotesize End of the proof of Prop.~\ref{prop:GS_Gf_Gfreg}}
\end{proof}

We are left to prove Claim~\ref{clm:Adam-is-bounded}.

\paragraph{Boundedness of right two-way chains induced by Adam \linebreak (Proof of Claim~\ref{clm:Adam-is-bounded})}
\begin{proof}[Proof idea]
% Suppose Adam wins $G^{reg}_f$ using a finite-memory strategy
% $\lambda_A^f : \Sigma^* \to \Tst$.
% % The plays consistent with $\lambda_A^f$ satisfy the following important property.
% Fix a
% moment $i$ and a register $x$. Then, after the moment $i$, only a bounded number
% of values can be inserted below the value of register $x$ at moment $i$.
% Similarly, if we fix two registers at moment $i$, there can only be a bounded
% number of insertions between the values of $x$ and $y$ at moment $i$. Indeed,
% by finiteness of Adam strategy, once the number of such insertions is
% larger than the memory of Adam, Eve can repeat her actions to force an infinite
% number of such insertions, leading to a play with an unfeasible action sequence
% and hence won by Eve.
If Adam has a finite-memory strategy, then if a decreasing right two-way chain $\gamma$ is sufficiently deep, Eve can force Adam to loop in a memory state in a way such that the loop can be iterated while preserving the chain. We can additionally ensure that this chain contains a strictly decreasing or increasing segment. When iterated, this segment makes the chain unfeasible. Indeed, if the segment is decreasing, iterating the loop yields an infinite descending chain in $\bbN$, which is not feasible. The case of an increasing fragment happens when $\gamma$ is decreasing from right to left (recall that it is a two-way chain), so increasing from left to right. When iterated, this yields an infinite increasing chain, which is perfectly fine in $\bbN$. However, it can be bounded from above with the help of $\gamma$: before decreasing from right to left, $\gamma$ has to go from left to right, since it is a right chain (i.e. it is not allowed to go to the left of its initial position). On the strictly increasing segment, this left-to-right prefix is either constant or decreasing, so when the loop is iterated it provides an upper bound for our increasing chain.
\end{proof}
\begin{proof}
  We now move to the formal proof. We could use a Ramsey argument in the spirit of Lemma~\ref{lem:satisf-N-1} to extract an infinite \emph{one-way} chain that is either increasing or decreasing. However, this amounts to breaking a butterfly upon the wheel, and we prefer to rely on a simpler pumping argument, which also gives a finer-grained perception of what is happening there. In particular, it provides a bound $\Bound$ that does not depend on a Ramsey number.

  Thus, let $\lambda_\forall^f$ be a finite-memory strategy of Adam with memory $M$ that is winning in $G_S$. Suppose, towards a contradiction, that there exists a play $\rho$ that is consistent with $\lambda_\forall^f$ and which contains a decreasing right two-way chain of depth $D > \size{M} \cdot 2^{2\size{R}^2}$.
We denote it $\gamma = (r_0,m_0) \tr_0 (r_1,m_1) \tr_1 (r_2,m_2) \tr_2 \dots \tr_{n-1} (r_n,m_n)$, where for all $0 \leq i \leq n$, $\tr_i \in \{>,=\}$, $r_i \in R$ and $m_i \in \bbN$. Given a two-way chain and a position $i \geq m_0$, we define the crossing section at $i$ as the sequence of registers that occur at position $i$, ordered by their appearance in the chain: $\crossSection{\gamma}{i}$ is the maximal subword of $\gamma$ that contains letters of the form $(r,i)$ for some $r \in R$ (see Fig.~\ref{fig:chain_crossing_section}, where we depicted a chain that has two identical crossing sections at positions $i$ and $j$).
   \begin{figure*}[htp]
     \centering
  \begin{subfigure}[t]{0.49\linewidth}
     \begin{tikzpicture}[xscale=0.35,yscale=0.4]
       \draw[->,thick,>=stealth'] (-1,0) -- (10,0);
       \coordinate[label=right:play] (lp) at (10,0);
       \coordinate[label=above:$i$] (li) at (0,0.25);
       \coordinate[label=above:$j$] (lj) at (9,0.25);
       \draw (0,0.25) -- (0,-11);
       \draw (9,0.25) -- (9,-11);
       \coordinate[label=left:$r_3$] (a0) at (0,-1.1);
       \coordinate[label=above left:$r_2$] (a1) at (0,-3);
       \coordinate[label={below left:$r_4$}] (a2) at (0,-4);
       \coordinate[label=above left:$r_1$] (a3) at (0,-8);
       \coordinate[label=below left:$r_6$] (a4) at (0,-10);
       \coordinate[label=above right:$r_3$] (b0) at (9,-1.1);
       \coordinate[label=below right:$r_2$] (b1) at (9,-2);
       \coordinate[label=above right:$r_4$] (b2) at (9,-5);
       \coordinate[label=below right:$r_1$] (b3) at (9,-6);
       \coordinate[label=right:$r_6$] (b4) at (9,-10);
       \coordinate (c1) at (8,-7);
       \draw (a0) -- (b0);
       \draw (b0) .. controls (12,-1) and (12,-1.5) .. (b1);
       \draw (b1) -- (a1);
       \draw (a1) arc (90:270:0.5);
       \draw (a2) -- (b2);
       \draw (b2) arc (90:-90:0.5);
       \draw (b3) .. controls (6,-6) and (6,-7) .. (c1) .. controls  (10,-7) and (8,-7.5) .. (a3);
       \draw (a3) arc (90:270:1);
       \draw[->] (a4) -- (b4);
     \end{tikzpicture}
     \caption[A chain with two identical crossing sections]{A chain with two identical crossing sections.}
     \label{fig:chain_crossing_section}
   \end{subfigure} \hfill %
   \begin{subfigure}[t]{0.49\linewidth}
     \begin{tikzpicture}[xscale=0.25, yscale=0.4]
       \draw[->,thick,>=stealth'] (-1,0.1) -- (20,0.1);
       \coordinate[label=right:play] (lp) at (20,0.25);
       \coordinate[label=above:$i$] (li) at (0,0.25);
       \coordinate[label=above:$j$] (lj) at (9,0.25);
       \coordinate[label=above:$j'$] (lk) at (18,0.25);
       \draw (0,0.25) -- (0,-15);
       \draw (9,0.25) -- (9,-15);
       \draw (18,0.25) -- (18,-15);
       \coordinate[label=above left:$r_3$] (a0) at (0,-1);
       \coordinate[label=above left:$r_2$] (a1) at (0,-4);
       \coordinate[label=below left:$r_4$] (a2) at (0,-5);
       \coordinate[label=above left:$r_1$] (a3) at (0,-12);
       \coordinate[label=below left:$r_6$] (a4) at (0,-14);
       \coordinate[label=above right:$r_3$] (b0) at (9,-1);
       \coordinate[label=below right:$r_2$] (b1) at (9,-3);
       \coordinate[label=above right:$r_4$] (b2) at (9,-6);
       \coordinate[label=below right:$r_1$] (b3) at (9,-10);
       \coordinate[label=below right:$r_6$] (b4) at (9,-14);
       \coordinate[label=above right:$r_3$] (c0) at (18,-1);
       \coordinate[label=below right:$r_2$] (c1) at (18,-2);
       \coordinate[label=above right:$r_4$] (c2) at (18,-7);
       \coordinate[label=below right:$r_1$] (c3) at (18,-8);
       \coordinate[label=below right:$r_6$] (c4) at (18,-14);
       \draw (a0) -- (b0) -- (c0);
       \draw (c0) .. controls (22,-1) and (22,-1.5) .. (c1);
       \draw (c1) -- (b1) -- (a1);
       \draw (a1) arc (90:270:0.5);
       \draw (a2) -- (b2) -- (c2);
       \draw (c2) arc (90:-90:0.5);
       \draw (c3) .. controls (15,-8) and (15,-9) .. (17,-9) .. controls  (19,-9) and (17,-9.5) .. (b3);
       \draw (b3) .. controls (6,-10) and (6,-11) .. (8,-11) .. controls  (10,-11) and (8,-11.5) .. (a3);
       \draw (a3) arc (90:270:1);
       \draw[->] (a4) -- (b4) -- (c4);
     \end{tikzpicture}
     \caption[Iterating a fragment of a play]{Iterating a fragment of a play. We are able to glue the chain since the crossing sections and the order between registers are the same at positions $i$ and $j$.}
     \label{fig:chain_crossing_section_iter}
     \end{subfigure}
   \end{figure*}
   This construction is reminiscent of the techniques that are used to study loops in two-way automata or transducers, hence the name.
   % However, here, we can assume without loss of generality that $\gamma$ is \kl(chain){non-looping} (\reflemma{feasible_non_looping_chain}), so each \kl{register} appears at most once, which simplifies the argument.
   At each position, there are $\size{M}$ distinct memory states for Adam, less than $2^{\size{R}^2}$ many distinct crossing sections and less than $2^{\size{R}^2}$ many possible orderings of the registers. As a consequence there exists two positions $m_0 \leq i < j$ such that $\crossSection{\gamma}{i} = \crossSection{\gamma}{j}$, the memory state of Adam at position $i$ and $j$ is the same, the order between registers at position $i$ is the same at position $j$, and there is at least one occurrence of $>$ in the chain segment. Since $\lambda_\forall^f$ is finite-memory, Eve can repeat her actions between positions $i$ and $j$ indefinitely to iterate this fragment of the play $\rho$. Since the crossing sections match and the order between registers is the same at positions $i$ and $j$, we can glue the chain fragments together to get an infinite two-way chain (see Fig.\ref{fig:chain_crossing_section_iter}), with infinitely many occurrences of $>$. There are two cases:
   \begin{itemize}
   \item There is a fragment that strictly decreases from left to right (as the chain fragment over register $r_4$ in Fig.\ref{fig:chain_crossing_section_iter}). Then, when Eve repeats her actions indefinitely, this yields an infinite descending chain, which means that the play is not feasible (Lemma~\ref{lem:satisf-N}), so Eve wins. This contradicts the fact that $\lambda_\forall^f$ is winning.
   \item All decreasing fragments occur from right to left (as do the fragments over $r_2$ and $r_1$ in Fig.\ref{fig:chain_crossing_section_iter}). Necessarily, the topmost fragment, i.e. the fragment of the register that appears first in $\crossSection{\gamma}{i}$, is left-to-right, since $\gamma$ is a right two-way chain. It is not strictly decreasing, otherwise we are back to the first case. Then, the strictly decreasing fragments are bounded from above by this constant fragment. Iterating the loop yields an infinite increasing chain that is bounded from above, which means that the play is again not feasible, so we again obtain a contradiction.
   \end{itemize}
   Overall, the depth of the decreasing right two-way chains induced by $\lambda_\forall^f$ is uniformly bounded by $\Bound = \size{M} \cdot 2^{2\size{R}^2}$, where $\size{M}$ is the size of Adam's memory.
\end{proof}

We finally have all the cards in hand to show:
\begin{theorem}\label{thm:synt-games-N}
  % For every Church game $G_S$ on a one-sided automaton $S$
  Let $S = (\Sigma, Q, \qinit, R, \delta, \alpha)$ be a one-sided register automaton over ${(\bbN,\leq)}$.
  \lo[1.~]
  \- \label{itm:synt-games-N-dec} The problem of determining if Eve wins the Church synthesis game $G = (\D,\D,S)$ is decidable in time polynomial in $\size{Q}$ and exponential in $c$ and $\size{R}$. % has a winning strategy in $G$.
  % \- If Eve wins, she can win using finite-memory
  % strategy.\todo{define}
  \- \label{itm:synt-games-N-det} $G_S$ is determined, i.e. either Eve or Adam has a winning strategy in $G_S$.
  \ol
\end{theorem}
% \paragraph{Proof of Theorem~\ref{thm:synt-games} for $(\bbN,\leq)$}
\begin{proof}
For $(\bbN,\leq)$,
item (1) follows from Proposition~\ref{prop:GS_Gf_Gfreg}
and from the fact that $G_f^{reg}$ is of size polynomial in $|Q|$ and exponential in $|R|$.
Item (2) on determinacy is proven as follows.
Assume Eve loses $G_S$.
By Proposition~\ref{prop:GS_Gf_Gfreg}, Eve loses $G_f^{reg}$.
In the proof of Proposition~\ref{prop:GS_Gf_Gfreg},
we have shown (Proposition~\ref{prop:Adam_Greg_GS}) that in this case Adam has a strategy winning in the original Church game.
As a consequence, our Church games are determined.
\end{proof}
With the help of Proposition~\ref{prop:GS_Gf_Gfreg}, since finite-memory winning strategies of Eve in $G_S^f$ correspond to register transducer implementations (Proposition~\ref{prop:RT_implem_win_strat_GS}), we also get:
\begin{theorem}\label{thm:synt-decidable-N}
  For specifications defined by deterministic input-driven output register
  automata over data domains $(\bbN,\leq)$,
  the register transducer synthesis problem is equivalent with the synthesis problem (for arbitrary implementations) and can be solved in time polynomial in
  $\size{Q}$ and exponential in $c$ and $\size{R}$.
\end{theorem}

\section{Satisfiability of Constraint Sequences in \texorpdfstring{$(\bbN,\leq)$}{(N,≤)}} \label{sec:constraints}
% \ak{future: make constraint sequences consistent by definition?}
This section studies the problem of checking whether a given infinite sequence of constraints
can be satisfied with values from domain $\bbN$. Recall that constraints and constraint sequences are respectively defined in Definitions~\ref{def:constraint} and~\ref{def:constraint_sequence} on page~\pageref{def:constraint}.
This section's structure is:
\li
\- We start with a simple and relatively known result on satisfiability of constraint sequences in data domain $\bbQ$.
   We then focus completely on $\bbN$.
\- Section~\ref{sec:chain-characterisation} describes conditions on chains that characterise satisfiable constraint sequences (in $\bbN$).
\- Section~\ref{sec:max-automata-characterisation} describes ``max-automata'' characterisation of satisfiable constraint sequences. The max-automaton characterisation checks the conditions on chains introduced in Section~\ref{sec:chain-characterisation}.
\- In the study of Church synthesis games on $\bbN$,
   the crucial role play lasso-shaped constraint sequences and their satisfiability.
   We rely on them when proving Proposition~\ref{prop:GS_Gf_Gfreg}.
   The satisfiability of such sequences is the focus of Section~\ref{sec:lasso-characterisation},
   which shows that the regularity of sequences allows for characterisation of the satisfiability
   using classical $\omega$-regular automata instead of max-automata.
   Thus, in the context of Church synthesis games, the max-automaton characterisation is not used.
\- Section~\ref{sec:data-assignment-function} shows that ``depth-bounded'' constraint sequences can be mapped to satisfying valuations on-the-fly:
   such a data assignment function is used when proving the decidability of Church synthesis games (Proposition~\ref{prop:GS_Gf_Gfreg}),
   namely, to show that winning Adam's strategies in abstracted finite-alphabet games can be instantiated to winning data Adam's strategies in Church synthesis games.
\il

%\paragraph{Satisfiability of constraint sequences in $\bbQ$}
%\label{sec:constraints-Q}
\paragraph*{Satisfiability of constraint sequences in $\bbQ$}
Before proceeding to our main topic of satisfiability of constraint sequences in $\bbN$,
we describe, for completeness, similar results for $\bbQ$.
% Constraint sequences over $\bbQ$ are defined as in $(\bbN,\leq)$ (see Definition~\ref{def:constraint_sequences}), except that the choice of $0$ is irrelevant.

The following result is glimpsed in several places (e.g.\ in~\cite[Appendix C]{ST11}):
%it is shown here for completeness.
a constraint sequence is satisfiable in $\bbQ$ iff it is consistent.
This is a consequence of the following property which holds because $\bbQ$ is dense:
for every constraint $C$ and $\v\in\bbQ^R$ such that $\v \models C_{|R}$,
there exists $\v'\!\! \in \!\bbQ^{R'}$\! such that $\v\!\cup\!\v' \!\models C$.
Consistency can be checked by comparing every two consecutive constraints of the sequence.
Thus, it is not hard to show that
consistent -- hence satisfiable -- constraint sequences in $\bbQ$
are recognisable by deterministic parity automata.

\begin{theorem} \label{thm:0-satisf-Q}
  There is a deterministic parity automaton with two colors and of size exponential in $|R|$
  that accepts exactly all constraint sequences satisfiable (or $0$-satisfiable) in $\bbQ$.
\end{theorem}

\manufixed{R1: p8 Thm 1: Actually, the condition is Buchi, maybe this could be remarked in the proof?}{mentioned in the proof that safety (aka looping) acceptance suffices}

To prove the result,
we first show that a constraint sequence in $\bbQ$ is satisfiable iff it is consistent,
then we construct an automaton checking the consistency.

\begin{lemma}\label{lem:satisf-Q}
  Let $R$ be a set of registers and $\D = \bbQ$.
  A constraint sequence $C_0 C_1\dots$ is satisfiable iff it is consistent.
  It is $0$-satisfiable iff it is consistent and ${C_0}_{|R} = \{r_1=r_2\mid r_1,r_2 \in R\}$.
\end{lemma}

\begin{proof}
  Direction $\Impl$ is simple for both claims, so we only prove direction $\Implied$.

  Consider the first claim, direction $\Implied$.
  Assume the sequence is consistent.
  We construct $\v_0 \v_1 \dots \in (\bbQ^R)^\omega$
  such that
  $\v_i\cup\v'_{i+1} \models C_i$ for all $i$.
  The construction proceeds step-by-step and relies on the following fact ($\dagger$):
  for every constraint $C$ and $\v\in\bbQ^R$ such that $\v \models C_{|R}$,
  there exists $\v' \in \bbQ^{R'}$ such that $\v\cup\v' \models C$.
  Then define $\v_0,\v_1\dots$ as follows:
  start with an arbitrary $\v_0$ satisfying $\v_0 \models {C_0}_{|R}$.
  Given $\v_i \models {C_i}_{|R}$, let $\v_{i+1}$ be any valuation in $\bbQ^R$ that satisfies $\v_i\cup\v'_{i+1} \models C_i$
  (it exists by ($\dagger$)).
  Since $\v_{i+1} \models {C_i}_{|R'}$, and $unprime({C_i}_{|R'}) = {C_{i+1}}_{|R}$ by consistency,
  we have $\v_{i+1} \models {C_{i+1}}_{|R}$, and we can apply the argument again.

  We are left to prove the fact ($\dagger$).
  The constraint $C$ completely specifies the order on $R \cup R'$,
  while $\v$ fixes the values for $R$, and $\v \models C_{|R}$.
  Thus, we can uniquely order registers $R'$ and the values $\{\v(r)\mid r\in R\}$ of $R$ on the $\bbQ$-line.
  Since $\bbQ$ is dense, it is always possible to choose the values for $R'$ that respect this order;
  we leave out the details.

  Consider the second claim, direction $\Implied$.
  Since $C_0 C_1\dots$ is consistent, then by the first claim, it is satisfiable, hence it has a witnessing valuation $\v_0 \v_1 \dots$.
  The constraint $C_0$ requires all registers in $R$ to start with the same value,
  so define $\d = \v_0(r)$ for arbitrary $r\in R$.
  Let $\v'_0 \v'_1 \dots$ be the valuations decreased by $\d$: $\v'_i(r) = \v_i(r) - \d$ for every $r \in R$ and $i\geq 0$.
  The new valuations satisfy the constraint sequence
  because the constraints in $\bbQ$ are invariant under the shift
  (follows from the fact: if $r_1<r_2$ holds for some $\v \in \D^R$, then it holds for any $\v-\d$ where $\d \in \D$).
  The equality $\v'_0 = 0^R$ means that the constraint sequence is $0$-satisfiable.
\end{proof}

We now prove Theorem~\ref{thm:0-satisf-Q}.

\begin{proof}[Proof of Theorem~\ref{thm:0-satisf-Q}]
  \newcommand\Constr{\mathit{Constr}}
  The sought automaton has an alphabet consisting of all constraints.
  By Lemma~\ref{lem:satisf-Q}, for satisfiability, it suffices to construct the automaton that checks consistency,
  namely that every two adjacent constraints $C_1 C_2$ in the input word satisfy the condition $unprime({C_1}_{|R'}) = {C_2}_{|R}$.
  We only sketch the construction.
  The automaton memorises the atoms ${C_1}_{|R'}$ of the last constraint $C_1$ into its state,
  and on reading the next constraint $C_2$ the automaton checks that $unprime({C_1}_{|R'}) = {C_2}_{|R}$.
  If this holds, the automaton transits into the state that remembers ${C_2}_{|R'}$;
  if the check fails, the automaton goes into the rejecting sink state.
  And so on.
  The automaton for checking $0$-satisfiability additionally checks that ${C_0}_{|R} = \{r=s\mid r,s \in R\}$.
  The number of states is exponential in $|R|$, the number of colors is $2$, and
  in fact the so-called safety (aka looping) acceptance suffices.
\end{proof}

For the rest of this section, we focus on domain $\bbN$.

%\paragraph{Satisfiability of constraint sequences in $\bbN$}
%\label{sec:constraints-N}

\subsection{Chains characterise satisfiability of constraint sequences}
\label{sec:chain-characterisation}
In this section we prove the characterisation of satisfiable constraint sequences that we used to $\omega$-regularly approximate the automaton game over $(\bbN, \leq)$ (Section~\ref{sec:AS_N}). Recall that chains are defined in Definition~\ref{def:chains} on page~\pageref{def:chains}.

While the target characterisation relies on one-way chains, we start by presenting a characterisation using two-way chains:
such chains compare register values forwards and backwards in time.
This characterisation is intuitive and easy to prove but difficult to implement using \emph{one}-way automata.
Therefore, later we provide an alternative characterisation using \emph{one}-way chains
which read constraint sequences in forward direction only.
The lifting from two-way to one-way chains is done using Ramsey theorem~\cite{RamseyTheorem}.
A similar proof strategy is employed in~\cite[Appendix C]{ST11},
but our notion of chains is simpler, and we describe the previously missing application of Ramsey theorem.
We start with the definitions of two-way chains, then describe the characterisations in Lemmas~\ref{lem:satisf-N}~and~\ref{lem:satisf-N-1}.
\begin{lemma}\label{lem:satisf-N}
  A consistent constraint sequence is satisfiable in $\bbN$ iff
  \li
  \-[\condAT.] it has no infinite-depth two-way chains, and
  \-[\condBT.] every ceiled two-way chain has a bounded depth\\
           (i.e., there exists $\Bound \in \bbN$ such that the depth of
            every ceiled two-way chain is $\leq\Bound$).
  \il
\end{lemma}
\manufixed{R2:  Lemma 3: Give the type of B in the Lemma. In the outline of the proof, A and B occur non-primed when they should be primed. To the reader, it is also not clear that the detailed proof follows after the outline. In the detailed proof, I am wondering if there always exists in ceiled chain and where $\condBT$ has been used.}{The condition B2 was used in the first paragraph of the proof ("data value in N does not exist if A2 or B2 are violated"); in the other direction, we added that B2 is needed to show that "the largest depth exists"}
%We first outline the proof then provide the details.
%\begin{proof}[Proof sketch]
%  The left-to-right direction is trivial: if $\condAT$ is not satisfied, then
%  one needs infinitely many values below the maximal initial value of
%  a register to satisfy the sequence, which is impossible in
%  $\bbN$. Likewise, if $\condBT$ is not satisfied, then one also needs infinitely
%  many values below the value of a maximal stable chain, which
%  is again impossible. For the other direction, we show that if $\condAT$ and
%  $\condBT$ hold, then one can construct a sequence of valuations
%  $\v_0\v_1\dots$ satisfying the constraint sequence, such that for
%  all $r\in R$, the value $\v_i(r)$ is the largest depth of a (decreasing)
%  two-way chain starting in $r$ at moment $i$.
%\end{proof}

\begin{proof}
  The direction $\Impl$ is proven by contradiction:
  if $\condAT$ is not satisfied,
  then one needs infinitely many values below the maximal initial value of a register to satisfy the sequence,
  which is impossible in $\bbN$.
  Similarly for $\condBT$.
  We now state this formally.
  Suppose a constraint sequence $C_0 C_1...$ is satisfiable by some valuations $\v_0 \v_1...$.
  Towards a contradiction, assume that $\condAT$ does not hold, i.e. there is an infinite decreasing two-way chain $\chi = (r_0,m_0)(r_1,m_1)...$.
  Let $\v_{m_0}(r_0) = \d^\star$ be the data value at the start of the chain.
  Each decrease $(r_i,m_i)>(r_{i+1},m_{i+1})$ in the chain $\chi$ requires the data to decrease as well:
  $\v_i(r_i)>\v_{i+1}(r_{i+1})$, so there must be an infinite number of data values
  between $\d^\star$ and $0$, which is impossible in $\bbN$. Hence $\condAT$ must hold.
  Now consider $\condBT$.
  If there are no ceiled chains, we are done, so assume there is at least one ceiled chain.
  Then there exists a maximal stable chain, by definition.
  Let $\d^\star$ be the value of the registers in the maximal stable chain.
  All ceiled chains lie below the maximal stable chain,
  therefore the values of their registers are bounded by $\d^\star$.
  Thus the depth of each such a chain is bounded by $\Bound = \d^\star$, so $\condBT$ holds.

  The direction $\Implied$.
  Given a consistent constraint sequence $C_0 C_1...$ satisfying $\condAT$ and $\condBT$,
  we construct a sequence of register valuations $\v_0 \v_1...$ such that $\v_i \cup \v'_{i+1} \models C_i$ for all $i \geq 0$
  (recall that $\v' = \{r' \mapsto \v(r)\mid r \in R\}$).
  For a register $r$ and moment $i \in \bbN$,
  let $d(r,i)$ be the largest depth of two-way chains from $(r,i)$;
  such a number exists by assumption $\condBT$;
  it is not $\infty$ by assumption $\condAT$;
  it can be $0$.
  \emph{Then, for every $r \in R$ and $i \in \bbN$, set $\v_i(r) = d(r,i)$.}

  We now prove that for all $i$, the satisfaction $\v_i\cup\v'_{i+1} \models C_i$ holds,
  i.e.\ all atoms of $C_i$ are satisfied.
  Pick an arbitrary atom $t_1 \bowtie t_2$ of $C_i$, where $t_1,t_2 \in R\cup R'$.
  Define $m_{t_1} = i+1$ if $t_1$ is a primed register, else $m_{t_1} = i$; similarly define $m_{t_2}$.
  There are two cases.
  \li
  \- $t_1 \bowtie t_2$ is $t_1=t_2$.
     Then the deepest chains from $(t_1,m_{t_1})$ and $(t_2,m_{t_2})$ have the same depth,
     $d(t_1,m_{t_1}) = d(t_2,m_{t_2})$, and hence $\v_i \cup \v'_{i+1}$ satisfies the atom.
  \- $t_1\bowtie t_2$ is $t_1 > t_2$.
     Then, any chain $(t_2,m_{t_2})...$ from $(t_2,m_{t_2})$ can be prefixed by $(t_1,m_{t_1})$ to create the deeper chain $(t_1,m_{t_1})>(t_2,m_{t_2})...$.
     Thus, $d(t_1,m_{t_1}) > d(t_2,m_{t_2})$, therefore $\v_i\cup\v'_{i+1}$ satisfies the atom.
  \il
  This concludes the proof.
\end{proof}

\begin{remark*}
  The proof describes a data-assignment function which maps a sequence of constraints to a sequence of valuations satisfying it. Such functions are widespread, see e.g.~\cite[Lemma~C.7]{ST11}~or~\cite[Lemma 15]{CKL13}. Later in Section~\ref{sec:data-assignment-function} we describe a different kind of data-assignment function, which does not see the whole constraint sequence beforehand but only the prefix read so far. This changes how much the register values get separated from each other: from $\Bound$ in the above proof to approx.\ $2^B$.
\end{remark*}

The previous lemma characterises satisfiability in terms of two-way chains,
but our final goal is the characterisation by automata.
It is hard to design a \emph{one}-way automaton tracing \emph{two}-way chains,
so we lift the previous lemma to one-way chains.
\manufixed{R2:  Page 10, line 44: You should explain the final goal of building an automaton way earlier as an overview so the reader knows what you are building up to.}{added an outline to this subsection}
\begin{lemma}\label{lem:satisf-N-1}
  A consistent constraint sequence is satisfiable in $\bbN$ iff
  \li
  \-[\condAO.] it has no infinitely decreasing one-way chains, and
  \-[\condBO.] every ceiled one-way chain has a bounded depth\\
  (i.e., there exists $\Bound \in \bbN$ such that the depth of
  every ceiled one-way chain is $\leq\Bound$).
  \il
\end{lemma}

We describe a proof idea then provide a full proof.
\begin{proof}[Proof idea]
  We start from Lemma~\ref{lem:satisf-N} and
  show that hypotheses $\condAT$ and $\condBT$ can be refined to $\condAO$ and $\condBO$ respectively.
  From an infinite (decreasing) two-way chain, we can
  always extract an infinite decreasing one-way chain, since
  two-way chains are infinite to the right and not to the
  left. Hence, for every moment $i$, there always exists a moment $j>i$
  such that one register of the chain is smaller at step $j$ than a
  register of the chain at step $i$.
  Then, given a sequence of ceiled two-way chains of unbounded
  depth, we are able to construct a sequence of one-way chains of
  unbounded depth. This construction is more difficult than in the above case.
  Indeed, even though there are by
  hypothesis deeper and deeper ceiled two-way chains, they may start at later
  and later moments in the constraint sequence and go to the left. Thus, one
  cannot simply take an arbitrarily deep two-way chain and extract an arbitrarily deep one-way chain from it.
  However, we demonstrate, using a Ramsey argument, that it is
  still possible to extract arbitrarily deep one-way chains since the
  two-way chains are not completely independent.
\end{proof}

\begin{proof}
  Thanks to Lemma~\ref{lem:satisf-N}, it suffices to show that $\condAO \Leftrightarrow \condAT$ and $\condBO \Leftrightarrow \condBT$.
  The implications $\condAT \Impl \condAO$ and $\condBT \Impl \condBO$ follow from the definitions of chains.

  Now, let us show that $\neg \condAT \Impl \neg \condAO$: let $C_0 C_1 \dots$ be a consistent constraint sequence, and assume that it has an infinite two-way chain $\chi = (r_a,i) \dots$.
  We then construct an infinite descending one-way chain $\chi'$.
  The construction is illustrated in Figure~\ref{fig:proof:eq:A}.
  Our one-way chain $\chi'$ starts in $(r_a,i)$.
  The area on the left from $i$-timeline contains $i \cdot |R|$ points,
  but $\chi$ has an infinite depth hence at some point it must go to the right from $i$.
  Let $r_b$ be the smallest register visited at moment $i$ by $\chi$;
  we first assume that $r_b$ is different from $r_a$ (the other case is later).
  Let $\chi$ go $(r_b,i) \triangleright (r',i+1)$.
  We append this to $\chi'$ and get $\chi' = (r_a,i) > (r_b,i) \triangleright (r',i+1)$.
  If $r_a$ and $r_b$ were actually the same, so the chain $\chi$ moved $(r_a,i) \triangleright (r',i+1)$,
  then we would append only $(r_a,i) \triangleright (r',i+1)$.
  By repeating the argument from the point $(r',i+1)$,
  we construct the infinite descending one-way chain $\chi'$.
  Hence $\neg \condAO$ holds.

  \begin{figure}[tb]
    \centering
    \begin{tikzpicture}[bend angle=45, auto, node distance=2.75cm,scale=0.15,nodes={font=\footnotesize}]
      \tikzstyle{every state}=[text=black]
      %\tikzstyle{input}=[rectangle,fill=red!30,minimum size=.5cm,inner sep=0cm]
      %\tikzstyle{output}=[fill=green!30,minimum size=.5cm,inner sep=0cm]
      %\tikzstyle{input char}=[text=red]
      %\tikzstyle{output char}=[text=green]
      \tikzstyle{dot}=[circle,draw=black,fill=black,minimum size=0.8mm,inner sep=0]

      \tikzset{every edge/.append style={font=\small}} % does not seem to matter

      % 16
      % |   |   |   |   |
      % |   |   |   |   |
      % |   |   |   |   |
      % |   |   |   |   |
      % 0---4---8--12--16------28

      % canvas
      \draw[->,black!40] (0, 4) -- (0, 26);
      \draw[->,black!40] (0, 4) -- (28, 4);
      \node at (-3,26) {\footnotesize order};
      \node at (29,3) {\footnotesize time};

      \foreach \x in {1,...,6}
      {
        \draw[-,black!20] (\x*4,4)--(\x*4,24);
      }

      \node at (12,3) {$i$};
      \node at (16.2,2.93) {$i\!+\!1$};

      \node[dot,label={[label distance=-1mm]above: $r_a$}] (1) at (12,24){};
      \node[dot] (2) at (8,23){};
      \node[dot] (3) at (4,22){};
      \node[dot] (4) at (8,21){};
      \node[dot] (5) at (4,20){};
      \node[dot] (6) at (0,19){};
      \node[dot] (7) at (4,18){};
      \node[dot] (8) at (8,17){};
      \node[dot] (9) at (12,16){};

      \node[dot] (10) at (16,15){};
      \node[dot] (11) at (12,14){};
      \node[dot] (12) at (8,13){};
      \node[dot,label={[label distance=-0.8mm]below: $r_b$}] (13) at (12,12){};

      \node[dot,label={[label distance=-0.8mm]right: $r'$}] (14) at (16,12){};
      \node[dot] (15) at (16,10){};
      \node[dot] (16) at (20,9){};
      \node[dot] (17) at (24,8){};
      \node[dot] (18) at (20,7){};
      \node[dot] (19) at (24,6){};

      \draw[line width=0.5mm,opacity=0.3] plot [smooth,tension=0.5] coordinates{(1) (2) (3) (4) (5) (6) (7) (8) (9) (10) (11) (12) (13) (14)};
      \draw[line width=0.5mm,opacity=0.3] plot [smooth] coordinates{(14) (15)};
      \draw[line width=0.5mm,opacity=0.3] plot [smooth,tension=0.5] coordinates{(15) (16) (17) (18) (19)};

      % chain
      \begin{scope}[on background layer]
        \draw[-,color=blue,line width=0.6mm,opacity=0.5]
             plot [smooth] coordinates { (1)  (13) };
        %\draw[-,color=blue,line width=0.6mm,opacity=0.2,transform canvas={yshift=0.5mm}]
        %     plot [smooth] coordinates { (13)  (14)  (15)  (16)  (18) };
        \draw[-,color=blue,line width=0.6mm,opacity=0.5,transform canvas={yshift=0.3mm}]
             plot coordinates {($(13)+(0.2,0.08)$)  ($(14)+(0.2,0.08)$)  ($(15)+(0.2,0.08)$)  ($(16)+(0,0.08)$)  (18) (19)};
      \end{scope}
    \end{tikzpicture}
    \caption{Proving the direction $\neg \condAT \Impl \neg \condAO$ in Lemma~\ref{lem:satisf-N-1}.
             The two-way chain is in grey, the constructed one-way chain is in blue.}
    \label{fig:proof:eq:A}
  \end{figure}
\manufixed{R2: Figure 3: A black line connected the same register before in Figure 2 while colors marked chains. Here, the black line seems to have a different semantics?}{changed to grey}
  Now, let us show $\neg \condBT \Impl \neg \condBO$.
  Given a sequence of ceiled two-way chains of unbounded depth,
  we need to create a sequence of ceiled one-way chains of unbounded depth.
  We extract a witnessing one-way chain of a required depth from a sufficiently deep two-way chain.
  To this end, we represent the two-way chain as a clique with colored edges,
  and whose one-colored subcliques represent all one-way chains.
  We then use the Ramsey theorem that says a monochromatic subclique of a required size always exists if a clique is large enough.
  From the monochromatic subclique we extract the sought one-way chain.

  The Ramsey theorem~\cite{RamseyTheorem} is about clique graphs with colored edges.
  For the number $n \in \bbN$ of vertices,
  let $K_n$ denote the clique graph and let $E_{K_n}$ be its set of edges.
  Then, we let $color\: E_{K_n} \to \{1,\dots,\#c\}$ be an edge-coloring function,
  where $\#c$ is the number of edge colors in the clique.
  A clique is \emph{monochromatic} if all its edges have the same color ($\#c=1$).
  The Ramsey theorem says:
  \begin{quote}
    \emph{%
    Fix the number $\#c$ of edge colors.
    $(\forall n)(\exists l)(\forall color\: E_{K_l}\to \{1,\dots,\#c\})$:
    there exists a monochromatic subclique of $K_l$ with $n$ vertices.
    The number $l$ is called the Ramsey number for $(\#c,n)$.}
\end{quote}
\manufixed{R2: Page 12, line 41: You quantify 'color' but you never use it afterwards.}{added clarification after the theorem, saying that the coloring function "color" does not affect the Ramsey number.}
  I.e., for any given $n$,
  there is a sufficiently large size $l$ such that any colored clique of this size contains a monochromatic subclique of size $n$.
  Ramsey numbers depend on the number $\#c$ of colors and size $n$ of the clique
  and are independent of a coloring function $color$.
  We use the theorem with three colors only: $\#c=3$.

  \begin{figure}[tb]
    \centering

    \begin{subfigure}[t]{.45\textwidth}
      \begin{tikzpicture}[->, >=stealth', bend angle=45, auto, node distance=2.75cm,xscale=1.5,scale=0.1,nodes={font=\footnotesize}]
        \tikzstyle{every state}=[text=black]
        \tikzstyle{dot}=[circle,draw=black,fill=black,minimum size=0.8mm,inner sep=0mm]

        \tikzset{every edge/.append style={font=\small}} % does not seem to matter

        % 16
        % |   |   |   |   |
        % |   |   |   |   |
        % |   |   |   |   |
        % |   |   |   |   |
        % 0---4---8--12--16------28

        % canvas
        \draw[->,black!40] (2, -2) -- (2, 28);
        \draw[->,black!40] (2, -2) -- (20, -2);
        \node at (-2,28) {\footnotesize order};
        \node at (22,-3) {\footnotesize time};

        \foreach \x in {1,...,4}
        {
          \draw[-,black!15] (\x*4,-2)--(\x*4,28);
        }

        % the 2-chain:

        \node[dot,label={[label distance=-1mm]left:\scriptsize 1}] (a) at (16,28){};
        \node[dot,label={[label distance=-1mm]left:\scriptsize 2}] (b) at (12,24){};
        \node[dot,label={[label distance=-1mm]left:\scriptsize 3}] (c) at (16,20){};
        \node[dot,label={[label distance=-1mm]left:\scriptsize 4}] (d) at (12,16){};
        \node[dot,label={[label distance=-1mm]left:\scriptsize 5}] (e) at (8,12){};
        \node[dot,label={[label distance=-1mm]left:\scriptsize 6}] (f) at (12,8){};
        \node[dot,label={[label distance=-1mm]left:\scriptsize 7}] (g) at (8,4){};
        \node[dot,label={[label distance=-1mm]left:\scriptsize 8}] (h) at (4,0){};

        \draw[-,line width=0.5mm,opacity=0.3] (a) -- (b) -- (c) -- (d) -- (e) -- (f) -- (g) -- (h);

      \end{tikzpicture}
      \caption{A given two-way chain (wo stuttering)}
    \end{subfigure}
    \begin{subfigure}[t]{.45\textwidth}
      \setcounter{subfigure}{3}
      \begin{tikzpicture}[->, >=stealth', bend angle=45, auto, node distance=2.75cm,xscale=1.5,scale=0.1,nodes={font=\footnotesize}]
        \tikzstyle{every state}=[text=black]
        \tikzstyle{dot}=[circle,draw=black,fill=black,minimum size=0.8mm,inner sep=0]

        \tikzset{every edge/.append style={font=\small}} % does not seem to matter

        % 16
        % |   |   |   |   |
        % |   |   |   |   |
        % |   |   |   |   |
        % |   |   |   |   |
        % 0---4---8--12--16------28

        % canvas
        \draw[->,black!40] (2, -2) -- (2, 28);
        \draw[->,black!40] (2, -2) -- (20, -2);
        \node at (-2,28) {\footnotesize order};
        \node at (22,-3) {\footnotesize time};

        \foreach \x in {1,...,4}
        {
          \draw[-,black!15] (\x*4,-2)--(\x*4,28);
        }

        % the chain:
        \node[dot,label={[label distance=-1mm]left:\scriptsize 1}] (a) at (16,28){};
        \node[dot,label={[label distance=-1mm]left:\scriptsize 2}] (b) at (12,24){};
        \node[dot,label={[label distance=-1mm]left:\scriptsize 3}] (c) at (16,20){};
        \node[dot,label={[label distance=-1mm]left:\scriptsize 4}] (d) at (12,16){};
        \node[dot,label={[label distance=-1mm]left:\scriptsize 5}] (e) at (8,12){};
        \node[dot,label={[label distance=-1mm]left:\scriptsize 6}] (f) at (12,8){};
        \node[dot,label={[label distance=-1mm]left:\scriptsize 7}] (g) at (8,4){};
        \node[dot,label={[label distance=-1mm]left:\scriptsize 8}] (h) at (4,0){};

        \draw[-,line width=0.5mm,opacity=0.3] (a) -- (b) -- (c) -- (d) -- (e) -- (f) -- (g) -- (h);

        \begin{scope}[on background layer]
          \draw[-,color=red,line width=0.6mm,opacity=0.5] (a) -- (b) -- (e) -- (h);
        \end{scope}
      \end{tikzpicture}
      \caption{Constructed increasing one-way chain}
      \label{fig:proof:equiv:B:1-chain}
    \end{subfigure}%
    \\
    \begin{subfigure}[t]{.42\textwidth}
      \setcounter{subfigure}{1}
      \begin{tikzpicture}[-, >=stealth', bend angle=45, auto, node distance=2.75cm,xscale=1.5,scale=0.1,nodes={font=\footnotesize}]
        \tikzstyle{every state}=[text=black]
        \tikzstyle{dot}=[circle,draw=black,fill=black,minimum size=0.8mm,inner sep=0]

        \tikzset{every edge/.append style={font=\small}} % does not seem to matter

        % canvas
        \draw[->,black!40] (2, -2) -- (2, 28);
        \draw[->,black!40] (2, -2) -- (20, -2);
        \node at (-2,28) {\footnotesize order};
        \node at (22,-3) {\footnotesize time};

        \foreach \x in {1,...,4}
        {
          \draw[-,black!15] (\x*4,-2)--(\x*4,28);
        }

        \node[dot] (a) at (16,28){};
        \node[dot] (b) at (12,24){};
        \node[dot] (c) at (16,20){};
        \node[dot] (d) at (12,16){};
        \node[dot] (e) at (8,12){};
        \node[dot,draw=black!40,fill=black!40] (f) at (12,8){};
        \node[dot,draw=black!40,fill=black!40] (g) at (8,4){};
        \node[dot,draw=black!40,fill=black!40] (h) at (4,0){};

        \draw[red] (a) -- (b) -- (e);
        \draw[red] (c) -- (d) -- (e);
        \draw[red] (a) -- (e);
        \draw[red] (a) -- (d);

        \draw[blue,semithick] (b) -- (c);
        \draw[green,semithick] (a) -- (c);
        \draw[green,semithick] (b) -- (d);
        \draw[red] (c) to[bend left=15] (e);

      \end{tikzpicture}
      \caption{Clique: shown the edges for the top $5$ points only. Try completing the rest.}
      \label{fig:proof:equiv:B:clique}
    \end{subfigure}
    \hspace{3mm}
    \begin{subfigure}[t]{.46\textwidth}
      \begin{tikzpicture}[-, >=stealth', bend angle=45, auto, node distance=2.75cm,xscale=1.5,scale=0.1,nodes={font=\footnotesize}]
        \tikzstyle{every state}=[text=black]
        \tikzstyle{dot}=[circle,draw=black,fill=black,minimum size=0.8mm,inner sep=0]

        \tikzset{every edge/.append style={font=\small}} % does not seem to matter

        % canvas
        \draw[->,black!40] (2, -2) -- (2, 28);
        \draw[->,black!40] (2, -2) -- (20, -2);
        \node at (-2,28) {\footnotesize order};
        \node at (22,-3) {\footnotesize time};

        \foreach \x in {1,...,4}
        {
          \draw[-,black!15] (\x*4,-2)--(\x*4,28);
        }

        \node[dot,label={[label distance=-1mm]left:\scriptsize 1}] (a) at (16,28){};
        \node[dot,label={[label distance=-1mm]left:\scriptsize 2}] (b) at (12,24){};
        \node[dot,label={[label distance=-1mm]left:\scriptsize 3}] (c) at (16,20){};
        \node[dot,label={[label distance=-1mm]left:\scriptsize 4}] (d) at (12,16){};
        \node[dot,label={[label distance=-1mm]left:\scriptsize 5}] (e) at (8,12){};
        \node[dot,label={[label distance=-1mm]left:\scriptsize 6}] (f) at (12,8){};
        \node[dot,label={[label distance=-1mm]left:\scriptsize 7}] (g) at (8,4){};
        \node[dot,label={[label distance=-1mm]left:\scriptsize 8}] (h) at (4,0){};

        \draw[red] (a) -- (b) -- (e);
        \draw[red] (a) -- (e);

        \draw[red] (e) -- (h);
        \draw[red] (b) to [bend right=15] (h);
        \draw[red] (a) to (h);

      \end{tikzpicture}
      \caption{Monochromatic subclique with elements $1$, $2$, $5$, $8$}
      \label{fig:proof:equiv:B:subclique}
    \end{subfigure}
    \caption{Proving the direction $\neg \condBT \Impl \neg \condBO$ in Lemma~\ref{lem:satisf-N-1}}
    \label{fig:proof:equiv:B}
  \end{figure}

  Given a sequence of two-way chains of unbounded depth,
  we show how to build a sequence of one-way chains of unbounded depth.
  Suppose we want to build a one-way chain of depth $n$,
  and let $l$ be the Ramsey number for $(3,n)$.
  Since the two-way chains from the sequence have unbounded depth,
  there is a two-way chain $\chi$ of depth $l$.
  From it we construct the following colored clique (the construction is illustrated in Figure~\ref{fig:proof:equiv:B}).
  \li
  \- Remove stuttering elements from $\chi$:\manufixed{Page 12, line 56: After removing stuttering elements from X the resulting sequence might skip points in time, correct? What are the implications of this?}{The main implication is that the resulting constructed sequence does not necessarily satisfy the definition of one-way chains, since it may skip time points. However, that is not an issue, as we can use removed points to fill in the missing gaps.}
     whenever $(r_i,m_i)=(r_{i+1},m_{i+1})$ appears in $\chi$, remove $(r_{i+1},m_{i+1})$.
     We repeat this until no stuttering elements appear.
     Let $\chi_> = (r_1,m_1) > \dots > (r_l,m_l)$ be the resulting sequence;
     it is strictly decreasing, and contains $l$ pairs (the same as the depth of the original $\chi$).
     Note the following property $(\dagger)$:
     for every not necessarily adjacent $(r_i,m_i)>(r_j,m_j)$,
     there is a one-way chain $(r_i,m_i)\dots(r_j,m_j)$;
     it is decreasing if $m_i<m_j$, and increasing otherwise;
     its depth is at least $1$.
     The resulting sequence may skip points in time, but this -- as will be explained later -- does not affect the construction.

  \- The elements $(r,m)$ of $\chi_>$ serve as the vertices of the colored clique.
     The edge-coloring function is: for every not necessarily adjacent $(r_a,m_a)>(r_b,m_b)$ in $\chi_>$,
     let
     $color\big((r_a,m_a), (r_b,m_b)\big)$ be
     ${\color{red}\nearrow}$ if $m_a < m_b$,
     ${\color{blue}\searrow}$ if $m_a > m_b$,
     ${\color{green}\,\downarrow}$ if $m_a = m_b$.
     Thus, we assign a color to an edge between every two vertices.
     Figure~\ref{fig:proof:equiv:B:clique} gives an example.
     %$color\big((r_a,m_a), (r_b,m_b)\big) =
     %  \begin{cases}
     %    {\color{red}\nearrow}     \text{ if } m_a < m_b,\\
     %    {\color{blue}\searrow}    \text{ if } m_a > m_b,\\
     %    {\color{green}\,\downarrow} \text{~ if } m_a = m_b.
     %  \end{cases}$
  \il
  By applying the Ramsey theorem,
  we get a monochromatic subclique of size $n$
  with vertices $V \subseteq \{(r_1,m_1),\dots,(r_l,m_l)\}$.
  Its color cannot be $\color{green}\downarrow$ when $n > |R|$,
  because a timeline has maximum $|R|$ points.
  Suppose the subclique's color is ${\color{red}\nearrow}$ (the case of $\color{blue}\searrow$ is similar).
  We build the increasing sequence $\chi^\star = (r_1^\star,m_1^\star) < \dots < (r_n^\star,m_n^\star)$,
  where $m_i^\star < m_{i+1}^\star$ and $(r_i^\star,m_i^\star) \in V$
  for every $i$.
  The sequence $\chi^\star$ may not satisfy the definition of one-way chains,
  because the removal of stuttering elements that we performed at the beginning can cause time jumps i.e.\ $m_{i+1} > m_i+1$.
  But it is easy---relying on the property $(\dagger)$---%
  to construct the one-way chain $\chi^{\star\star}$ of depth $n$ from $\chi^\star$
  by inserting the necessary elements between $(r_i,m_i)$ and $(r_{i+1},m_{i+1})$.
  The case when the subclique has color $\color{blue}\searrow$, the resulting constructed chain is decreasing.
\manufixed{R2: Page 13 come online 15: Is the 'color' assignment transitive, i.e., do we assign a color two nonadjacent elements of $X_>$?}{the color-assignment function assigns colors for edges between all elements of X>. We added a clarification.}

  Thus, for every given $n$, we constructed either a decreasing or increasing ceiled one-way chain of depth $n$. In other words, a sequence of such chains of unbounded depth.
  Hence $\neg \condBO$ holds, which concludes the proof.
\end{proof}

The next easy lemma (first stated on page~\pageref{lem:0-satisf-N-1}) refines the characterisation to $0$-satisfiability:
\lemzsatisfNo*
\begin{proof}
  \label{pf:lem:0-satisf-N-1}
  Direction $\Impl$.
  The first two items follow from Lemma~\ref{lem:satisf-N-1}; the third one follows from the definition of satisfiability.
  Consider the last item: suppose there is such a chain.
  Then, at the moment when the chain strictly decreases and goes to some register $s$,
  the register $s$ would need to have a value below $0$, which is impossible in $\bbN$.

  Direction $\Implied$.
  The first two items are exactly $\condAO$ and $\condBO$ from Lemma~\ref{lem:satisf-N-1}, so the sequence is satisfiable,
  hence it also satisfies the conditions $\condAT$ and $\condBT$ from Lemma~\ref{lem:satisf-N}.
  In the proof of Lemma~\ref{lem:satisf-N},
  we showed that in this case the following valuations $\v_0 \v_1...$ satisfy the sequence:
  for every $r \in R$ and moment $i \in \bbN$,
  set $\v_i(r)$ (the value of $r$ at moment $i$) to the largest depth of the two-way chains starting in $(r,i)$.
  We construct $\v_0 \v_1...$ as above, and get a witness of satisfaction of our constraint sequence.
  Note that at moment $0$, $\v_0 = 0^R$, by the last item. Hence the constraint sequence is $0$-satisfiable.
\end{proof}

\paragraph{Action words and constraint sequences}
In this section, we provide the proof of the following lemma, stated on page~\pageref{lem:mapping-constr}:
\lemMappingConstr*
\begin{proof}\label{page:constr-description}
  Given $\pi$, $\tst$, $\asgn$, we define the mapping $constr : (\pi,\tst,\asgn) \mapsto C$ as follows.
  The definition is as expected, but we should be careful about handling of $r_d$, it is the last item.
  \li
  \- The constraint $C$ includes all atoms of the state constraint $\pi$ (that relates the registers at the beginning of the step).
  \- Recall that neither $\tst$ nor $\asgn$ talk about $r_d$.
     For readability, we shorten $(t_1 \bowtie t_2) \in C$ to simply $t_1 \bowtie t_2$,
     $(* \bowtie r) \in \tst$ to $* \bowtie r$,
     and $a \leq b$ means $(a<b) \lor (a=b)$.

  \- We define the order at the end of the step as follows. For every two different $r,s \in R$:
     \li
     \- $r' = s'$ iff
          $(r = s) \land r,s\not\in\asgn$ or
          $r \in \asgn \land (*=s)$ or
          $r,s\in\asgn$;

     \- $r' < s'$ iff
          $(r<s) \land r,s\not\in\asgn$ or
          $(*<s) \land r \in \asgn \land s \not\in\asgn$;

     \- $r' = r_d'$ iff $(r=*)$ or $r \in \asgn$;

     \- $r' \bowtie r_d'$ iff
          $(r\bowtie *) \land r\not\in\asgn$, for $\bowtie\, \in \{<,>\}$;
     \il
  \- So far we have defined the order of the registers at the beginning and the end of the step.
     Now we relate the values between these two moments. For every $r \in R$:
     \li
     \- $r = r'$ iff $r \not\in \asgn$ or $r \in \asgn \land (*=r)$;
     \- $r \bowtie r'$ iff $r \in \asgn \land (r \bowtie *)$, for $\bowtie\, \in \{<,>\}$;
     \il
  \- Finally, we relate the values of $r_d$ between the moments.
     There are two cases.
     \li
     \- The value of $r_d$ crosses another register:
        $\exists r \in R\: (r_d<r) \land (* \geq r)$.
        Then $(r_d'>r_d)$.
        Similarly for the opposite direction: if $\exists r \in R\: (r_d>r) \land (* \leq r)$ then $(r_d'<r_d)$.

     \- Otherwise, the value of $r_d$ does not cross any register boundary.
        Then $r_d' = r_d$.
     \il
  \il

  Using the mapping $constr$,
  every action word $\overline{a} = (\tst_0\asgn_0)(\tst_1\asgn_1)\dots$ can be uniquely mapped to
  the constraint sequence $C_0 C_1 \dots$ as follows:
  $C_0 = constr(\pi_0,\tst_0,\asgn_0)$,
  set $\pi_1 = unprime({C_0}_{|R_d'})$,
  then $C_1 = constr(\pi_1,\tst_1,\asgn_1)$,
  and so on.

  We now prove that an action word is feasible iff the constructed constraint sequence is $0$-satisfiable.
  This follows from the definitions of feasibility and $0$-satisfiability, and from the following simple property of feasible action words.
  Every feasible action word has a witness $\v_0\d_0\v_1\d_1\dots \in (\D^R\cdot\D)^\omega$ such that:
  if some $\tst$ is repeated twice and no assignment is done, then the value $\d$ stays the same.
  This property is needed due to the last item in the definition of $constr$ where we set $r'_d = r_d$.
\end{proof}

\subsection{Max-automata recognise satisfiable constraint sequences}
\label{sec:max-automata-characterisation}
%We
%now
%state
%the main result about recognisability of satisfiable constraint sequences by \emph{max-automata}~\cite{B11}.

This section presents an automaton characterisation of constraint sequences satisfiable in $\bbN$.
The automaton construction verifies the conditions on one-way chains stated in Lemma~\ref{lem:satisf-N-1}:
the absence of (\condAO) \emph{infinite} decreasing one-way chains
and of (\condBO) \emph{unbounded} one-way ceiled chains.
The boundedness requirement of the second condition cannot be checked by $\omega$-regular automata%
\footnote{For a formal statement, see~\cite[Theorem~4.3]{ST11}
saying that the class of languages of finite-alphabet projections of ``constraint automata'' and
the class of $\omega$B-languages coincide.},
and for that reason in~\cite{ST11} the authors used nondeterministic $\omega$B-automata.
Since nondeterminism is usually hard to handle in synthesis,
we picked deterministic max-automata~\cite{B11},
which are incomparable with $\omega$B-automata, expressivity-wise.
We now define max-automata and then present the characterisation.

\emph{Deterministic max-automata} extend classic finite-alphabet parity automata with a finite set of
counters $c_1,\dots,c_n$ which can be incremented, reset to $0$, or updated by taking the
maximal value of a set of counters, but the counters cannot be tested.
On reading a word, the automaton builds a sequence of counter valuations.
The acceptance condition is given as a conjunction of the parity acceptance condition
and a Boolean combination of conditions
``counter $c_i$ is bounded along the run''.
Such a condition on a counter is satisfied by a run if
there exists a bound $\Bound\in\bbN$
such that counter $c_i$ has value at most $\Bound$ along the run.
By using negation, conditions such as ``$c_i$ is unbounded along the run'' can also be expressed.
A run is accepting if it satisfies the parity condition and the Boolean formula on the counter conditions.
Deterministic max-automata are strictly more expressive than $\omega$-regular automata.
For instance,
they can express the non-$\omega$-regular language of words of the form $a^{n_1}b a^{n_2} b \dots$ such that
$n_i\leq \Bound$ for all $i\geq 0$, for some $\Bound\in\bbN$ that can vary from word to word.
A max-automaton recognising the language is in Figure~\ref{fig:max-atm-example}.

%\begin{figure}[tb]
%  \begin{minipage}[c]{0.75\textwidth}
%    \caption{Max-automaton recognising the non-$\omega$-regular language $\{a^{n_1}b a^{n_2} b \ldots \mid \exists \Bound\in\bbN\ \forall i\:n_i < \Bound \}$. It uses a single counter $c$, the acceptance condition is ``counter $c$ is bounded'', and the parity acceptance is trivial (always accept).}
%    \label{fig:max-atm-example}
%  \end{minipage}
%  \begin{minipage}[r]{0.2\textwidth}
%    ~~~~~~~\begin{tikzpicture}[initial text={}]
%       \node[initial,state,minimum size=5mm] (q_0) {}; 
%       \path[->]
%        (q_0) edge [loop above] node {\scriptsize $a: \textit{increase}~c$} ()
%              edge [loop below] node {\scriptsize $b: \textit{reset}~c$} ();
%    \end{tikzpicture}
%  \end{minipage}
%\end{figure}

\begin{figure}[tb]
  \centering
    \begin{tikzpicture}[initial text={}]
       \node[initial,state,minimum size=5mm] (q_0) {}; 
       \path[->]
        (q_0) edge [loop above] node {\scriptsize $a: \textit{increase}~c$} ()
              edge [loop below] node {\scriptsize $b: \textit{reset}~c$} ();
    \end{tikzpicture}
    \caption{Max-automaton recognising $\{a^{n_1}b a^{n_2} b \ldots \mid \exists \Bound\in\bbN\ \forall i\:n_i \leq \Bound \}$. It uses a single counter $c$, the acceptance condition is ``counter $c$ is bounded'', and the parity acceptance is trivial (always accept). The operation max is not used.}
    \label{fig:max-atm-example}
\end{figure}

We now prove the main result of this section.
\begin{theorem}\label{thm:0-satisf-N}
  For every $R$, there is a deterministic max-automaton accepting exactly all constraint sequences satisfiable in $\bbN$.
  The number of states is exponential in $|R|$, the number of counters is $O(|R|^2)$, and the number of priorities is polynomial in $|R|$.
  The same holds for $0$-satisfiability in $\bbN$.
\end{theorem}
\manufixed{R2:  Theorem 6: Give intuition on how max-automata are used.}{added a clarification into the proof idea that the max operation is used to ensure that we track the largest chains only}
\begin{proof}[Proof idea]
  We design a deterministic max-automaton that checks conditions $\condAO$ and $\condBO$ of Lemma~\ref{lem:satisf-N-1}.
  Condition $\condAO$, namely the absence of infinitely decreasing one-way chains, is checked as follows.
  We construct a nondeterministic B\"uchi automaton that guesses a chain and verifies that it is infinitely decreasing, i.e. that `$>$' occurs infinitely often and that there is no `$<$' (only `$>$' and `$=$').
  Determinising and complementing yields a deterministic parity automaton,
  that can be disjuncted through a synchronised product with the deterministic max-automaton checking condition $\condBO$.
  \manufixed{R2: Page 14, line 14: Give some intuition why omega-regular automata don't suffice.}{added in the intro to the section; to the description of $A_\text{\sc b}$ we added a short introduction}
  The latter condition (the absence of ceiled one-way chains of unbounded depth)
  is more involved.
  We design a master automaton that tracks every chain $\chi$ that currently exhibits a stable behaviour.
  To every such a chain $\chi$, the master automaton assigns a tracer automaton whose task is to ensure the absence of unbounded-depth ceiled chains below $\chi$.
  For that, the tracers use $2|R|$ counters -- one for tracking increasing and one for tracking decreasing chains -- and requires them to be bounded.
  We use the max operation on counters to ensure that we trace the largest chains only.
  The overall acceptance condition ensures that if the chain $\chi$ is stable,
  then there are no ceiled chains below $\chi$ of unbounded depth.
  %Since the master automaton tracks \emph{every} such potential chain, we are done.
  Finally, we take the product of all these automata,
  which preserves determinism.
\end{proof}

In the next section, we provide the details of the proof.

\paragraph{Proof of Theorem~\ref{thm:0-satisf-N}}
  \newcommand{\pargray}[1]{\smallskip\noindent{\sethlcolor{lightgray}\hl{#1}\sethlcolor{yellow}}\,}
  We describe a max-automaton $A$ that accepts a constraint sequence
  iff it is consistent and has no infinitely decreasing one-way chains and no ceiled one-way chains of unbounded depth.
  By Lemma~\ref{lem:satisf-N-1}, such a sequence is satisfiable.

  The automaton has three components $A = A_c \land A_{\neg\infty} \land A_\text{\sc b}$.
  %that can be described as follows:
  %a constraint sequence is accepted iff
  %it is consistent ($A_c$),
  %has no infinitely descending chains ($A_{\neg \infty}$),
  %either has no stabilising threads ($A_{\neg s}$)
  %or one of the registers is maximal ($A_m^r$) and there are no unbounded $r$-ceiled chains ($A_\text{\sc b}^r$).

  \pargray{$A_c$}
  The parity automaton $A_c$ checks consistency,
  i.e.\ that $\forall i\: unprime({C_i}_{|R'}) = (C_{i+1})_{|R}$.
  It has exponential in $|R|$ number of states and two priorities (the safety language).

  \pargray{$A_{\neg \infty}$}
  The parity automaton $A_{\neg\infty}$ ensures there are no infinitely decreasing one-way chains.
  First, we construct its negation, an automaton that accepts a constraint sequence iff it has such a chain.
  Intuitively, the automaton guesses such a chain and then verifies that the guess is correct.
  It loops in the initial state $\qinit$ until it nondeterministically decides that now is the starting moment of the chain and guesses the first register $r_0$ of the chain,
  and transits into the next state while memorising $r_0$.
  When the automaton is in a state with $r$ and reads a constraint $C$,
  it guesses the next register $r_n$,
  verifies that $(r_n' > r) \in C$ or $(r_n' = r) \in C$,
  and transits into the state that remembers $r_n$.
  The B\"uchi acceptance condition ensures that the automaton leaves the initial state
  and transits from some $r$ to some $r_n$ with $(r_n'>r) \in C$ infinitely often.
  Determinising and complementing this automaton gives $A_{\neg\infty}$.
  The number of states is exponential and the number of priorities is polynomial in $|R|$,
  due to the determinisation.

  % some command for what follows
  \newcommand{\Tr}{\mathit{Tr}}
  \renewcommand{\idle}{\mathit{idle}}
  \newcommand{\GetTr}{\mathit{getTr}}
  \newcommand{\GetCn}{\mathit{getCn}}
  \newcommand\Cn{\mathit{Cn}}

  \newcommand\cidle{{\tt   idle}\xspace}
  \newcommand\cstart{{\tt  start}\xspace}
  \newcommand\cmove{{\tt   move}\xspace}
  \newcommand\creset{{\tt  reset}\xspace}

  \pargray{$A_\text{\sc b}$}
  The max-automaton $A_\text{\sc b}$ ensures that all ceiled one-way chains have bounded depth.
  It relies on the \emph{master} automaton controlling the team of $|R|$ \emph{chain tracers} $\Tr = \{tr_1, ..., tr_{|R|}\}$.
  Each tracer $tr$ is equipped with a counter $\idle_{tr}$ and a set $\Cn_{tr}$ of $2|R|$ of counters,
  thus overall there are $|R|(2|R|+1)$ counters.
  The construction ensures that
  every stable chain is tracked by a single tracer $tr$ and its counter $\idle_{tr}$ is bounded;
  and vice versa, if a tracer $tr$ has its counter $\idle_{tr}$ bounded, it tracks a stable chain.
  Suppose for a moment that tracer $tr$ tracks a stable chain $\chi$.
  Then the goal of counters $\Cn_{tr}$ is to track the deepest increasing and decreasing chains below $\chi$.
  Since there are only $|R|$ registers,
  it suffices to track $|R|$ decreasing chains,
  every chain ending in a different register (similarly for increasing chains).
  This is because there is no need to track two decreasing chains ending in the same register:
  once the two chains ``meet'' in a register $r$,
  we continue tracking only the one with the larger depth and forget about the other.
  We use the max operation of automata to implement this idea.
  Overall, the construction ensures that the counters in $\Cn_{tr}$ are bounded iff
  the increasing and decreasing chains ceiled by the stable chain tracked by the tracer $tr$ have bounded depths.
  The acceptance of $A_\text{\sc b}$ is the formula
  $$
  \bigwedge_{tr \in \Tr} \big(\idle_{tr}\text{ is bounded }\!\impl \!\bigwedge_{c \in \Cn_{tr}}\!\!\! c \text{ is bounded}\big).$$
  \label{eq:Ab-acceptance}

  The work of tracers is controlled by the master automaton via four commands
  \cidle (``track nothing''),
  \cstart (``start tracking a potentially stable chain''),
  \cmove (``continue tracking''), and
  \creset (``stop tracking'').
  Before we formally describe the master and the tracers,
  we define the concept of ``levels'' used in the presentation.
  Intuitively, the levels abstract concrete data values,
  and the tracers actually track the levels instead of specific registers.

  \manufixed{R2: Page 14, line 59: Why do we need to define 'levels'? Give the intuition behind their purpose as motivation.}{Added intuition (`level' abstracts a concrete data value)}
  \label{def:levels}
  Fix a constraint $C$.
  A \emph{level} $l \subseteq R \,{\setminus} \{\emptyset\}$ is an equivalence class of registers
  wrt.\ $C_{|R}$ or wrt.\ $unprime(C_{|R'})$.
  Thus, in the constraint $C$ we distinguish the levels of two kinds:
  \emph{start levels} (at the beginning of the step) and \emph{end levels} (at the end of the step).
  A start level $l\subseteq R$ \emph{disappears} when
  $C$ contains no atoms of the form $r=s'$ for $r \in l$ and $s \in R$;
  this means that a data value abstracted by the level disappears from the registers.
  An end level $l\subseteq R$ is \emph{new} if
  $C$ contains no atoms of the form $r=s'$ where $r \in R$ and $s \in l$;
  intuitively, the constraint requires a new data value to appear in registers $l$.
  A start level $l$ \emph{morphs} into an end level $l'$ if
  $C$ contains an atom $r=s'$ for some $r \in l$ and $s \in l'$;
  i.e., the constraint requires the registers in $l'$ to hold the data value previously held by the registers in $l$.
  Notice that there can be at most $|R|$ start and $|R|$ end levels, for a fixed constraint $C$.
  Figure~\ref{fig:levels} illustrates the definitions.
  We are now ready to describe the master and the tracers.

  \begin{figure}[t]
    \begin{minipage}[c]{0.75\textwidth}
      \caption{Example of levels:
        start levels are $\{r_1,r_2\}$ and $\{r_3\}$,
        end levels are $\{r_3\}$, $\{r_2\}$, and $\{r_1\}$.
        The start level $\{r_1,r_2\}$ morphs into end level $\{r_3\}$,
        the start level $\{r_3\}$ disappears,
        and two new end levels appear, $\{r_1\}$ and $\{r_2\}$.
        The constraint is $\{r_1=r_2=r'_3>r'_2>r_3>r'_1\}$.%
        }
      \label{fig:levels}
    \end{minipage}\hspace{0.6cm}
    \begin{minipage}[c]{0.18\textwidth}
      \vspace{-4mm}
      \centering
      \includegraphics[width=1\textwidth]{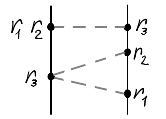}
    \end{minipage}
  \end{figure}

  \parit{Master}
  States of $A_\text{\sc b}$ are of the form $(\GetTr, \vec{q})$,
  where the partial mapping $\GetTr:l \mapsto tr$ maps a level $l \subseteq R\,{\setminus}\{\emptyset\}$ to a tracer $tr \in \Tr$,
  and $\vec{q} = (q_1,...,q_{|\Tr|})$ describes the states of individual tracers.
  The master updates the state component $\GetTr$ while the tracers update their states.
  Initially, there is only one start level $R$ (assuming the registers start with the same value),
  so we define $\GetTr = \{R \mapsto tr_1\}$.
  Suppose the automaton reads a constraint $C$,
  let $L$ and $L'$ be the start and end levels of $C$,
  and suppose the automaton is in state $(\GetTr, \vec{q})$ and $\GetTr : L \to \Tr$.
  We define the successor state $(\GetTr', \vec{q}\,')$, where $\GetTr': L' \to \Tr$,
  and operations on the counters using the following procedure.
  \li
  \- To every tracer $tr$ that does not currently track a level,
     i.e.\ $tr \in \Tr\setminus\GetTr(L)$,
     the master commands \cidle (causing the tracer to increment $\idle_{tr}$).

  \- For every start level $l \in L$ that morphs into $l' \in L'$:~ let $tr = \GetTr(l)$, then
     \li
     \- the master sends $\cmove(r_\top)$ to $tr$ where $r_\top \in l$ is chosen arbitrary;
        this will cause the tracer $tr$ to update its counters $\Cn_{tr}$ and move into a successor state $q'_{tr}$;
        the register $r_\top$ will be used as a descriptor of a stable chain tracked by $tr$.
     \- we set $\GetTr'(l') = \GetTr(l)$, thus the tracer continues to track it.
     \il

  \- For every start level $l \in L$ that disappears: let $tr=\GetTr(l)$, then
     \li
     \- the master sends \creset to $tr$,
        which causes the reset of the counters in $\Cn_{tr}$
        and the increment of $\idle_{tr}$.
     \il

  \- For every new end level $l' \in L'$:
     \li
     \- we take an arbitrary $tr$ that is not yet mapped by $\GetTr'$ and map $\GetTr'(l') = tr$;
     \- the master sends \cstart to $tr$.
     \il
     \il

     \manufixed{R2:  Page 15, line 34: Give this overview earlier.}{moved this overview to an earlier place}
  \parit{Tracers}
  We now describe the tracer component.
  Its goal is to trace the depths of ceiled chains.
  When the counters of a tracer are bounded,
  the depths of the chains it tracks are also bounded.
  The tracer consists of two components, $B_{\ssearrow}$ and $B_{\nnearrow}$,
  which track decreasing and increasing chains.
  We only describe $B_\ssearrow$, the other one is similar.

  \newcommand\Bd{B_\ssearrow}
  The component $\Bd$ has a set $\Cn\cup\{\idle\}$ of $|R|+1$ counters.
  A state of $\Bd$ is either the initial state $\qinit$
  or a partial mapping $\GetCn : R \pto \Cn$.
  %We write $cn(R)$ to denote the counters used by the mapping.
\manufixed{R2: Page 15, line 48: This does not make sense to the reader
  yet. Put the register r into the move command and explain its use
  below.}{Here, the register r is different from the register passed by the move command. For clarity, we explicitly parameterised the move command with a register called $r_\top$.}
\manufixed{R2: Page 15, line 49: Introduce getTr'(l') before using it.}{done: added the definition of getTr'}
  Intuitively, in each $\GetCn$-state,
  for each register $r$ mapped by $\GetCn$,
  the value of the counter $\GetCn(r)$ reflects the depth of the deepest ceiled decreasing one-way chain ending in $r$.
  When several chains end in $r$, the counter gets the maximal value of the depths.
  We maintain this property of $\GetCn$ during the transition of $\Bd$ on reading a constraint $C$,
  using operations of max-automata on counters and register-order information from $C$.
  %The transition also uses the register $r_\top$ passed by the master automaton,
  %and assumes that the level of $r_\top$ does not disappear in the current constraint.
  The component $\Bd$ does the following:
  \li
  \- If the master's command is \cidle,
     then increment the counter $\idle$ and stay in $\qinit$.

  \- If the master's command is \creset,
     reset all counters in $\Cn$,
     increment the counter $\idle$,
     and go into state $\qinit$.

  \- If the master's command is \cstart,
     move from state $\qinit$ into the state with the empty mapping $\GetCn$.
  \il
  Otherwise, the master's command is $\cmove(r_\top)$,
  for some $r_\top \in R$ passed by the master and serving as a descriptor of a stable chain traced by the current tracer.
  The tracer performs the operations on its counters and updates the mapping $\GetCn$ as follows.
  \li
  \- Release counters.
     For every $r$ such that $r<r_\top<r'$,
     the component resets the counter $\GetCn(r)$ and removes $r$ from the mapping $\GetCn$.
     I.e.,
     we stop tracking chains ending in register $r$
     since such chains are no longer below the stable chain assigned to the tracer.

  \- Allocate counters.
     For every $r$ such that $r \geq r_\top > r'$:
     pick a counter $c \in \Cn\setminus \GetCn(R)$ and map $\GetCn(r) = c$.
     I.e., we start tracking chains ending in $r$.

  \- Update counters.
     For every $r$ such that $r \leq r_\top$ and $r' < r_\top$ do the following.
     Let $R_{>r'} = \{r_o \mid r' < r_o < r_\top\}$
     be the registers larger than the updated $r$ but below $r_\top$,
     and let $\GetCn(R_{>r'})$ be the associated counters.
     Let $r_=$ be a register s.t.\ $r_= = r'$ (may not exist).
     We update the counter $\GetCn(r)$ depending on the case:
     \li
     \- $R_{>r'}$ is empty and $r_=$ does not exist:
        the condition means that no decreasing ceiled chain can be extended into $r'$.
        Then we $\mathit{reset}$ the counter $\GetCn(r)$.

     \- $R_{>r'}$ is empty and $r_=$ exists:
        only the chains ending in $r_=$ can be extended into $r'$, and since $r_= = r'$,
        the deepest chain keeps its depth.
        Therefore, we $copy(\GetCn(r_=))$ into the counter $\GetCn(r)$.

     \- $R_{>r'}$ is not empty and $r_=$ does not exist:
        the chains from registers in $R_{>r'}$ can be extended into $r'$,
        and since $r'$ is lower than any register in $R_{>r'}$, their depths increase.
        The new value of counter $\GetCn(r)$ must reflect the deepest chain,
        therefore the counter gets the value $max\big(\GetCn(R_{>r'})\big)+1$.

     \- $R_{>r'}$ is not empty and $r_=$ exists:
        some chains from registers in $R_{>r'}$ can be decremented into $r'$,
        there is also a chain from $r_=$ that can be extended into $r'$ without its depth changed.
        The counter gets $max\big(max(\GetCn(R_{>r'}))+1, \GetCn(r_=)\big)$,
        which describes the deepest resulting chain.
     \il
     %\ak{ideally, add pictures illustrating these cases...}
     \il

     \manufixed{R2: - Page 16, line 20: Is $B_{arrow up}$ analogous? Why does it exist if we only consider ceiled chains?
- Page 16, line 23: Give the intuition for getCn(r) earlier.
- Page 16, line 40: Give the intuition of why we release counters.
}{added clarifications in respective places}

  The number of states in $\Bd$ is no more than $|R|^{|R|}+1$,
  and the number of counters is $|R|+1$.
  The construction for $B_{\nnearrow}$ is similar to this construction for $\Bd$,
  except that we need to track \emph{increasing} ceiled chains instead of decreasing ones.
  The number of counters in $\Bd$ and $B_{\nnearrow}$ is $2|R|+1$.
  Since we use $|R|$ number of tracers, the total number of counters becomes $|R|(2|R|+1)$.
  Overall, $A_\text{\sc b}$ has an exponential in $|R|$ number of states,
  the number of counters is in $O(|R|^2)$,
  and the parity condition is trivial.
  This concludes the description of the tracers and of the automaton $A_\text{\sc b}$.

  \smallskip
  We have described all three components $A = A_c \land A_{\neg\infty}\land A_\text{\sc b}$,
  where
  $A_c$ expresses a safety language,
  $A_{\neg\infty}$ is a classic deterministic parity automaton,
  and $A_\text{\sc b}$ is a deterministic max-automaton with the trivial parity acceptance condition.
  All the automata has no more than an exponential in $|R|$ number of states,
  $A_{\neg\infty}$ has a polynomial in $|R|$ number of colors,
  and $A_\text{\sc b}$ has a polynomial in $|R|$ number of counters.
  It is not hard to see that the product of these automata gives the desired automaton $A$
  with exponentially many states, polynomially many colors and counters, in $|R|$.
  The acceptance condition is the parity acceptance in conjunction with the formula of $A_\text{\sc b}$ described on page \pageref{eq:Ab-acceptance}.

  Finally, for the case of $0$-satisfiability,
  the automaton $A$ also needs to satisfy the additional conditions stated in Lemma~\ref{lem:0-satisf-N-1},
  in particularly there shall be no decreasing one-way chains from moment $0$ of depth $\geq\!1$.
  This check is simple and omitted.
  This concludes the proof of Theorem~\ref{thm:0-satisf-N}.
  \qed

\parit{Remark}
In~\cite[Appendix C]{ST11} it is shown that satisfiable constraint sequences in $\bbN$
are characterised by nondeterministic $\omega$B-automata~\cite{BC06}.
These automata are incomparable with deterministic max-automata.%

  The following two languages separate these classes:
  $\overline{(a^B b)^\omega}$ is recognised by det max automata but not by nondet $\omega$B automata,
  and $\{a^{n_1} b\, a^{n_2} b\, a^{n_3} b \ldots \mid \lim \inf n_i < \infty\}$ witnesses the opposite direction.
  The latter language is recognisable by the nondet $\omega$B automaton
  which guesses a bounded subsequence of $n_1 n_2 \ldots$.
  The non-recognisability by det max automata follows from~\cite[Section 6]{B11}.%

  We prove the claim about $\overline{(a^B b)^\omega}$.
  First,
  the language $(a^B b)^\omega$ is recognisable by det $\omega$B automata
  and hence by det max automata.
  Since det max automata are closed under the complement,
  $\overline{(a^B b)^\omega}$ is also recognisable by det max automata.
  Now, by contradiction,
  assume that $\overline{(a^B b)^\omega}$ is recognisable by nondet $\omega$B automata.
  The result~\cite[Lemma 2.5]{BC06} says:
  if an $\omega$B language over alphabet $\{a,b\}$
  contains a word with infinitely many $b$s then it contains a word from $(a^B b)^\omega$.
  The language $\overline{(a^B b)^\omega}$ contains the former
  (e.g.\ take any word from $(a^Sb) ^\omega$)
  but not the latter.
  Contradiction. Hence it is not an $\omega$B language.
%%%Also, two-player games with nondeterministic $\omega$B-automaton winning objectives are \emph{not} known to be decidable, while they are decidable for deterministic max-automata
%%%(follows from \cite{B11}+\cite[Example 2]{DBLP:conf/icalp/Bojanczyk14a}).

\subsection{Satisfiability of lasso-shaped sequences}
\label{sec:lasso-characterisation}
An infinite sequence is \emph{lasso-shaped} (or \emph{regular}) if it is of the form $w = uv^\omega$.
Lasso-shaped sequences are prevalent in automata theory and in the data setting in particular.
For instance,
\cite{DD07} studies satisfiability of logic Constraint LTL in the data domain $(\bbN,\leq)$
and shows that considering lasso-shaped witnesses of satisfiability is sufficient.
Another work~\cite{DBLP:conf/icalp/ExibardF022} shows that
if there is an $\omega$-regular over-approximation of satisfiable constraint sequences
and which is exact on lasso-shaped sequences,
then a synthesis problem is decidable in $(\bbN,\leq)$.
In this paper, when proving the decidability of Church synthesis problem,
we do not directly rely on lasso-shaped sequences,
but we use a characterisation similar to the one proven in this section.

This section shows that considering lasso-shaped constraint sequences greatly simplifies
the task of characterisation of satisfiability.
We first show how lasso-shaped sequences simplify the condition $\condBO$ of characterisation Lemma~\ref{lem:satisf-N-1},
then describe the chain characterisation under assumption of lasso-shaped sequences, and
finally state the $\omega$-regular automaton characterisation.

\begin{lemma}\label{lem:lasso-inf-unb}
  For every lasso-shaped consistent constraint sequence,
  it has ceiled one-way chains of \emph{unbounded} depth
  \,iff\,
  it has ceiled one-way chains of \emph{infinite} depth.
\end{lemma}

\begin{proof}
  Direction $\Implied$ is trivial, so consider direction $\Impl$.
  The argument uses the standard pumping technique.
  Fix a lasso-shaped constraint sequence $C_0 \dots C_{k-1} (C_k \dots C_{k+l})^\omega$
  having ceiled chains of unbounded depth.
  Since these chains have unbounded depth,
  they pass through $C_k$ more and more often.
  At moments when the current constraint is $C_k$,
  each such a chain is in one of the finitely-many registers.
  Hence there is a chain, say increasing,
  that on two separate occasions of reading the constraint $C_k$ goes through the \emph{same} register $r$,
  \emph{and} the chain suffix from the first pass through $r$ until the second pass has at least one $<$.
  Then we create an increasing chain of infinite depth by repeating this suffix forever.
\end{proof}

The above lemma together with Lemma~\ref{lem:0-satisf-N-1} yields the following result.

\begin{lemma}\label{lem:N-satisf-lasso}
  A lasso-shaped consistent constraint sequence is $0$-satisfiable iff it is quasi-feasible, i.e.:
  \li
  \- it has no infinite-depth decreasing one-way chains,
  \- it has no ceiled infinite-depth increasing one-way chains,
  \- it has no decreasing one-way chains of depth $\geq\!1$ from moment $0$, and
  \- it starts with $C_0$ s.t.\ ${C_0}_{|R} = \{r=s\mid r,s\in R\}$.
  \il
\end{lemma}

The conditions of this lemma can be checked by an $\omega$-regular automaton:
Its construction is similar to the components $A_c$ and $A_{\neg\infty}$ from the proof of Theorem~\ref{thm:0-satisf-N} and is omitted.
Thus, we get the theorem below.

\begin{theorem}\label{lem:DPA-for-satisf-lasso}
  For every $R$,
  there is a deterministic parity automaton
  that accepts a lasso-shaped constraint sequence
  iff it is $0$-satisfiable in $\bbN$;
  its number of states and priorities is exponential and polynomial in $|R|$, respectively.
\end{theorem}

\subsection{Data-assignment function}
\label{sec:data-assignment-function}
In this section,
we design a data-assignment function that
maps a sequence of constraints to a sequence of register valuations satisfying it,
while doing it on the fly, i.e.\ by reading the constraint sequence from left to right.
It is significant that the entire constraint sequence is not known in advance.
Such a function is used in Section~\ref{sec:synt-games} when proving Proposition~\ref{prop:GS_Gf_Gfreg},
namely that Adam's winning strategy in the finite-alphabet game transfers
to the winning strategy in the Church synthesis game.
There, Adam has to produce data values given only the prefix of a play.

In the next section, we state the lemma on existence of a data-assignment function,
and then devote a significant amount of space to proving it.

\subsubsection{Lemma~\ref{lem:data-assign-func} on existence of a data-assignment function}
\label{sec:data-assign-func-exists}
Intuitively, a data-assignment function produces register valuations while reading a constraint sequence from left to right.
We are interested in functions that produce register valuations satisfying given constraint sequences.
Since data-assignment functions cannot look into the future and
do not know how many values will be inserted between any two registers,
knowing a certain bound on such insertions is necessary.
Moreover, to simplify the presentation,
we restrict how many new data values can appear during the step.
In our Church synthesis games,
at most one new value provided by Adam can appear.
We start by defining data-assignment functions, then describe the assumptions and state the lemma.

Let ${\sf C}$ denote the set of all constraints over registers $R$,
and let ${\sf C}_{|R}$ denote the set of all constraints over atoms over $R$ only.
A \emph{data-assignment function} has the type
$({\sf C}_{|R}\cup {\sf C}^+) \to \bbN^R$.
A data-assignment function $f$ \emph{maps} a constraint sequence $C_0 C_1 ...$ into a sequence of valuations
$f({C_0}_{|R}) f(C_0) f(C_0 C_1)...$.

We now describe the two assumptions used by our data-assignment function.

Intuitively, the first assumption states that only a bounded number of insertions between any two registers can happen,
and this bound is known.
To formalise the assumption, we define a special kind of chains, called right two-way chains.
Informally, right chains are two-way chains that operate to the right of their starting point.
Knowing a bound on the depths of right chains amounts to knowing how many values in the future can be inserted between the registers.
\manufixed{R2: Page 18, line 29: When defining right chains give the intuition that everything has to happen on the right of the chain's starting point, that is the chain is not allowed to go to the left of its initial position. You mention this later on page 36 line 58 but the earlier you give the reader this intuition the more useful it is to them.}{added the intuition before the definition}
Fix a constraint sequence.
Given a moment $i$ and a register $x$,
a (decreasing) \emph{right two-way chain starting in $(x,i)$} (\emph{r2w} for short)
is a two-way chain
$(x,i) \tr_1 (r_1,m_1) \tr_2 (r_2,m_2) \ldots $
such that $m_j\geq i$, $\tr_j \in \{=,>\}$,
for all $j$.
As these chains are two-way,
they can start and end in the same moment $i$.
Notice that in Lemma~\ref{lem:satisf-N}
on characterisation of satisfiable constraint sequences
we can replace two-way chains by r2w chains.
Our data-assignment function will assume the knowledge of a bound on the r2w chains.
\manufixed{R2: - Page 18, line 37: Why is it a reasonable assumption that the bound B is known? line 38: What we are recalling here has not been mentioned this clearly before, consider doing that earlier.}{added intuition why having a bound on r2w chains is reasonable: because the bound tells us how many insertions between any two registers can we expect in the future, knowing such a bound is necessary for any data-assignment function.}
\begin{wrapfigure}{r}{20mm}
  \vspace{-3mm}
  \centering
  \includegraphics[width=0.16\textwidth]{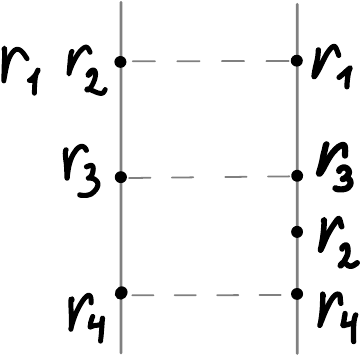}
  \vspace{-10mm}
\end{wrapfigure}

We now describe the second assumption about one-new-value appearance during a step.
Its formalisation uses the notion of levels introduced in Section~\ref{sec:max-automata-characterisation} on page~\pageref{def:levels} (see also Figure~\ref{fig:levels}).
We briefly recall those notions.
Recall that a constraint describes a set of totally ordered equivalence classes of registers from $R \cup R'$.
The figure on the right describes a constraint that can be defined by the ordered
equivalence classes $\{r_4,r'_4\} < \{r'_2\} < \{r_3,r'_3\} < \{r_1,r_2,r'_1\}$.
It shows two columns of levels, start levels (in the left column) and end levels (in the right column),
where a level describes a set of registers that are equivalent at this point of time.
%The levels respect the constraint:
%if a level at moment $m$ is on a higher/lower/equal level than a dot at moment $m+1$,
%then the constraint requires the registers of the second dot to have higher/lower/equal values
%than the registers of the first dot.
%A dot on a new level appears in the second column,
%i.e.\ there were no dots on that level at moment $m$,
%if and only if
%the constraint contains an equivalence class consisting solely of $R'$-registers.
%A level that was present in the first column disappears from the second column
%if and only if
%the constraint contains an equivalence class consisting solely of $R$-registers.
The assumption $\dagger$ says:
\begin{equation}\tag{$\dagger$}
\begin{array}{l}
  \textit{In every constraint of a given sequence, at most one new end level appear.}
  \end{array}
\end{equation}
The constraint depicted in the above figure satisfies this assumption, the one in Figure~\ref{fig:levels} does not.
This assumption helps to simplify the proofs, and is satisfied by the constraint sequences induced in our Church synthesis games.

One final notion before stating the lemma.
A constraint sequence is \emph{$0$-consistent} if
it is consistent, starts in $C_0$ with ${C_0}_{|R} = \{r=s\mid r,s\in R\}$,
and has no decreasing chains of depth $\geq1$ starting at moment $0$.
\manufixed{R1 p18 59: The definition of meaningful is easy to miss. Maybe it can be emphasized more?}{changed the notion to $0$-consistent}
Note that a $0$-consistent constraint sequence whose r2w chains are bounded is $0$-satisfiable (follows from Lemma~\ref{lem:satisf-N}).

\begin{lemma}[data-assignment function]\label{lem:data-assign-func}
  For every $\Bound\geq 0$,
  there exists a data-assignment function
  $f : ({\sf C}_{|R}\cup {\sf C}^+) \to \bbN^R$ such that
  for every finite or infinite $0$-consistent constraint sequence $C_0 C_1 C_2 ...$ satisfying assumption~\!$\dagger$
  and whose r2w chains are depth-bounded by $\Bound$,
  the register valuations $f({C_0}_{|R}) f(C_0) f(C_0 C_1)...$
  satisfy the constraint sequence.
\end{lemma}

\begin{proof}[Proof idea]
  We define a special kind of $xy^{(m)}$-chains
  that help to estimate how many insertions between the values of registers $x$ and $y$ at moment $m$ we can expect in the future.
  As it turns out, without knowing the future,
  the distance between $x$ and $y$ has to be exponential in the maximal depth of $xy^{(m)}$-chains.
  We describe a data-assignment function that maintains such exponential distances.
  The function is surprisingly simple:
  if the constraint inserts a register $x$ between two registers $r$ and $s$ with already assigned values $\d_r$ and $\d_s$,
  then set $\d_x = \floor{\frac{\d_r+\d_s}{2}}$;
  and if the constraint puts a register $x$ above all other registers,
  then set $\d_x = \d_M + 2^{\Bound}$ where $\d_M$ the largest value currently held in the registers and $\Bound$ is the given bound on the depth of r2w chains.
\end{proof}

The rest of the section is devoted to the proof of this lemma.

\subsubsection{Proof of Lemma~\ref{lem:data-assign-func}}

\paragraph{$xy^{(m)}$-connecting chains and the exponential nature of register valuations}
% \ak{is using $\alpha+\beta$ simpler?}

Fix an arbitrary $0$-satisfiable constraint sequence $C_0 C_1 ...$
whose r2w chains are depth-bounded by $\Bound$.
Consider a moment $m$ and two registers $x$ and $y$ such that $(x>y) \in C_m$.

We would like to construct witnessing valuations $\v_0 \v_1 ...$ using the current history only,
e.g.\ a register valuation $\v_m$ at moment $m$ given only the prefix $C_0 ... C_{m-1}$.
Note that the prefix $C_0 ... C_{m-1}$ defines the ordered partition of registers at moment $m$ as well,
since $C_{m-1}$ is defined over $R \cup R'$.
Let us see how much space we might need between $\v_m(x)$ and $\v_m(y)$,
relying
\begin{wrapfigure}{r}{19mm}
  \centering
  \includegraphics[width=0.15\textwidth]{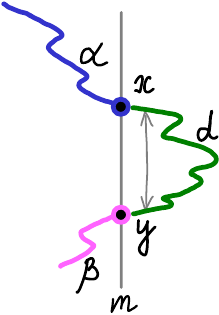}
  \vspace{-5mm}
\end{wrapfigure}
on the fact that the depths of r2w chains are bounded by $\Bound$.
Consider decreasing two-way chains that
start at moment $i\leq m$,
end in $(x,m)$,
and which are contained within time moments $\{i,...,m\}$ (shown in blue).
Further, consider decreasing two-way chains
starting in $(y,m)$,
ending at moment $j \in \{i,...,m\}$,
and contained within time moments $\{j,...,m\}$ (shown in pink).
Among such chains,
pick two chains of depths $\alpha$ and $\beta$, respectively,
that maximise the sum $\alpha+\beta$.
After seeing $C_0 C_1 ... C_{m-1}$,
we do not know how the constraint sequence will evolve,
but by boundedness of r2w chains,
any r2w chain starting in $(x,m)$ and ending in $(y,m)$
(contained within time moments $\geq m$)
will have a depth $d \leq \Bound - \alpha - \beta$
(otherwise, we could add prefix $\alpha$ and postfix $\beta$ to it and construct an r2w chain of depth larger than $\Bound$).
We conclude that $\v_m(x) - \v_m(y) \geq \Bound - \alpha - \beta$,
since the number of values in between two registers should be greater or equal than the longest two-way chain connecting them.
To simplify the upcoming arguments,
we introduce $xy^{(m)}$-connecting chains which consist of $\alpha$ and $\beta$ parts
and directly connect $x$ to $y$.

An \emph{$xy^{(m)}$-connecting chain}
is any r2w chain of the form
$(a,i) \tr \ldots (x,m) > (y,m) \tr \ldots \tr (b,j)$:
it starts in $(a,i)$ and ends in $(b,j)$, where $i \leq j \leq m$ and $a,b \in R$,
and it \emph{directly} connects $x$ to $y$ at moment $m$.
Note that it is located solely within moments $\{i,...,m\}$.
Continuing the previous example,
the $xy^{(m)}$-connecting chain starts with $\alpha$,
directly connects $(x,m)>(y,m)$, and ends with $\beta$;
its depth is $\alpha+\beta+1$
(we have ``+1'' no matter how many registers are between $x$ and $y$,
 since $x$ and $y$ are connected directly).

With this new notion,
the requirement $\v_m(x) - \v_m(y) \geq \Bound - \alpha - \beta$ becomes
$\v_m(x) - \v_m(y) \geq \Bound - d_{xy} + 1$,
where $d_{xy}$ is the largest depth of $xy^{(m)}$-connecting chains.

\begin{wrapfigure}{r}{26mm}
  \vspace{-8mm}
  \includegraphics[width=0.2\textwidth]{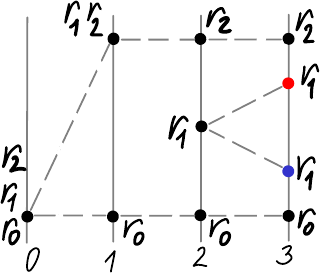}
  \vspace{-16mm}
\end{wrapfigure}

However, since we do not know how the constraint sequence evolves after $C_0 ... C_{m-1}$,
we might need even more space between the registers at moment $m$.
Consider an example on the right,
with $R = \{r_0,r_1,r_2\}$ and the bound $\Bound=3$ on the depth of r2w chains.
\li
\- Suppose at moment $1$, after seeing the constraint $C_0$,
   which is $\{r'_1,r'_2\}>\{r_0,r_1,r_2,r'_0\}$,
   the valuation is $\v_1 = \{r_0\mapsto 0; r_1,r_2 \mapsto 3\}$.
   It satisfies $\v_1(r_2) - \v_1(r_0) \geq \Bound-d_{r_2r_0}+1$
   (indeed, $\Bound=3$ and $d_{r_2 r_0}=1$ at this moment);
   similarly for $\v(r_1)-\v(r_0)$.

\- Let the constraint $C_1$ be $\{r_1,r_2,r'_2\}>\{r'_1\}>\{r_0,r'_0\}$.
   What value $\v_2(r_1)$ should register $r_1$ have at moment $2$?
   Note that the assignment should work no matter what $C_2$ will be in the future.
   Since the constraint $C_1$ places $r_1$ between $r_0$ and $r_2$ at moment $2$,
   we can only assign $\v_2(r_1)=2$ or $\v_2(r_1)=1$.
   If we choose $2$, then the constraint $C_2$ having
   $\{r_2,r'_2\}>\{r_1'\}>\{r_1\}>\{r_0,r'_0\}$
   (the red dot in the figure)
   shows that there is not enough space between $r_2$ and $r_1$ at moment $2$
   ($\v_2(r_2)=3$ and $\v_2(r_1)=2$).
   Similarly for $\v_2(r_1)=1$:
   the constraint $C_2$ having
   $\{r_2,r'_2\}>\{r_1\}>\{r'_1\}>\{r_0,r'_0\}$
   (the blue dot in the figure)
   eliminates any possibility for a correct assignment.
\il
Thus, at moment $2$, the register $r_1$ should be equally distanced from $r_0$ and $r_2$,
i.e.\ $\v_2(r_1)\approx\frac{\v_2(r_0)+\v_2(r_2)}{2}$,
since its evolution can go either way, towards $r_2$ or towards $r_0$.
This hints at the exponential nature of distances between the registers.
This is formalised in the next lemma
showing that any data-assignment function that places two registers $x$ and $y$ at any moment $m$
closer than $2^{\Bound-d_{xy}}$ is bound to fall.
Intuitively,
$\Bound - d_{xy}$ describes how many more times an insertion between the values of registers $x$ and $y$ can happen in the future.
Since each newly inserted value should be equidistant from the boundaries,
we get the $2^{\Bound-d_{xy}}$ lower bound.

\begin{lemma}[tightness]\label{lem:conn:exp}
  Fix $\Bound\geq 3$, registers $R$ of $|R|\geq 3$,
  a $0$-consistent constraint sequence prefix $C_0 ... C_{m-1}$
  where $m \geq 1$ and whose r2w chains are depth-bounded by $\Bound$,
  two registers $x,y \in R$ s.t.\ $(x'>y') \in C_{m-1}$,
  and a data-assignment function $f : ({\sf C}_{|R} \cup {\sf C}^+) \to \bbN^R$.
  Let $\v_m=f(C_0 ... C_{m-1})$
  and $d_{xy}$ be the maximal depth of $xy^{(m)}$-connecting chains.
  If $\v_m(x) - \v_m(y) < 2^{\Bound-d_{xy}}$,
  then there exists a continuation $C_m C_{m+1} ...$
  such that the whole sequence $C_0 C_1 ...$ is $0$-consistent and its r2w chains are depth-bounded by $\Bound$ (hence $0$-satisfiable),
  yet $f$ cannot satisfy it.
\end{lemma}

\begin{proof}
  We use the idea from the previous example.
  %The continuation can only deepen $xy^{(m)}$-connecting chains
  %(i.e.\ by replacing the part $(x,m)>(y,m)$ we might be able to create a deeper r2w chain),
  %while other r2w chains, which exist at moments $\{0,...,m\}$,
  %won't be extendable by new constraints (toprove-1).
  The constraints $C_m C_{m+1}...$ are:
  \lo
  \- If at moment $m$ there are registers different from $x$ and $y$,
     we add the step that makes them equal to $x$ (or to $y$):
     this does not affect the depth of $xy$-connecting chains at moments $m$ and $m+1$;
     also, the maximal depths of r2w chains defined at moments $\{0,...,m\}$ and $\{0,...,m+1\}$ stay the same.
     % The proof of this fact is left to the reader.
     Therefore, below we assume that at moment $m$ every register is equal to $x$ or to $y$.

  \-\label{item:contra}
     If $\Bound-d_{xy} = 0$, we are done:
     $\v_m(x)-\v_m(y) < 2^{\Bound-d_{xy}}$ gives $\v_m(x) \leq \v_m(y)$ but $C_{m-1}$ requires $\v_m(x)>\v_m(y)$.
     The future constraints then simply keep the registers constant.
     Otherwise, when $\Bound-d_{xy}>0$, we proceed as follows.
  \manufixed{R2: - Page 21, line 14: How did you come up with the bound $2^{B-d_{xy}}$?}{added a clarification before the proof}
  \- To ensure consistency of constraints,
     $C_m$ contains all atoms over $R$ that are implied by atoms over $R'$ of $C_{m-1}$.

  \- $C_m$ contains $x=x'$ and $y=y'$.

  \- $C_m$ places a register $z$ between $x$ and $y$: $x'>z'>y'$.\\
     This gives $d'_{xz} = d'_{zy} = d_{xy}+1 \leq b$,
     where $d_{xy}$ is the largest depth of connecting chains for $xy^{(m)}$,
     $d'_{xz}$--- for $xz^{(m+1)}$, and
     $d'_{zy}$--- for $zy^{(m+1)}$.
     Since $\v_{m+1}(x) - \v_{m+1}(y) < 2^{\Bound-d_{xy}}$,
     either $\v_{m+1}(x)-\v_{m+1}(z)<2^{\Bound-d'_{xz}}$ ~or~ $\v_{m+1}(z)-\v_{m+1}(y)<2^{\Bound-d'_{zy}}$;
     this is the key observation.
     If the first case holds, we have the original setting $\v_{m+1}(x) - \v_{m+1}(z) < 2^{\Bound-d'_{xz}}$
     but at moment $m+1$ and with registers $x$ and $z$;
     for the second case --- with registers $z$ and $y$.
     Hence we repeat the entire procedure, again and again,
     until reaching the depth $\Bound$, which gives the sought conclusion in item (\ref{item:contra}).
  \ol
  Finally, it is easy to prove that the whole constraint sequence $C_0 C_1 ...$ is $0$-satisfiable,
  e.g.\ by showing that it satisfies the conditions of Lemma~\ref{lem:0-satisf-N-1}.
  %\li
  %\- it is consistent and starts in $C_0$ with $(r=s) \in C_0$ for all $r,s \in R$:
  %   indeed, the prefix $C_0 ... C_{m-1}$ is $0$-satisfiable and our extension is consistent;
  %\- it has no infinitely descending chains and no ceiled one-way chains of unbounded depth:
  %   this holds because our extension creates only a finite number of register changes
  %   followed by a suffix with no register changes;
  %\- it has no one-way decreasing chains of depth $\geq 1$ from $(r,0)$:
  %   this holds because the $0$-satisfiability of $C_0 ... C_{m-1}$ implies that
  %   there are no such chains until moment $m$,
  %   and our extension always puts registers above or equal to $(y,m)$.
  %\il
  Moreover, it is $0$-consistent, and all r2w chains of $C_0 C_1 ...$ are depth-bounded by $\Bound$ because:
  (a) in the initial moment $m$, all r2w chains are depth-bounded by $\Bound$; and
  (b) the procedure deepens only $xy$-connecting chains and only until the depth $\Bound$,
      whereas other r2w chains existing at moments $\{0,...,m\}$ keep their depths unchanged
      (or at moments $\{0,...,m+1\}$, if we executed item 1).
\end{proof}

\paragraph{Proof of Lemma~\ref{lem:data-assign-func} under additional assumption about 0}\label{page:data-assign-func}

Tightness by Lemma~\ref{lem:conn:exp} tells us that if a data-assignment function exists,
it should separate the register values by at least $2^{\Bound-d_{xy}}$.
Such separation is sufficient as we show below.
We first describe a data-assignment function,
then prove an invariant about it,
and finally conclude with the proof of Lemma~\ref{lem:data-assign-func}.
For simplicity,
we assume that the constraints contain a register that never changes and always holds $0$.
That is not true in general, so later we will lift this assumption.

\subparagraph{\normalfont \textbf{Data-assignment function}}
The function $f : ({\sf C}_{|R} \cup {\sf C}^+) \to \bbN^R$
is constructed inductively on the length of $C_0...C_{m-1}$ as follows.

Initially, $f({C_0}_{|R})=\v_0$ where $\v_0(r)=0$ for all $r \in R$
(since $C_0$ has $r=s$, $\forall r,s \in R$).
Suppose at moment $m$,
the register valuation is $\v_m = f({C_0}_{|R} C_0 ... C_{m-1})$.
Let $C_m$ be the next constraint, then $\v_{m+1} = f({C_0}_{|R} C_0 ... C_m)$ is as follows:
\newcommand\DAone{{\sf D}$\frak1$\xspace}
\newcommand\DAtwo{{\sf D}$\frak2$\xspace}
\newcommand\DAthree{{\sf D}$\frak3$\xspace}
\lo
\-[\DAone.] If a register $x$ at moment $m+1$ lays above all registers at moment $m$,
   i.e.\ $(x'>r) \in C_m$ for every register $r$,
   then set $\v_{m+1}(x) = \v_m(r)+2^\Bound$, where $r$ is one of the largest registers at moment $m$.
   In Church games this case happens when the test contains the atom $*>r$.

\-[\DAtwo.] If a register $x$ at moment $m+1$ lays between two adjacent registers $a>b$ at moment $m$,
   then $\v_{m+1}(x) = \floor{\frac{\v_m(a)+\v_m(b)}{2}}$.
   In Church games this happens when the test contains $a>*>b$.

\-[\DAthree.] If a register $x$ at moment $m+1$ equals a register $r$ at previous moment $m$,
   so $(r=x') \in C_m$,
   then $\v_{m+1}(x) = \v_m(r)$.
   In Church games this case corresponds to a test containing the atom $*=r$ for some register $r$.
\ol
Note that the case when a register $x$ must lay below all registers never happens,
since the special register $r_0$ always holds $0$ and a given constraint sequence is $0$-consistent and hence never requires $r_0>r'$ for some register $r$.
This is where $r_0$ comes handy.
%Since we cannot insert a value below the register holding value $0$,
%these cases are sufficient.

\subparagraph{Invariant}
The data-assignment function satisfies the following invariant:
$$
  \forall m \in \bbN. ~\forall x,y \in R \textit{~s.t.~} (x>y) \in C_m\:~ \v_m(x) - \v_m(y) \geq 2^{\Bound-d_{xy}},
$$%
where $d_{xy}$ is the largest depth of $xy^{(m)}$\!-connecting chains and $\Bound$ is the bound on the depth of r2w chains.

\manufixed{R2: Page 22: I assume you refer with the items on the bottom of the page to the conditions on the next constraint at the top of the page even though these conditions have never been introduced as 'items'. The references are also inconsistent: (item 3) in line 49, item 1 in line 59, and item (3) in line 46 on page 23.}{we gave them the proper names to avoid any confusion}
\subparagraph{Proof of the invariant}
The invariant holds initially since $(r_1=r_2) \in C_0$ for all $r_1,r_2 \in R$.
Assuming it holds at step $m$, we show that it holds at $m+1$.
Fix two arbitrary registers $x,y \in R$ such that $(x'>y') \in C_m$;
we will prove that $\v_{m+1}(x)-\v_{m+1}(y) \geq 2^{\Bound-d_{xy}}$,
where $d_{xy}$ is the largest depth of $xy^{(m+1)}$-connecting chains.
There are four cases depending on whether the levels of $x$ and $y$ at moment $m+1$
are present at moment $m$ or not, illustrated in Figure~\ref{fig:da-invariant-proof}.

\newcommand\dr{\ensuremath{d_{r_2r_1}}}
\newcommand\ds{\ensuremath{d_{s_2w_1}}}

\parit{Case 1: both present}
The levels of $x$ and $y$ at $m+1$ also exist at moment $m$.
Let $a,b$ be registers s.t.\ $(a>b) \in C_m$ laying at moment $m$ on the same levels as $x$ and $y$ at moment $m+1$.
By data-assignment function (item \DAthree),
$\v_m(a) = \v_{m+1}(x)$ and $\v_m(b) = \v_{m+1}(y)$.
Note that the number of levels between $x$-$y$ and between $a$-$b$ may differ.
Consider the depths of connecting chains for $ab^{(m)}$ and $xy^{(m+1)}$:
Since every $ab^{(m)}$-connecting chain can be extended to $xy^{(m+1)}$-connecting chain of the same depth as shown on the figure,
we have\footnote{A stronger result holds, namely $d_{ab}=d_{xy}$, but it is not needed here.} $d_{ab} \leq d_{xy}$,
and hence $2^{\Bound-d_{ab}} \geq 2^{\Bound-d_{xy}}$.
Using the inductive hypothesis,
we conclude
$\v_{m+1}(x)-\v_{m+1}(y)=\v_m(a)-\v_m(b)\geq 2^{\Bound-d_{ab}}\geq 2^{\Bound-d_{xy}}$.

\begin{figure}
  \centering
  \begin{subfigure}[t]{.18\textwidth}
    \centering
    \includegraphics[width=\textwidth]{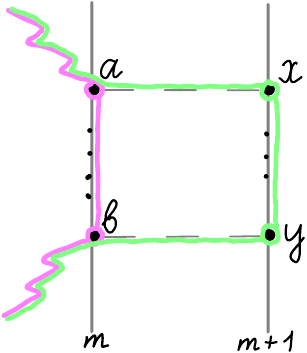}
    \caption*{Case 1}
  \end{subfigure}~~~~~~~
  \begin{subfigure}[t]{.17\textwidth}
    \centering
    \includegraphics[width=\textwidth]{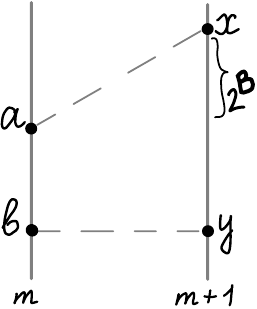}
    \caption*{Case 2}
  \end{subfigure}~~~~~~~
  \begin{subfigure}[t]{.18\textwidth}
    \centering
    \includegraphics[width=\textwidth]{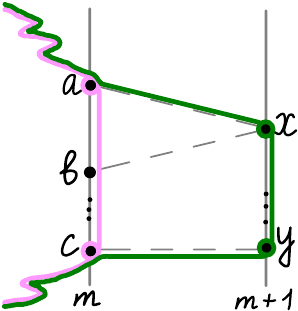}
    \caption*{Case 3}
  \end{subfigure}~~~~~~~
  \begin{subfigure}[t]{.17\textwidth}
    \centering
    \includegraphics[width=\textwidth]{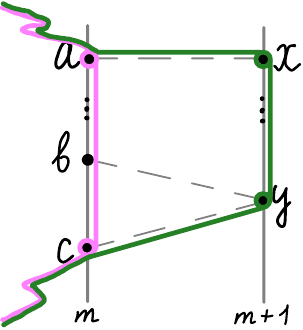}
    \caption*{Case 4}
  \end{subfigure}
  \caption{Proving the invariant}
  \label{fig:da-invariant-proof}
\end{figure}

\parit{Case 2: $x$ is new top}
The register $x$ lies on the top level of both moments $m$ and $m+1$,
and $y$ lies on a level that was also present at moment $m$.
This corresponds to item \DAone.
Let $(b = y') \in C_m$ and $a$ lies on the largest level at moment $m$
($a$ and $b$ may coincide).
Thus, $\v_{m+1}(x) = \v_m(a) + 2^\Bound$.
The invariant holds for $x,y$ because
$\v_{m+1}(x)=\v_m(a)+2^\Bound$ and $\v_m(a)\geq\v_m(b)=\v_{m+1}(y)$.

\parit{Case 3: $x$ is middle new, $y$ was present}
The register $x$ at moment $m+1$ lies on a new level that is between the levels of $a$ and $b$
at moment $m$, so $\v_{m+1}(x) = \floor{\frac{\v_m(a)+\v_m(b)}{2}}$ by item \DAtwo of data-assignment function.
The register $y$ at moment $m+1$ lies on a level that was also present at moment $m$, witnessed by register $c$.
Formally,
$C_m$ contains $a>x'>b$ for $a$ and $b$ adjacent at moment $m$,
$c=y'$, and $x'>y'$.
Note that $c$ and $b$ may coincide.
Then,
$
\v_{m+1}(x) - \v_{m+1}(y) = \floor{\frac{\v_m(a)+\v_m(b)}{2}} - \v_m(c) =
\floor{\frac{\v_m(a)-\v_m(c)}{2} + \frac{\v_m(b)-\v_m(c)}{2}} \geq
\floor{\frac{\v_m(a)-\v_m(c)}{2}} + \floor{\frac{\v_m(b)-\v_m(c)}{2}} \geq
\floor{2^{\Bound-d_{ac}-1}} + \floor{2^{\Bound-d_{bc}-1}}\geq
2^{\Bound-d_{ac}-1} + \floor{2^{\Bound-d_{bc}-1}}
$;
the latter holds because $d_{a c} < b$ while $d_{bc} \leq b$.
We need to prove that the last sum is greater or equal to
$2^{\Bound-d_{xy}}$.
Figure~\ref{fig:da-invariant-proof} (case 3) shows how the green $xy^{(m+1)}$-connecting chain can be constructed from the pink $ac^{(m)}$-connecting chain, hence $d_{xy} \geq d_{ac}+1$,
so we get $2^{\Bound-d_{ac}-1} \geq 2^{\Bound-d_{xy}}$.
Hence,
$\v_{m+1}(x)-\v_{m+1}(y) \geq 2^{\Bound-d_{ac}-1} + \floor{2^{\Bound-d_{bc}-1}} \geq 2^{\Bound-d_{xy}}$.

\parit{Case 4: $x$ was present, $y$ is middle new}
The case is similar to the previous one, but we prove it for completeness.
The constraint $C_m$ contains $a = x'$, $x'>y'$, $b>y'>c$,
where $b$ and $c$ are adjacent ($a$ and $b$ might be the same).
Then,
$
\v_{m+1}(x) - \v_{m+1}(y) =
\v_m(a) - \floor{\frac{\v_m(b)+\v_m(c)}{2}} \geq
\floor{\frac{\v_m(a)-\v_m(b)}{2} + \frac{\v_m(a)-\v_m(c)}{2}} \geq
\floor{\frac{\v_m(a)-\v_m(b)}{2}} + \floor{\frac{\v_m(a)-\v_m(c)}{2}} \geq
\floor{2^{\Bound-d_{ab}-1}} + \floor{2^{\Bound-d_{a c}-1}}\geq
\floor{2^{\Bound-d_{ab}-1}} + 2^{\Bound-d_{ac}-1}
$,
and since $d_{a c}+1 \leq d_{x y}$,
we get
$\v_{m+1}(x) - \v_{m+1}(y) \geq
 \floor{2^{\Bound-d_{a b}-1}} + 2^{\Bound-d_{ac}-1} \geq
 2^{\Bound-d_{xy}}$.
%\end{proof}
\qed

\subparagraph{Proof of Lemma~\ref{lem:data-assign-func}}
It is sufficient to show that for every atom
$(r \bowtie s)$ or $(r \bowtie s')$ of $C_m$,
where $r,s \in R$ and ${\bowtie} \in \{<,>,=\}$,
the expressions
$\v_m(r) \bowtie \v_m(s)$ or $\v_m(r) \bowtie \v_{m+1}(s)$ hold,
respectively.
Depending on $r \bowtie s$, there are the following cases.
\li
\- If $C_m$ contains $(r=s)$ or $(r=s')$ for $r,s \in R$,
   then item \DAthree implies resp.\ $\v_m(r)=\v_m(s)$ or $\v_m(r) = \v_{m+1}(s)$.

\- If $(r>s) \in C_m$, then $\v_m(r) > \v_m(s)$ by the invariant.

\- Let $(r > s') \in C_m$
   and the level of $s$ at moment $m+1$ be present at moment $m$,
   i.e.\ there is a register $t$ such that $(t=s') \in C_m$.
   Since $\v_m(t)=\v_{m+1}(s)$ by item \DAthree and
   since $\v_m(r) > \v_m(t)$ by $(r > t = s') \in C_m$,
   we get $\v_m(r) > \v_{m+1}(s)$.
   Similarly for the case $(r<s') \in C_m$ where $s$ lies on a level also present at moment $m$.

\- Let $(r < s') \in C_m$ and $s$ lies on the highest level among all levels at moments $m$ and $m+1$.
   Then $\v_m(r) < \v_{m+1}(s)$ because $\v_{m+1}(s) \geq \v_m(r)+2^\Bound$ by item \DAone.

\- \label{<>}
   Finally, there are two cases left:
   $(r > s') \in C_m$ or $(r < s') \in C_m$,
   where $s$ lies on a newly created level at moment $m+1$,
   and there are higher levels at moment $m$.
   This corresponds to item \DAtwo.
   Let $(a>b)\in C_m$ be two adjacent registers at moment $m$
   between which the register $s$ is inserted at moment $m+1$,
   so $(a>s'>b) \in C_m$.
   Let $d_{ab}$ be the maximal depth of $ab^{(m)}$-connecting chains;
   fix one such chain.
   We change it by going through $s$ at moment $m+1$,
   i.e.\ substitute the part $(a,m)>(b,m)$ by $(a,m)>(s,m+1)>(b,m)$:
   the depth of the resulting chain is $d_{ab}+1$ and it is $\leq \Bound$ by boundedness of r2w chains.
   Hence $d_{a b} \leq \Bound-1$,
   so $\v_m(a)-\v_m(b) \geq 2$,
   implying $\v_m(a) > \floor{\frac{\v_m(a)+\v_m(b)}{2}} > \v_m(b)$.
   When $(r>s')\in C_m$ we get $\v_{m+1}(r)\geq \v_m(a)$,
   and when $(r<s') \in C_m$ we get $\v_{m+1}(r)\leq \v_m(b)$,
   therefore we are done.
\il
Finally, the function always assigns nonnegative numbers, from $\bbN$, so we are done.
\qed

\paragraph{Lifting the assumption about $0$}

We now lift the assumption about a register always holding $0$.
This assumption was used in the definition of the data-assignment function (items \DAone, \DAtwo, \DAthree).
The idea is to convert a given constraint sequence over registers $R$
into a sequence over registers $R \uplus \{r_0\}$
while preserving satisfiability.

\subparagraph{Conversion function}
Given a $0$-consistent constraint sequence $C_0 C_1 ...$ over $R$ without a special register holding $0$,
we will construct, on-the-fly, a $0$-consistent sequence $\tilde C_0 \tilde C_1 ...$ over $R\uplus\{r_0\}$ that has such a register.
Intuitively, we will add atoms $r=r_0$ only if they follow from what is already known otherwise we add atoms $r>r_0$.

Initially, in addition to the atoms of $C_0$, we require $r=r_0$ for every $r\in R$
(recall that the original $C_0$ contains $r_1=r_2$ for all $r_1,r_2 \in R$).
This gives an incomplete constraint $\tilde C_0$ over $R_0 \cup R'_0$:
it does not yet have atoms of the form $r \bowtie r'_0$, $r_0 \bowtie r'$, $r'_0 \bowtie r'$,
where $r \in R_0$.
\newcommand\Ct{{\tilde C_{m|R_0}}}%

At moment $m \geq 0$,
given
a constraint $\Ct$ over $R_0$ (without primed registers $R'_0$) and
a constraint $C_m$ over $R \cup R'$ (without register $r_0$),
we construct $\tilde C_m$ over $R_0\cup R_0'$ as follows:
\li
\- $\tilde C_m$ contains all atoms of $C_m$.

\- $(r_0=r'_0) \in \tilde C_m$.

\-\label{conv:r}
   For every $r \in R$:
   if $r' = r_0$ is implied by the current atoms of $\tilde C_m$,
   then we add it, otherwise we add $r'>r_0$.

   Notice that the atom $r'<r_0$ is never implied by $\tilde C_m$, as we show now.
   Suppose the contrary.
   Then, since $C_m$ does not talk about $r_0$ nor $r_0'$,
   there should be $s\in R$ such that $(s=r_0) \in \tilde C_{m|R_0}$ and $(r'<s) \in C_m$.
   By construction, if this is the case, then there is a one-way chain
   $(r_1,0) = (r_2,1) = ... = (s, m)$ of zero depth. As a consequence,
   we can construct the one-way decreasing chain $(r_1,0)=(r_2,1)=...=(s,m)>(r,m+1)$ of depth $1$,
   which implies that $C_0 C_1 ...$ is not $0$-consistent.
   We reached a contradiction, so $(r'<r_0) \in \tilde C_m$ is not possible.

\- Finally, to make $\tilde C_m$ maximal,
   we add all atoms implied by $\tilde C_m$ but not present there.
\il
Using this construction,
we can easily define $\mathit{c0nv} : C^+ \to \tilde C$
and map a given $0$-consistent constraint sequence $C_0 C_1 ...$
to $\tilde C_0 \tilde C_1 ...$ with a dedicated register holding $0$.
Notice that the constructed sequence is also $0$-consistent,
because we never add inconsistent atoms and never add an atom $r'<r_0$ (see the third item).
Finally,
in the constructed sequence the depths of r2w chains can increase by at most $1$,
due to the register $r_0$:
it can increase the depth of a finite chain by one, unless the chain is already ending in a register holding $0$.
Hence we get the following lemma.

\begin{lemma}\label{lem:conv-zero-reg}
  For every $0$-consistent constraint sequence $C_0 C_1 ...$,
  the sequence $\tilde C_0 \tilde C_1 ...$ constructed with $\mathit{c0nv}$
  is also $0$-consistent.
  Moreover, the maximal depth of r2w chains cannot increase by more than $1$.
\end{lemma}

\subparagraph{Final proof of Lemma~\ref{lem:data-assign-func}}
We lift the assumption about constraint sequences having a special register always holding zero.
Using $\mathit{c0nv}$,
we automatically translate a given $0$-consistent constraint sequence prefix $C_0 ... C_m$ over $R$
into $\tilde C_0 ... \tilde C_m$ over $R \uplus \{r_0\}$ that contains a register $r_0$ always holding $0$.
Now we can apply the data-assignment function as described before.
By definition of $\mathit{c0nv}$,
the original constraint $C_i \subset \tilde C_i$ for every $i \geq 0$,
so the resulting valuation satisfies the original constraints as well.
This concludes the proof of Lemma~\ref{lem:data-assign-func}.\qed

\section{Conclusion}

Our main result states that one-sided Church games for
specifications given as \emph{deterministic} register automata over
$(\mathbb{N}, \leq)$ are decidable, in \textsc{ExpTime}. Moreover, we
show that those games are determined, and that strategies implemented by
transducers with registers suffice to win.

The decidability result involves a characterisation of satisfiable
infinite constraint sequences over $(\mathbb{N},\leq)$:
they must not have decreasing two-way chains of infinite depth, nor ceiled (bounded from the above) chains of unbounded depth.
A similar characterisation can be established for $(\mathbb{Z},\leq)$.
For instance, it should require that the two-way chains
which are bounded from both above and below have bounded depth.
Then, the decidability of one-sided Church synthesis for $(\mathbb{Z},\leq)$
can be established in a similar way to $(\bbN,\leq)$.
The decidability for $(\bbZ, \leq)$ can also be proven by reducing
to the problem for $(\bbN,\leq)$ as follows.
From a specification $S$, given as
a set of words $\d_1\sigma_1\d_2\sigma_2\dots$ alternating between a
value $\d_i\in\mathbb{Z}$ and a letter $\sigma_i$ from a finite
alphabet $\Sigma$, we construct a specification $S'$ of words of the form
$\text{max}(0,\d_1)\#\text{max}(0,-\d_1)\sigma_1\text{max}(0,\d_2)\#\text{max}(0,-\d_2)\sigma_2\dots
\in (\mathbb{N}(\Sigma\cup\{\#\}))^\omega$, where $\#$ acts as a
waiting symbol. Non-zero values given by Adam at positions $4n+1$
correspond to positive values, and non-zero values at positions $4n+3$ correspond
to negative values. Thus, if $S$ is given as a deterministic register automaton, one can construct a deterministic register automaton that recognises $S'$, which preserves the existence of solutions to
synthesis. An interesting future direction is to
establish a general reduction between data domains such that
decidability results for one-sided Church synthesis transfer from
one domain to the other.
A candidate notion for such a reduction was defined in the context of register-bounded transducer synthesis~\cite{DBLP:conf/icalp/ExibardF022}.

Another important future direction is to consider logical formalisms
instead of automata to describe specifications in a more
declarative and high-level manner.
Data word first-order logics~\cite{BMSSD06,DBLP:journals/corr/abs-1110-1439} have been studied with respect to the
satisfiability problem but when used as specification languages for synthesis, only few results are known.
The first steps in this direction were done in~\cite{DBLP:journals/lmcs/FigueiraMP20,BP22} for Constraint LTL on $(\bbZ,\leq)$; see also \cite{DQ23} for an overview of nonemptiness of constraint tree automata;
and see~\cite{DBLP:conf/fossacs/BerardBLS20} for a slightly different context of parameterised synthesis.

\bibliographystyle{plain}
\bibliography{refs}% common bib file

\end{document}

%%% Local Variables:
%%% mode: latex
%%% TeX-master: "main"
%%% End: